\documentclass[11pt]{article}
\usepackage{amsmath}
\usepackage{amssymb}
\usepackage{amsthm}
  \newtheorem{assumption}{Assumption}
  
  \newtheorem{definition}{Definition}
  
  \newtheorem{lemma}{Lemma}
  \newtheorem{proposition}{Proposition}
  \newtheorem{theorem}{Theorem}
  \newtheorem{property}{Property}
  \newtheorem{remark}{Remark}
\usepackage{caption}
\usepackage{fancyhdr}
\usepackage{float}
\usepackage[margin=1in]{geometry}
\usepackage{graphicx}  %
\usepackage{hyperref}
  \hypersetup{
    citecolor=blue,
    colorlinks=true,
    linkcolor=blue,
    filecolor=magenta,
    urlcolor=blue
  }
\usepackage{cleveref}  %
\usepackage{setspace}
\usepackage[table,xcdraw]{xcolor}
\usepackage{tikz}
  \usetikzlibrary{positioning}
  \usetikzlibrary{arrows.meta}
  \usetikzlibrary{arrows}
\usetikzlibrary{calc}
\usetikzlibrary{fit}
\usepackage[round]{natbib}
\usepackage{booktabs}
\usepackage{multirow}

\DeclareMathOperator*{\argmin}{arg\,min}
\newcommand{\E}{\ensuremath{\mathbb{E}}}
\newcommand{\I}{\ensuremath{\mathbb{I}}}
\newcommand{\1}{\ensuremath{{\bf 1}}}
\newcommand{\0}{\ensuremath{{\bf 0}}}

\newcommand{\cindep}{\raisebox{0.05em}{ 
  \!\rotatebox[origin=c]{90}{$\models$} \!
}}

\newcommand\ci{\protect\mathpalette{\protect\independenT}{\perp}}
\def\independenT#1#2{\mathrel{\rlap{$#1#2$}\mkern2mu{#1#2}}}

\newcommand\mci{\protect\mathpalette{\protect\mindependenT}{\perp}}
\def\mindependenT#1#2{\mathrel{\rlap{$#1#2$}\mkern2mu{#1#2}_m}}

\title{Dynamic Local Average Treatment Effects}
\author{Ravi B. Sojitra\\
Management Science and Engineering\\
Stanford University\\
\tt{ravisoji@stanford.edu}
\and
Vasilis Syrgkanis\\
Management Science and Engineering\\
Stanford University\\
\tt{vsyrgk@stanford.edu}
}
\date{
  First draft: May 2, 2024 \\
  This draft: \today
}

\begin{document}
\maketitle

\begin{abstract}
We consider Dynamic Treatment Regimes (DTRs) with One Sided Noncompliance that arise in applications such as digital recommendations and adaptive medical trials.
These are settings where decision makers encourage individuals to take treatments over time, but adapt encouragements based on previous encouragements, treatments, states, and outcomes.
Importantly, individuals may not comply with encouragements based on unobserved confounders.
For settings with binary treatments and encouragements, we provide nonparametric identification, estimation, and inference for Dynamic Local Average Treatment Effects (LATEs), which are expected values of multiple time period treatment effect contrasts for the respective complier subpopulations.
Under One Sided Noncompliance and sequential extensions of the assumptions in \cite{imbens1994identification}, we show that one can identify Dynamic LATEs that correspond to treating at single time steps.
In Staggered Adoption settings, we show that the assumptions are sufficient to identify Dynamic LATEs for treating in multiple time periods.
Moreover, this result extends to any setting where the effect of a treatment in one period is uncorrelated with the compliance event in a subsequent period.
{\let\thefootnote\relax\footnote{{
Ravi B. Sojitra was partially supported by NSF Award IIS-2337916.
Vasilis Syrgkanis was supported by NSF Award IIS-2337916 and partially by a 2022 Amazon Research Award.
We also thank the following colleagues and communities for helpful comments: 
Michael Baiocchi,
Lea Bottmer,
Adel Daoud,
Sukhjin Han,
Guido Imbens,
Ramesh Johari,
Apoorva Lal,
Lihua Lei,
Aidan Perreault,
Joel Persson,
Jann Spiess,
Anna Thomas,
Stefan Wager,
Yiqing Xu,
EC24 Workshop on Frontiers of Online Advertising,
Stanford Causal Machine Learning Seminar,
Stanford Metrics Lunch,
Stanford Data Science,
Stanford Causal Science Center.
}}}
\end{abstract}

\section{Introduction}

A Dynamic Treatment Regime (DTR) such as personalized digital recommendation or medical treatment is one in which each individual receives treatments over time that are personalized based on intermediate treatment responses (e.g. \citep{murphy2003optimal,robins1986new}).
Decision makers in these settings would like to assess the average effects of switching to different treatment sequences, but personalization reduces the amount of randomness available in treatments to inform this assessment.
Moreover, even when treatment is assigned randomly, there is often noncompliance where individuals do not use treatments assigned to them.
We focus on settings with One Sided Noncompliance where access to one of two interventions is controlled so that in every time period only individuals encouraged (i.e. assigned by the experimenter) in that time period may use the intervention.

Randomized experiments with noncompliance have a long history in the causal inference literature (see e.g. \citep[Chapters 23-24]{imbens2015causal}).
Moreover, experiments with One Sided Noncomliance date back to Zelen's Randomized Content Designs \citep{Zelen1979,zelen1990randomized} and later under the name Randomized Encouragement Designs \citep{powers1984effects,holland1988causal,bloom1984accounting}.
In such designs, randomized encouragements can be thought as an instrument, creating exogenous variation in the chosen treatment, and therefore instrumental variable techniques developed traditionally in the econometrics and the social sciences are applicable \citep{Wright1921CorrelationAndCausation,wright1934method,haavelmo1944probability}.
\cite{heckman1990varieties} showed that the average treatment effect (ATE) cannot be recovered via instrumental variable methods unless one is willing to make strong assumptions on the structural functions.
Seminal work by \cite{imbens1994identification} and \cite{angrist1996identification} formalized the causal estimand that can be recovered by applying the instrumental variable method and, in particular, showed that the average treatment effect of the complier population, known as the local average treatment effect (LATE), can be recovered by the standard instrumental variable methods. 

With the increased adoption of adaptive trials and micro-randomized trials \citep{klasnja2015microrandomized} in a variety of fields, such as medicine \citep{bhatt2016adaptive,berry2012adaptive,pallmann2018adaptive}, digital experimentation \citep{agarwal2016multiworld}, and digital marketing \citep{kalyanam2007adaptive}, developing a similar understanding of adaptive trials with noncompliance is increasingly relevant. The goal of our work is to provide nonparametric identification and estimation results for dynamic analogues of LATEs, in adaptive trials with one sided noncompliance, essentially extending the work of \cite{angrist1996identification} to the dynamic treatment regime and characterizing the identifiability of several types of dynamic analogues of LATEs. In fact, we examine the more general setting of adaptive \emph{stratified} trials with one sided noncompliance, where the binary encouragement offered to each unit at each period, as well as the chosen treatment by the unit, can depend on all previous observed short-term outcomes, states, encouragements and chosen treatments, as well as initial unit characteristics. The latter setting can also capture observational data scenarios with noncompliance, as long as the factors that went into the determination of the encouragement at each period are observed.

Adaptive stratified trials with one sided noncompliance arise in many application domains.
For example, a clothing brand will target discounts to individuals over time at various points in their purchase journeys.
In this setting, discounts are the encouragements, redemptions of discounts are the treatments, purchase journey positions are states, and outcomes are long term revenue.
Typically, only a small subset of individuals will redeem discounts at every point discounts are offered, which means that targeting discounts in an additional time period has a small average impact on sales per targeted individual.
However, it is natural to ask whether the average impact is large within the subpopulation that will redeem the discount in the additional time period.
If so, targeting for an additional time period would be compelling as it could improve brand loyalty and gross profits among the customers who do find the discounts compelling.

In medical treatment settings such as depression or HIV treatment, a healthcare provider may offer supplementary resources such as meditation classes or digital reminders, but patients may not always make use of these resources.
In this setting, free access to resources are the encouragements, utilization of the free resources are treatments, and the outcomes are depression symptoms or viral loads.
In this case, there is typically a cost associated with offering such resources, so if the average effect on symptoms is small because the effect is small for everyone, it may be more worthwhile to use the budget to make the baseline treatment more affordable or accessible instead.
If the average is small because there is a large effect for a small subset of the population that does utilize the offered resources, it could be worthwhile to offer this, since a large change in mood or viral load can have a large impact on mortality within one year.

The combination of adaptive trials, or more generally, dynamic treatment regimes, with instrumental variable methods, to address the noncompliance problem has received considerable attention in the recent literature \citep{han2023optimal,chen2023estimating, han2021identification, heckman2007dynamic,heckman2016dynamic,miquel2002identification,michael2023instrumental,cui2023instrumental}.
However, prior works do not offer a characterization of the identifiability of dynamic notions of LATEs without further restrictions on the independence of the compliance behavior with the heterogeneous dynamic treatment effect or resort to only providing partial identification results.
In particular, \citep{han2023optimal,chen2023estimating} consider only partial identification of optimal dynamic regimes, under noncompliance.
While \citealt{han2021identification,heckman2007dynamic,heckman2016dynamic} impose assumptions (analogous to \citep{vytlacil2007dummy}), which are strong enough to enable the identification of the dynamic analogue of the ATE, as opposed to the LATE.
Similarly, \cite{michael2023instrumental,cui2023instrumental} employ independence assumptions between the compliance effect and the unobserved confounder, that are also strong enough to identify ATEs, rather than LATEs.
Finally, \cite{miquel2002identification}, considers LATE estimands similar to the ones we do, but imposes assumptions that preclude the encouragement to be administered based on a dynamic adaptive policy, based on historical treatments and states.
Hence, it precludes the study of adaptive trials with noncompliance.

We provide nonparametric identification, estimation, and inference for Dynamic Local Average Treatment Effect (Dynamic LATE) estimands.
These estimands are average differences in outcomes between pairs of treatment sequences for individuals whose treatments would match the encouragements perfectly in both scenarios. We consider the case when the encouragement and the treatment at each period is binary.
We allow encouragements, treatments, and states in each time period to depend on the entire history of these three sets of variables.
Moreover, we allow for arbitrary unobserved confounding between treatments, states, and the final outcome.

We show that in this general setting, with an arbitrary number of periods, adding  One Sided noncompliance to standard identifying assumptions for DTRs enables identification of all dynamic When-to-Treat LATEs, i.e. LATEs corresponding to treatment in one time period only.
We also show that these types of When-to-Treat LATEs are not identifiable when one solely assumes a sequential variant of the \emph{monotonicity} assumption (i.e., that encouragement can only make someone take the treatment and not the opposite). Moreover, we show that in the absence of further restrictions, even in two-period settings with one sided noncompliance, the dynamic LATE that corresponds to treatment being offered in both periods, namely the Always-Treat LATE, cannot be identified. 

To complement the latter result, we also offer a minimal cross-period effect-compliance independence (``Staggered Compliance'') assumption, under which Always-Treat LATEs are also identifiable in the two-period setting. 
This assumption can be satisfied in various ways. For example, it is satisfied in Staggered Compliance settings where individuals who opt into treatment cannot revert their decision and will continue to opt into treatment in the future, whenever encouraged again to do so. It also holds in settings where there is endogeneity in only the first time period due to individuals opting into a sequential experiment and units are perfectly compliant with future encouragements if they choose to opt into the adaptive trial.
Under such Staggered Compliance settings, our identification argument for the Always-Treat LATE, also extends to the many period setting. Our two-period result requires a more relaxed assumption than Staggered Compliance.

We complement our identification results with debiased machine learning estimation and inference methods, applying the recent works on debiased \citep{chernozhukov2018double} and automatically debiased machine learning \citep{chernozhukov2022automatic} to our setting. We empirically validate the performance of our estimates and the coverage of our confidence interval construction procedure on synthetic data.

\section{Related Literature}

While there is extensive work on identification in DTRs (e.g. \cite{murphy2003optimal,robins1986new}) and in static settings with noncompliance (e.g. \cite{imbens1994identification}), their intersection is substantially less developed. 
Moreover, work within this intersection either relies on restrictive identifying assumptions to obtain general results or obtains limited results using weak assumptions.
Our approach builds on standard identifying assumptions for the DTR setting by introducing familiar assumptions such as One Sided Noncompliance.
To situate our work, we organize the related literature based on (1) whether it focuses on the single or multiple time period setting, (2) whether identification allows for encouragements to be adaptive (i.e. dynamic) or not (i.e. static), and (3) whether identification focuses on settings with noncompliance.

Recall that the treatment effect estimand is typically the Average Treatment Effect (ATE) $\E[Y(1) - Y(0)]$.
However, in settings with noncompliance, estimands are commonly Local Average Treatment Effects (LATE), which are ATEs conditional on compliance events representing subpopulations who will use treatments $D$ given encouragements $Z=z$.
For example $\E[Y(1) - Y(0) \mid D(1)-D(0)=1]$ in \cite{imbens1994identification}.
In this latter setting with noncompliance, the causal effect of the encouragement is called the Intent to Treat ATE: $\E[Y(D(1)) - Y(D(0))]$.

For the static, single time period setting with noncompliance, many works have extended LATE identification (\cite{imbens1994identification}) which gave a result for the setting with for a single categorical encouragement $Z_1$ and binary treatment variable $D_1$.
For example, \cite{mogstad2021causal,mogstad2024policy} use a weaker Partial Monotonicity condition, which requires binary treatments under a vector of encouragements to be greater than or equal to the treatments under another if the former vector of encouragements differs from the latter vector of encouragements in no more than one coordinate.
The condition allows the ordering of the values in the differing coordinate to change depending on the values of the remaining coordinates.
The authors show that under Partial Monotonicity, the Two Stage Least Squares (2SLS) estimand is equal to a non-negatively weighted combination of (static) Local Average Treatment Effects on the same treatment.
\cite{goff2020vector} introduced a more restrictive, version of this assumption called ``Vector Monotonicity'', which requires an ordering within each coordinate of encouragement vectors.
\cite{heckman2018unordered} focused on the static setting with multiple encouragement and treatment variables.
The identifying condition they introduce is called ``Unordered Monotonicity'', which requires that switching between two encouragements should shift everyone towards a particular treatment, which is quite restrictive even in the static setting.
\cite{bhuller20222sls} focus on 2SLS using average conditional monotonicity and ``no cross effects''.
The main limitations in applying this category of results to DTRs is that encouragements, treatments, and states in DTRs depend on previous encouragements, treatments, and states.

For the static, multiple time period setting with noncompliance, LATE identification can be achieved using stronger restrictions to circumvent the fact that the encouragement is held fixed over time.
For example, \cite{ferman2023identifying} assume that a binary encouragement is fixed, but there may be variation in when individuals first choose treatment.
They introduce two important identifying conditions.
The first is Staggered Adoption, which requires that those who enter the treatment group cannot exit the group over time.
The second requires the subpopulations contaminating the reduced form estimands to have constant or homogeneous treatment effects over time, which is restrictive for our motivating examples.
Related to this setting are outcome or duration models with instrumental variables (e.g. \cite{beyhum2023instrumental}), which also restrict heterogeneity of effects to achieve identification.
In particular, \cite{beyhum2023instrumental} impose rank invariance of the potential outcomes.

For the dynamic, multiple time period setting without noncompliance, identification of Dynamic Intent-to-Treat ATEs can be argued by requiring sequential extensions of encouragement overlap (i.e. positivity) and conditional ignorability (i.e. unconfoundedness) conditions (e.g. see \cite{murphy2003optimal,robins1986new}).
Informally, sequential ignorability requires that treatment in each time period is conditionally independent of future potential outcomes conditional on previous treatments and short term outcomes.
Sequential overlap requires that the treatments being contrasted within the same time period have positive probabilities.
We also require the commonly used no interference condition (e.g. see \cite{neyman1935statistical,rubin1990comment,imbens2015causal}).

\subsection{Related work on Dynamic Treatment Regimes with Noncompliance}

For the dynamic, multiple time period setting with noncompliance, two areas of work include identification of optimal treatment regimes (\cite{han2023optimal,chen2023estimating}) and of identification of treatment effects.
Work on identification of treatment effect estimands using instrumental variables has focused on both Average Treatment Effects (ATEs) (\cite{han2021identification,heckman2007dynamic,heckman2016dynamic,michael2023instrumental,cui2023instrumental} and LATEs (\cite{miquel2002identification}).
However, most existing work on causal effect identification imposes independence between states $S$ and unobservable random variables such as treatment-outcome confounders $U$ and individual causal effects of encouragements on treatments $C_t \triangleq D_t(1) - D_t(0)$.
These approaches are restrictive for Dynamic LATE identification in the settings we consider because we expect these unobserved random variables (e.g. $U, C_t$)
to depend on states, which are random variables containing static covariates, time varying covariates, and short term outcomes.

\cite{han2023optimal}, \cite{spicker2024optimal}, and \cite{chen2023estimating} take steps towards identification of optimal treatment regimes.
In particular, \cite{han2023optimal} provides \textit{partial} identification of the optimal treatment regime in settings without sequential ignorability.
That is, he identifies a set of treatment regimes that contains the optimal one given observational data, which may not satisfy sequential ignorability (i.e. exogeneity or unconfoundedness).
\cite{chen2023estimating} focus on a more modest policy learning task whose objective is to identify a dynamic treatment regime better than the baseline though not necessarily the optimal one.

\cite{han2021identification,heckman2007dynamic,heckman2016dynamic} focus on ATE identification.
\cite{han2021identification} provides nonparametric identification of the ATE using encouragements by extending to multiple time periods assumptions A-2 and A-5 in \cite{vytlacil2007dummy}, which provided ATE identification in the single time period setting with noncompliance.
In particular, \cite{han2021identification}'s Assumption SX extends A-2, which requires error terms in the outcome and treatment models to be independent of covariates $X$ and encouragements $Z$.
Assumption SP extends Assumption A-5 in \cite{vytlacil2007dummy}, which requires that for every covariate value $X = x$ with positive support there must exist another covariate value with support, and equal likelihoods of encouragement, and equal outcome under a \textit{different} treatment.
This assumption is restrictive, but enables ATE identification in the setting with instrumental variables and unobserved confounding.
\cite{heckman2016dynamic} use an analogous assumption (Assumption A-1e) and also impose restrictions on unobservables in the outcome and treatment models in that the error terms in the outcomes and treatments must be independent (Assumption A1-c) and the encouragements $Z$ must be independent of prognostic factors (Assumption A1-d).
They focus on identification in settings where the treatments are irreversible and consider the case where treatments are ordinal as in \cite{heckman2007dynamic} or where they are unordered, but embedded in a tree data structure.

\cite{miquel2002identification} consider the two-period case and showed that one can identify what we define as when-to-treat and always-treat LATEs (see \citealp[Theorem~7]{miquel2002identification}). However, they impose very strong assumptions on the dependence between encouragements and past observations.
In particular, \citep[Assumption~11]{miquel2002identification} precludes the second period encouragement from depending on the first period treatment, the intermediate state and furthermore, the second period treatment, can only depend on the second period encouragement. Hence, the only inter-period dependence that is allowed is that the second period encouragement can depend on the first period encouragement.
This setting, does not allow for the encouragements to be administered based on a dynamic treatment regime, so these results cannot capture adaptive stratified trials with noncompliance, which is the main goal of our work.

An alternative approach to ATE identification in DTRs is to impose restrictions on unobserved random variables that impact treatments and outcomes.
\cite{michael2023instrumental} assume that the causal effect of encouragements on treatments 
is independent of the unobserved confounder $U$: $\E[D_t(1) - D_t(0) \mid U_t, H_t]=\E[D_t(1) - D_t(0) \mid H_t]$, where $H_t \triangleq Z_{<t}, D_{<t}, S_{<t}$ is the history of encouragements, treatments, and states and $D_t(z)$ are the potential outcomes for the different values of the encouragement vector $z$. Note that under this assumption, in the static case, the ATE and not just the LATE can be identified.
They identify dynamic average treatment effects, or equivalently mean counterfactual outcomes under multi-period treatment interventions, under a marginal structural model \citep{robins1986new}, solely using this minimal compliance effect-confounder independence, without further structural assumptions. Moreover, the work of \citep{cui2023instrumental} extends this approach to dynamic settings with survival outcomes and censoring, under a Cox model. However, the goal of our work is to extend the static LATE literature and hence we strive to avoid making assumptions that in the static case, would be strong enough to identify the ATE and thereby we allow for the unobserved confounder to be correlated with the compliance effect. Even our result for always-treat LATEs, solely makes an assumption (see Assumption~\ref{assume:mean_independence} of our work) on cross-period compliance and outcome effect independence properties and not same-period independence properties.

Our approach uses two credible assumptions that have been motivated by and used in practice.
In particular, One Sided Noncompliance is invoked in both static settings (e.g. \cite{frolich2013identification}) and in multiple time period settings (e.g. 
\cite{bijwaard2005correcting}).
This is an interpretable restriction that the analyst or decisionmaker can have confidence about in settings where access to interventions are tightly controlled.
The second condition we introduce for our secondary result on identification of Always-Treat LATEs, is a cross-period compliance-effect conditional mean independence that is satisfied in various settings, including the Staggered Compliance (a special case of which is Staggered Adoption; related to settings studied in \cite{heckman2007dynamic,heckman2016dynamic,ferman2023identifying}),  
where once someone is encouraged and treated, they continue to be treated in all subsequent time periods, whenever they are encouraged to do so. For instance, the work of \citep{ferman2023identifying} can be viewed as always encouraging treatment in the second period and the units perfectly comply with this encouragement.
This assumption is also satisfied when noncompliance occurs only when deciding whether to enroll in an adaptive trial, but once enrolled, the unit deterministically complies with all subsequent recommendations.
Our mean independence assumption is more permissive than either of these settings.

\section{Notation}\label{sec:notation}
All capital letters represent random scalars, vectors, or matrices except for $T$, which is a constant representing the total number of time steps.
We have one observed scalar outcome $Y$ during time period $T$.
The remaining observed random variables represent encouragements $Z$, treatments $D$, and states $S$, which are sequences of $T$ random variables.
In particular, during each time period $t \in [T]$, the encouragement $Z_t$ takes values in $\mathcal{Z}_t = \{0,1\}$, the treatment $D_t$ takes values in $\mathcal{D}_t = \{0,1\}$, and the state generated by the previous time period $S_{t-1}$ takes values in $\mathcal{S}_{t-1} = \mathbb{R}^p$, where $p$ is the dimension of the state space representing observed variables such as static covariates $X$, time varying covariates $X_t$, and short term outcomes $Y_t$.
In general, calligraphic font (e.g. $\mathcal{Y}$) indicates the outcome space of random variables (e.g. $Y$) whose realizations are represented by lowercase letters (e.g. $y$).
$Z_{<t}, Z_{t:t'}, Z_{>t}$ denote sequences of the first $t-1$, the last $T-t$, and the middle $t'-t-1$ encouragements, respectively.
Subcomponents of all vectors are denoted analogously.
$\1, \0$ represent $T$-dimensional realizations of ones and zeroes, respectively.
$\prec, \succ, \preceq, \succeq$ represent entrywise inequalities.

\section{Two Time Period Setting}

\begin{figure}[t]
\centering
\begin{tikzpicture}[
  -{Latex[length=3mm,width=2mm]}, thick, node distance=.75cm,
  thick,
  every node/.style={scale=.77},
  unobserved/.style={
    circle,
      draw=black!60,
      fill=white!60,
      very thick,
      minimum size=10mm,
      font=\sffamily\large\bfseries
  },
  observed/.style={
    circle,
    draw=black!60,
    fill=gray!25,
    very thick,
    minimum size=10mm,
    font=\sffamily\large\bfseries
  },
  placeholder/.style={
    circle,
    draw=white!60,
    fill=white!25,
    very thick,
    minimum size=10mm,
    font=\sffamily\large\bfseries
  },
  box/.style = {draw,very thick,black,inner sep=10pt,rounded corners=5pt, inner sep=1cm}] 
]

\node[observed]     (Z1)                      {$Z_{1}$};
\node[placeholder]  (P1)   [right=of Z1]      {};
\node[placeholder]  (P2)   [below=of P1]      {};
\node[observed]     (Z2)   [right=1.5 of P1]  {$Z_2$};

\node[observed]     (D1)   [below=of Z1]      {$D_1$};
\node[observed]     (D2)   [below=of Z2]      {$D_2$};

\node[observed]     (S1)   [below=of P2]      {$S_1$};

\node[placeholder]  (P0)   [left=of S1]       {};
\node[unobserved]   (C1)   [left=of P0]       {$U$};
\node[observed]     (S0)   [above=of C1]       {$S_0$};
\node[placeholder]  (P4a)  [above=.7 of S0]   {};
\node[placeholder]  (P4b)  [above=1.4 of S0]  {};
\node[placeholder]  (P4c)  [above=2 of S0]    {};

\node[placeholder]  (P3)   [right=of D2]      {};
\node[observed]     (Y)    [below=of P3]      {$Y$};

\node[box,rotate fit=0,fit=(Z1)(P1)(P2)(Z2)(S1)(Y)] (plate1) {};

\path[every node/.style={font=\sffamily\small}]
    (S0) edge[solid]node[] {} (plate1)
    (Z1) edge[solid] node[] {} (D1)
    (Z1) edge[solid] node[] {} (Z2)
    (Z1) edge[solid] node[] {} (D2)

    (Z2) edge[solid] node[] {} (D2)

    (S1)  edge[solid] node[] {} (Y)
    (S1)  edge[solid] node[] {} (Z2)
    (S1)  edge[solid] node[] {} (D2)

    (D1) edge[solid] node[] {} (Z2)
    (D1) edge[solid] node[] {} (D2)
    (D1) edge[solid] node[] {} (S1)
    (D1) edge[solid] node[] {} (Y)
    
    (D2) edge[solid] node[] {} (Y)

;

\path[every path/.style={font=\sffamily\small,draw=orange}]
    (C1) edge[bend left=15]   node  [left] {} (D1)
    (C1) edge[bend right=25]  node  [left] {} (S1)
    (C1) edge[bend left=15]   node  [left] {} (D2)
    (C1) edge[bend left=25]   node  [left] {} (S0)
    (C1) edge[bend right=25]  node  [left] {} (Y)
    
;
\end{tikzpicture}
~~
\begin{tikzpicture}[
  -{Latex[length=3mm,width=2mm]}, thick, node distance=.75cm,
  thick,
  every node/.style={scale=.765},
  unobserved/.style={
    circle,
      draw=black!60,
      fill=white!60,
      very thick,
      minimum size=10mm,
      font=\sffamily\large\bfseries
  },
  observed/.style={
    circle,
    draw=black!60,
    fill=gray!25,
    very thick,
    minimum size=10mm,
    font=\sffamily\large\bfseries
  },
  placeholder/.style={
    circle,
    draw=white!60,
    fill=white!25,
    very thick,
    minimum size=10mm,
    font=\sffamily\large\bfseries
  },
  box/.style = {draw,very thick,black,inner sep=10pt,rounded corners=5pt, inner sep=1cm}] 
]

\node[observed]     (Z1)                      {$Z_{1}$};
\node[placeholder]  (P1)   [right=of Z1]      {};
\node[placeholder]  (P2)   [below=of P1]      {};
\node[observed]     (Z2)   [right=1.5 of P1]  {$Z_2$};

\node[observed]     (D1)   [below=of Z1]      {$D_1$};
\node[observed]     (D2)   [below=of Z2]      {$D_2$};

\node[observed]     (S1)   [below=of P2]      {$S_1$};

\node[placeholder]  (P0)   [left=of S1]       {};
\node[unobserved]   (C2)   [below left=1.2 of S1]       {$U_1$};
\node[unobserved]   (C3)   [below right=1.2 of S1]       {$U_2$};
\node[unobserved]   (C1)   [left=1.2 of C2]       {$U_0$};
\node[placeholder]     (P5)   [above=2.2 of C1]       {};
\node[observed]     (S0)   [right=-.3 of P5]       {$S_0$};
\node[placeholder]  (P4a)  [above=.7 of S0]   {};
\node[placeholder]  (P4b)  [above=1.4 of S0]  {};
\node[placeholder]  (P4c)  [above=2 of S0]    {};

\node[placeholder]  (P3)   [right=of D2]      {};
\node[observed]     (Y)    [below=of P3]      {$Y$};

\node[box,rotate fit=0,fit=(Z1)(P1)(P2)(Z2)(S1)(Y)] (plate1) {};

\path[every node/.style={font=\sffamily\small}]
    (S0) edge[solid]node[] {} (plate1)
    (Z1) edge[solid] node[] {} (D1)
    (Z1) edge[solid] node[] {} (Z2)
    (Z1) edge[solid] node[] {} (D2)

    (Z2) edge[solid] node[] {} (D2)

    (S1)  edge[solid] node[] {} (Y)
    (S1)  edge[solid] node[] {} (Z2)
    (S1)  edge[solid] node[] {} (D2)

    (D1) edge[solid] node[] {} (Z2)
    (D1) edge[solid] node[] {} (D2)
    (D1) edge[solid] node[] {} (S1)
    (D1) edge[solid] node[] {} (Y)
    
    (D2) edge[solid] node[] {} (Y)

;

\path[every path/.style={font=\sffamily\small,draw=orange}]
    (C1) edge[bend left=25]   node  [left] {} (S0)
    (C1) edge[]   node  [left] {} (C2)
    (S0) edge[bend right=15]   node  [left] {} (C2)
    
    (C2) edge[bend left=15]   node  [left] {} (D1)
    (C2) edge[]  node  [left] {} (S1)
    (C2) edge[]  node  [left] {} (C3)
    (S1) edge[]  node  [left] {} (C3)

    (C3) edge[bend left=15]   node  [left] {} (D2)
    (C3) edge[]  node  [left] {} (Y)    
;
\end{tikzpicture}
\caption{
  Directed Acyclic Graphs that adhere to the exclusion restrictions in Assumption \ref{assume:2_consistency}.
  Encouragements $Z$, treatments $D$, states $S$, and long term outcomes $Y$ are represented by gray circles to indicate that they are observed variables.
  $U$'s are unobserved confounders, and are illustrated using a white circle to indicate that they are unobserved.
  Edges (arrows) pointing from one random variable to another indicates that the latter is allowed to be a function of the former.
  The node $S_0$ should be thought of as sending edges to all observed random variables.
}\label{fig:dag_1}
\end{figure}

In the two time period setting, we observe sequences of random variables $(S_0, Z_1, D_1, S_1, Z_2, D_2, Y)$ for each unit.
We assume that these variables are sampled independently across units.
$S_0$ corresponds to a pre-experiment state of the unit.
At each time period $t>0$, an encouragement $Z_t$ is first generated based on prior encouragements, treatments, and states.
Next, a treatment $D_t$ is generated as a function of the encouragement $Z_t$ and all previously observed random variables, as well as potentially unobserved confounding factors or processes.
Finally, states $S_t$ and long term outcome $Y$ are generated as a function of $D_t$ and all prior observed and unobserved random variables, \textit{but not directly as a function of prior encouragements} $Z_{\leq t}$.
States $S_0, S_1$ are sets of observed random variables that each may contain short-term outcomes, time-varying covariates, covariates prior to the experiment.
As discussed in Section \ref{sec:notation}, we write the realized states as $s=(s_0, s_1) \in \mathcal{S}$, encouragements as $z=(z_1, z_2), z'=(z_1', z_2') \in \mathcal{Z}$, treatments as $d=(d_1, d_2) \in \mathcal{D}$, and long term outcomes as $y \in \mathcal{Y}$.
Figure \ref{fig:dag_1} illustrates examples of Data Generating Processes (DGPs) that adhere to this description.

The typical approach in structural causal modeling and nonparametric structural equation literature is to incorporate the exact structure of the unobserved confounding process in the definition of the primitive counterfactual or potential outcomes. We will develop identification results that are agnostic to the unobserved confounding structure (e.g. they would apply either to both DGPs in Figure~\ref{fig:dag_1}, as well as potentially other confounding structures). To enable this, we will define the primitive counterfactual processes only as a function of observed variables, respecting the exclusion restrictions that we outlined in the previous paragraph. We will then present a set of conditional independence assumptions that these processes need to satisfy for our identification results. These assumptions would hold for instance in either of the DGPs in Figure~\ref{fig:dag_1}.

Our Assumption~\ref{assume:2_consistency} formalizes this discussion and imposes restrictions that extend the Consistency properties and the Exclusion Restriction in \citep{imbens1994identification} to the sequential setting.
\begin{assumption}[Primitive Counterfactual Processes, Consistency and Exclusion Restrictions]\label{assume:scm}\label{assume:2_consistency} We assume the existence of a set of primitive counterfactual random processes that satisfy the exclusion restrictions encoded in the graphs of Figure~\ref{fig:dag_1}. Namely, 
\begin{align}\label{eqn:counterfactuals}
\left\{S_0(.), Z_1(s_0), D_1(z_1, s_0), S_1(d_1, s_0), Z_2(z_1, d_1, s), D_2(z, d_1, s), Y(d, s)\mid s\in \mathcal{S}, z\in\mathcal{Z}, d\in\mathcal{D}\right\}.
\end{align}
The observed variables $(S_0, Z_1, D_1, S_1, Z_2, D_2, Y)$ are generated by first drawing the counterfactual random processes from some fixed distribution, independently across units, and then recursively evaluating them: 
\begin{align*}
S_0 :=~& S_0(\cdot)   \\
Z_1 :=~& Z_{1}(S_0)           \\
D_1 :=~& D_{1}(Z_1, S_0)   \\
S_1 :=~& S_{1}(D_1, S_0)   \\
Z_2 :=~& Z_{2}(Z_1, D_1, S)   \\
D_2 :=~& D_{2}(Z, D_1, S)  \\
Y :=~& Y(D, S).
\end{align*}
\end{assumption}

\noindent Since the counterfactual random variables are identically and independently distributed across units, we implicitly assume the no interference assumption (i.e. no network effects), which is also referred to as SUTVA (Stable Unit Treatment Value Assumption).
Given Assumption~\ref{assume:2_consistency}, we can also define intervention counterfactuals as follows. 
\begin{definition}[Intervention Counterfactuals]\label{defn:int-cnt}
For any endogenous variable $V$ and subset of endogenous variables $X$ taking values in the domain ${\cal X}$, we denote with $V(X\to x)$ the counterfactual random outcome of $V$ when we intervene and fix the variables in $X$ to the value $x\in {\cal X}$.
This variable is produced by generating the endogenous variables of the structural causal model recursively, but any time we encounter any variable in $X$ on the right-hand side of the structural equation, we fix its value to the corresponding value in $x\in {\cal X}$, instead of its recursive evaluation (see e.g. \cite[Section~7.3]{chernozhukov2024applied} or \cite{richardson2013single} for more details on such intervention counterfactuals, known as \emph{fix} interventions).
\end{definition}
\noindent For clarity, in Appendix~\ref{app:examples} we provide examples of such intervention counterfactuals that will be used throughout our analysis. Note that, by definition, 
the intervention counterfactuals also satisfy the following consistency property.

\begin{property}[Consistency of Counterfactual Variables] For any subset of intervening variables $X$ and any observed variable $V$, we have the consistency property $V=V(X)$.
\end{property}

\noindent
Throughout the paper we will utilize the following shorthand notation for the two specific intervention counterfactuals that will be central to our analysis:
\begin{definition}[Shorthand Notation of Central Intervention Counterfactuals]\label{defn:central-cnt}
We will use the shorthand notation $Y(d)$ to denote the counterfactual outcome generated by intervening on the treatment variables $D$ and fixing them to value $d\in \mathcal{D}$ and the shorthand notation $D(z)$ to denote the counterfactual treatment vector generated by intervening on the instrument variables $Z$ and fixing them to value $z\in {\cal Z}$, i.e., 
\begin{align}
Y(d)\equiv~& Y(D\to d) & 
D(z)\equiv~& D(Z\to z).
\end{align}
\end{definition}
\noindent We can explicitly write these counterfactuals using the primitive counterfactual processes in Assumption \ref{assume:2_consistency}:
\begin{align*}
    Y(d) =~& Y(d, S_0, S_1(d_1, S_0))
\end{align*}
and similarly for $D(z)$:
\begin{align*}
    D_1(z) =~& D_1(z_1, S_0) & 
    D_2(z) =~& D_2(z, D_1(z), S_0, S_1(D_1(z), S_0)).
\end{align*}

To complete the definition of our setting, we need to impose some assumptions on the distribution of the counterfactual random processes defined in Equation~\eqref{eqn:counterfactuals}.
The typical approach in the nonparametric structural equation model literature is to explicitly incorporate the dependence on the unobserved confounding factors. For instance, the left graph in Figure~\ref{fig:dag_1} explicitly incorporates the variable $U$ to permit dependencies between counterfactual variables defined in Assumption~\ref{assume:scm} for each fixed $s,d,z$. We can then define the augmented counterfactual processes
\begin{align}
    \left\{U(.), S_0(u), Z_1(s_0), D_1(z_1, s_0, u), S_1(d_1, s_0, u), Z_2(z_1, d_1, s), D_2(z, d_1, s, u), Y(d, s, u)\right\}, \label{eqn:aug_counterfactuals}
\end{align}
with $s\in {\cal S}, z\in {\cal Z}, d\in {\cal D}, u\in {\cal U}$. 
Then, assume that these counterfactuals stem from a structural causal model with independent errors \citep{pearl1995causal}. That is, for every variable $V$ with parent variable values $pa_V$:
\begin{align}
    V(pa_V) = f_V(pa_V, \epsilon_V)
\end{align}
where the error variables $\epsilon_V$ are drawn from an exogenous fixed distribution and are mutually independent across the variables.
Alternatively, the more relaxed FFRCISTG approach \citep{robins1986new}, solely imposes that for every fixed $s, z, d, u$, the variables in Equation~\eqref{eqn:aug_counterfactuals}
are mutually independent. The latter only imposes independence statements that are all testable via some hypothetical experiment, known as single-world intervention properties (see e.g. \cite{richardson2013single} for a discussion). 

We will not impose either of these but instead define the set of minimal properties that we require for our counterfactual processes, as defined in Assumption~\ref{assume:scm}, for our identification argument. Our identification result will be applicable to any structural causal model that implies these properties. For instance, these properties are satisfied for the structural causal model associated with the graph in Figure~\ref{fig:dag_1} under both the independent errors or the FFRCISTG assumption.

\noindent 
Our first assumption states that the counterfactual processes satisfy the following sequential notion of conditional ignorability:
\begin{assumption}[Sequential Ignorability of Instruments]\label{assume:2_ignorability}
For every $z\in {\cal Z}$:
\begin{align*}
    \{Y(D(z)), D(z)\} \cindep& Z_1 \mid S_0 & 
    \{Y(D_1, D_2(Z_1, z_2)), D_2(Z_1, z_2)\} \cindep& Z_2 \mid S, D_1, Z_1.
\end{align*}
\end{assumption}
\noindent 
Informally, this means that the instrument is conditionally independent of future potential outcomes and treatments, that would have been generated under instrument interventions, given the observed history of instruments, treatments, and states.
This assumption is the sequential extension of the standard ignorability (i.e. unconfoundedness) condition on the instruments. As we show in Appendix~\ref{app:ignorability}, this assumption is implied by the structural causal model depicted in Figure~\ref{fig:dag_1}, even under the relaxed FFRCISTG model.

Apart from the structural causal model assumption, we will also need to the following extra assumptions that are extensions of typical assumptions in the static setting. The sequential overlap assumption ensures that all instrument variable values are probable conditional on past observations and the relevance assumption ensures that the instruments have an effect on the treatment conditional on the history of past observations.
\begin{assumption}[Sequential Overlap]\label{assume:2_overlap}
$\Pr\{Z_1 = z_1 \mid S_0\} > 0$ and $\Pr\{Z_2=z_2 \mid Z_1, D_1, S_0, S_1\} > 0$, almost surely, for all $z\in {\cal Z}$.
\end{assumption}

\begin{assumption}[Sequential Relevance]\label{assume:2_relevance}
$\mathbb{E}[D_1(1) - D_1(0) \mid S_0] > 0$ and $\mathbb{E}[D_2(Z_1, 1) - D_2(Z_1, 0) \mid Z_1, D_1, S] > 0$, almost surely.
\end{assumption}

The final assumption that we will be make is is the analogue of the one sided noncompliance property of the static setting.
Informally, this assumption requires for every time period that individuals can choose to be treated only if they were encouraged to do so, but otherwise the treatment is not available to them.
\begin{assumption}[One Sided Noncompliance]\label{assume:2_one_sided}
$\Pr\{D_1(z_1) \leq z_1, D_2(z) \leq z_2\}=1$ for all $z\in {\cal Z}$.
\end{assumption}

\subsection{Estimands}

Our goal is to identify several types of Dynamic Local Average Treatment Effects (LATEs).
These estimands quantify the effect of treatment sequences $d$ on an outcome $Y$ relative to never treating for different subpopulations.
Since the compliance event $D(0,0)=(0,0)$ happens almost surely under One Sided Noncompliance, we may omit it from conditioning events characterizing the subpopulations.

One set of estimands of interest are effects of treatment sequences for subpopulations who comply with encouragement sequences. 
Formally, we define $\theta(z,d)$ to be the average treatment effect of treatment vector $d$ among the subpopulation who chooses treatments $d$ when given encouragements $z$.

\begin{definition}[Dynamic LATE]\label{def:dyn_late}
For every $z,d \in \mathcal{Z} \times \mathcal{D}$ we define the Dynamic LATE as the following conditional expectation (and its short-hand notation when $z=d$): 
\begin{align}
    \theta(z,d) =~& \E[Y(d) - Y(0,0)\mid D(z)=d] &
    \tau_d     =~& \theta(d,d).
\end{align}
\end{definition}

\noindent As part of our analysis, we will also be highlighting a quantity that corresponds to a mixture of dynamic LATEs. This quantity will be useful as a stepping stone for other LATE quantities and could potentially be of interest in its own write. This Dynamic Mixture LATE quantity corresponds to a mixture of several Dynamic LATEs for several treatment intervention vectors. In particular, it is the average effect of the treatment interventions chosen by the sub-population of units that choose to receive treatment at some period, i.e. comply with at least one encouragement. 
\noindent

\begin{definition}[Dynamic Mixture LATE]
For every $z \in \cal{Z}$ such that $z \neq (0,0)$, we define 
\begin{align}
    \beta_z = \E\left[Y(D(z)) - Y(0,0) \mid D(z) \neq (0,0)\right]
\end{align}
\end{definition}

\noindent
 In addition to the unconditional LATEs, we also aim to identify heterogeneous Dynamic LATEs. We define heterogeneous effects $\theta(z,d, S_0), \beta_z(S_0)$ analogous to the estimands above. Informally, we are interested in the effects defined above conditional on pre-experiment characteristics $S_0$ of each unit.

\begin{definition}[Heterogeneous Dynamic LATE]\label{def:hetero_dyn_late}
For every $z,d \in \mathcal{Z} \times \mathcal{D}$ we define the Heterogeneous Dynamic LATE as the following conditional expectation (and its short-hand notation when $z=d$): 
\begin{align}
    \theta(z,d, S_0) =~& \E[Y(d) - Y(0,0)\mid D(z)=d, S_0], &
    \tau_d(S_0)     =~& \theta(d,d, S_0)
\end{align}
\end{definition}

\begin{definition}[Heterogeneous Dynamic Mixture LATE]
\begin{align}
    \beta_z(S_0) = \E\left[Y(D(z)) - Y(0,0) \mid D(z) \neq (0,0), S_0\right]
\end{align}
\end{definition}

\section{Identification of Dynamic Local Average Treatment Effects}

We are now ready to state our first main result on the identification of a subset of the Dynamic LATEs in the two period case. In particular, we show that Dynamic LATEs $\tau_d$ with $d\in \{(0,1), (1,0)\}$ and Dynamic Mixture LATEs for any $z\in {\cal Z}^2$ and their heterogeneous counterparts are identifiable. Thus, we can identify the effect of a treatment policy that treats only at either period one or period zero. This can be viewed as a \emph{When-to-Treat LATE} as we are asking counterfactual questions where we treat only once and the only things that vary are the period of treatment and subpopulation. Moreover, we are comparing these counterfactual outcomes to the baseline outcome under no-treatment. Each of these treatment policy contrasts is measured on the sub-population that when encouraged to take the treatment only at the corresponding period, will comply with the recommendation.

Crucially, this main result excludes the Dynamic LATE with $z=d=(1,1)$, which we will refer to as \emph{Always-Treat LATE}. This LATE is inherently not identifiable solely on the basis of the assumptions we have made thus far. Section \ref{sec:tightness} proves this impossibility and Section \ref{sec:always-treat} introduces a condition to enable its identification. However, we note that our main result in this section also provides identification for the dynamic LATE mixture $\beta_z$, which contains always-treat contrasts as part of the mixture. This mixture can be identified without further assumptions.

\begin{theorem}[Identification of Dynamic LATEs]\label{thm:2_identification}
Assume
 \ref{assume:2_consistency}, 
 \ref{assume:2_ignorability}, \ref{assume:2_overlap}, \ref{assume:2_relevance}, \ref{assume:2_one_sided}.
The following equalities hold.

\noindent For $z=d\in\{(0,1), (1,0)\}$, 
\begin{align*}
\tau_d
 &=\frac{
    \mathbb{E}[Y(D(z)) ]
    - \mathbb{E}[Y(D(0,0)) ]
  }{\Pr\{ D(z)=d \}
  }.
\end{align*}
For all $z \in \mathcal{Z}$ such that $z \neq(0,0)$, 
\begin{align*}
\beta_z
 &=\frac{
    \mathbb{E}[ Y(D(z)) ]
    - \mathbb{E}[ Y(D(0,0)) ]
  }{1 - \Pr\{ D(z)=(0,0) \}
  }
\end{align*}
For every $d,z \in \mathcal{D} \times \mathcal{Z}$ such that $d \preceq z$, the counterfactual averages above are identified:
\begin{align}
\E[Y(D(z))]
=~& \E[\E[\E[Y \mid S, D_1, Z=z] \mid S_0, Z_1=z_1]]\\
\Pr(D(z)=d)
=~& \E[\E[\Pr(D=d\mid S, D_1, Z=z)\mid S_0, Z_1=z_1]].\label{eqn:id_d_cnt}
\end{align}
\end{theorem}

\noindent The same identification results extend to Heterogeneous Dynamic LATEs, i.e. Dynamic LATEs as a function of the initial state $S_0$ of each unit (i.e. unit characteristics at initiation).

\begin{theorem}[Identification of Heterogeneous Dynamic When-to-Treat LATEs]\label{thm:2_identification_hetero}
Assume
 \ref{assume:2_consistency}, 
 \ref{assume:2_ignorability}, \ref{assume:2_overlap}, \ref{assume:2_relevance}, \ref{assume:2_one_sided}.

\noindent
For $z=d\in\{(0,1), (1,0)\}$,
\begin{align*}
\tau_d(S_0)
 =\frac{
    \mathbb{E}[Y(D(z)) \mid S_0]
    - \mathbb{E}[Y(D(0,0)) \mid S_0]
  }{\Pr\{ D(z)=d \mid S_0\}
  }.
\end{align*}
For all $z \in \mathcal{Z}$ such that $z \neq(0,0)$, 
\begin{align*}
\beta_z(S_0)
 =\frac{
    \mathbb{E}[ Y(D(z)) \mid S_0]
    - \mathbb{E}[ Y(D(0,0)) \mid S_0]
  }{1 - \Pr\{ D(z)=(0,0) \mid S_0\}
  }.
\end{align*}
For every $d,z \in \mathcal{D} \times \mathcal{Z}$ such that $d \preceq z$, the counterfactual averages above are identified: 
\begin{align}
\E[Y(D(z))\mid S_0]
=~&\E[\E[Y \mid S, D_1, Z=z] \mid S_0, Z_1=z_1]\\
\Pr(D(z)=d\mid S_0)
=~& \E[\Pr(D=d\mid S, D_1, Z=z)\mid S_0, Z_1=z_1].
\end{align}
\end{theorem}

\paragraph{Towards Always-Treat LATE}
The mixture LATE quantity $\beta_z$ does not correspond to a dynamic LATE, since it entails a mixture of outcome contrasts for various treatment vectors.
However, as we show next it corresponds to a natural mixture of Dynamic LATEs and the weights in this mixture are identifiable from the data. In particular, we show that $\beta_{11}$ is the convex combination of Dynamic LATEs for different subpopulations, where the weights are the probabilities of the subpopulations conditional on \emph{not} being a Never Taker under $Z=z$.

\begin{proposition}\label{prop:mixture}
For every $z \in \mathcal{Z}$ such that $z \neq(0,0)$, 
\begin{align*}
\beta_z
=~&\sum_{d \preceq z:~d\neq (0,0)} \theta(z,d) \cdot w(z,d) &
w(z,d) :=~&  \frac{\Pr(D(z)=d)}{\sum\limits_{d' \preceq z:~d'\neq (0,0)}\Pr(D(z)=d')}, 
\end{align*}
where the probabilities in mixture weights $w(z,d)$ are identified by Equation~\eqref{eqn:id_d_cnt}.
\end{proposition}
\noindent
Notice that when $z\in \{(0,1), (1,0)\}$, the set $\{d \preceq z:d\neq (0,0)\}$ is a singleton that contains only $d=z$, so $\theta(z,d)=\theta(d,d)=\tau_d$ and $w(z,d)=1$, so $\beta_z=\tau_d$.
The estimand $\beta_{1,1}$ is a mixture of these treatment contrasts for appropriately defined compliance populations, i.e., for some identifiable mixture weights $w_{11}, w_{10}, w_{01}$ that lie in the simplex:
\begin{align}
    \beta_{1,1} =~& \E[Y(1,1)-Y(0,0)\mid D(1,1)=(1,1)]\, w_{11}\\
    ~& + \E[Y(1,0)-Y(0,0)\mid D(1,1)=(1,0)]\, w_{10} \\
    ~& + \E[Y(0,1)-Y(0,0)\mid D(1,1)=(0,1)]\, w_{01}
\end{align}
with the weights being identifiable from the data. This identifiable mixture quantity could potentially also be policy relevant, since it contains elements of the $Y(1,1)-Y(0,0)$ contrast. If for instance, we do not expect large effect heterogeneities of the When-to-Treat contrasts among the different complier sub-populations (i.e. $D(1,0)=(1,0)$ vs. $D(1,1)=(1,0)$), then we would expect the second and third terms in the above summand to be of the same order as the identifiable quantities from the main theorem, i.e. for $z=d\in \{(0,1),(1,0)\}$
\begin{align}
    \E[Y(d) - Y(0,0) \mid D(1,1)=d] \approx \E[Y(d) - Y(0,0) \mid D(z)=d] = \tau_d
\end{align}
Thus if we observe that $\beta_{1,1}$ is substantially larger than $\theta(z,d)$ for $z=d\in \{(0,1),(1,0)\}$, this can be attributed to the first term in the mixture. Thereby showcasing that the Always-Treat policy can have a much larger effect on a subset of the population and therefore revealing that there are substantial complementarity effects among the different treatment exposures. 

Moreover, note that if we make the further (strong) restriction that:
\begin{equation}\label{eqn:strong-ass}
\begin{aligned}
    \E[Y(1,0) - Y(0,0) \mid D(1,1)=(1,0)] =~& \E[Y(1,0) - Y(0,0) \mid D(1,0)=(1,0)] =: \tau_{10}\\
    \E[Y(0,1) - Y(0,0) \mid D(1,1)=(0,1)] =~& \E[Y(1,0) - Y(0,0) \mid D(0,1)=(0,1)] =: \tau_{01}
\end{aligned}
\end{equation}
then the Always-Treat LATE can be identified as:
\begin{align}\label{eqn:easy-always-treat}
    \tau_{11} = \E[Y(1,1)-Y(0,0)\mid D(1,1)=(1,1)] = \frac{\beta_{11} - \tau_{10}\,w_{10} - \tau_{01}\,w_{01}}{w_{11}}
\end{align}
In Section~\ref{sec:always-treat} we will provide a much more relaxed assumption under which the Always-Treat LATE is also identifiable, while in the next section, we show that Always-Treat LATEs are not identifiable without further restrictions.

\section{Tightness of \Cref{thm:2_identification}}\label{sec:tightness}

A natural question is whether the identification result presented in Theorem~\ref{thm:2_identification} could be proven without  One Sided Noncompliance.
For example, one might wonder whether the sequential version of Monotonicity \cite{imbens1994identification} would suffice. Another natural question is whether Theorem~\ref{thm:2_identification} could be extended to show identification of the Always-Treat LATE under  One Sided Noncompliance. We offer a formal negative answer to both questions.

In contrast to One Sided Noncompliance (Assumption \ref{assume:2_one_sided}), Sequential Monotonicity (Assumption \ref{assume:2_sequential_monotonicity}) is insufficient for identification of all When-to-Treat Dynamic LATEs.
Moreover, One Sided Noncompliance (Assumption \ref{assume:2_one_sided}) is insufficient for identification of Always-Treat Dynamic LATEs.
To demonstrate each limitation, we construct two example data generating processes (DGPs) that have opposite signed underlying Dynamic LATE parameters yet have the same observed data distribution $\Pr\{D, Z, Y\}$ and satisfy identifying assumptions.
Taken together, One Sided Noncompliance improves upon a Sequential Monotoncity by enabling identification of all When-to-Treat LATEs, but it does not enable identification of Always-Treat LATEs, which motivates the second identifying condition we introduce in Section \ref{sec:always-treat}.

\begin{table}[]
\centering
\caption{
  $\mathbb{E}[Y(1,0) - Y(0,0) \mid D(1,0)-D(0,0)=(1,0)]$ is not identified under (Sequential) Monotonicity.
}\label{tbl:non_ided_under_monotonicity}
\begin{tabular}{ccccccccccccc}
\rowcolor[HTML]{C0C0C0}
  $\Pr$
  & $\mathbf{Z}$
  & $\mathbf{D(0,0)}$
  & $\mathbf{D(0,1),D(1,0),D(1,1)}$
  & $\mathbf{D}$
  & $\mathbf{Y(D(\zeta))}$
  & $\mathbf{Y(1,0)}$
  & $\mathbf{Y’(1,0)}$
  & $\mathbf{Y}$
  & $\mathbf{Y’}$
  & $\mathbf{\tau}$
  & $\mathbf{\tau'}$ \\
1/8 & 0,0 & (0,0) & (0,1),(1,0),(1,1) & 0,0 & -2
    &  2 & -2 & -2 & -2 & 4 & 0     \\
1/8 & 0,0 & (0,1) & (0,1),(1,0),(1,1) & 0,1 & -2
    & -2 &  2 & -2 & -2 & 0 & 4     \\
1/8 & 1,0 & (0,0) & (0,1),(1,0),(1,1) & 1,0 & -2
    &  2 & -2
    &\cellcolor[HTML]{FFCC67} 2
    &\cellcolor[HTML]{FFCC67}-2
    & 4 & 0     \\
1/8 & 1,0 & (0,1) & (0,1),(1,0),(1,1) & 1,0 & -2
    & -2 &  2
    &\cellcolor[HTML]{FFCC67}-2
    &\cellcolor[HTML]{FFCC67} 2
    & 0 & 4 \\
1/8 & 0,1 & (0,0) & (0,1),(1,0),(1,1) & 0,1 & -2
    &  2 & -2 & -2 & -2 & 4 & 0     \\
1/8 & 0,1 & (0,1) & (0,1),(1,0),(1,1) & 0,1 & -2
    & -2 &  2 & -2 & -2 & 0 & 4     \\
1/8 & 1,1 & (0,0) & (0,1),(1,0),(1,1) & 1,1 & -2
    &  2 & -2 & -2 & -2 & 4 & 0     \\
1/8 & 1,1 & (0,1) & (0,1),(1,0),(1,1) & 1,1 & -2
    & -2 &  2 & -2 & -2 & 0 & 4     \\
\end{tabular}
\newline
\newline
\centering
\caption{
  $\mathbb{E}[Y(1,1)-Y(0,0) \mid D(1,1)-D(0,0) = (1,1)]$ is not identified under  One Sided Noncompliance.
}\label{tbl:non_ided_under_one_sided_noncompliance}
\begin{tabular}{cccccccccccc}
\rowcolor[HTML]{C0C0C0}
  $\Pr$
  & $\mathbf{Z}$
  & $\mathbf{D(1,1)}$
  & $\mathbf{D(0,0),D(1,0),D(0,1)}$
  & $\mathbf{D}$
  & $\mathbf{Y(D(\zeta))}$
  & $\mathbf{Y(0,0)}$
  & $\mathbf{Y’(0,0)}$
  & $\mathbf{Y}$
  & $\mathbf{Y’}$
  & $\mathbf{\tau}$ 
  & $\mathbf{\tau'}$ \\
1/8 & 0,0 & (1,0) & (0,0),(1,0),(0,1) & 0,0 & 2 &  2 & -2
    & \cellcolor[HTML]{FFCC67} 2
    & \cellcolor[HTML]{FFCC67}-2
    & 0 & 4                                            \\
1/8 & 0,0 & (1,1) & (0,0),(1,0),(0,1) & 0,0 & 2 & -2 &  2
    & \cellcolor[HTML]{FFCC67}-2
    & \cellcolor[HTML]{FFCC67} 2
    & 4 & 0                                            \\
1/8 & 1,0 & (1,0) & (0,0),(1,0),(0,1) & 1,0 & 2
    &  2 & -2 & 2 & 2 & 0 & 4     \\
1/8 & 1,0 & (1,1) & (0,0),(1,0),(0,1) & 1,0 & 2
    & -2 &  2 & 2 & 2 & 4 & 0     \\
1/8 & 0,1 & (1,0) & (0,0),(1,0),(0,1) & 0,1 & 2
    &  2 & -2 & 2 & 2 & 0 & 4     \\
1/8 & 0,1 & (1,1) & (0,0),(1,0),(0,1) & 0,1 & 2
    & -2 &  2 & 2 & 2 & 4 & 0     \\
1/8 & 1,1 & (1,0) & (0,0),(1,0),(0,1) & 1,0 & 2 &  2 & -2 & 2 & 2 & 0 & 4  \\
1/8 & 1,1 & (1,1) & (0,0),(1,0),(0,1) & 1,1 & 2 & -2 &  2 & 2 & 2 & 4 & 0  \\
\end{tabular}
\newline
\newline
\caption*{
  Tables \ref{tbl:non_ided_under_monotonicity} and \ref{tbl:non_ided_under_one_sided_noncompliance} tabulate joint distributions of two DGPs each.
  Within each table, every row represents a distinct pair of encouragement vector $z \in \{0,1\}^2$ and subpopulation $\{ D(z) : z \in \{0,1\}^2 \}$ that has non-zero probability in both Data Generating Processes.
  Column $\Pr$ tabulates the joint probabilities, which are equal in both DGPs, of encouragement-subpopulation pairs.
  Column $Z$ tabulates encouragement vectors with non-zero probability in both DGPs.
  Columns $D(0,0),D(0,1),D(1,0),D(1,1)$ tabulates potential treatments under encouragements $D(0,0),D(0,1),D(1,0),D(1,1)$, which are equal in both DGPs.
  This latter set of columns characterizes the distinct unobserved subpopulations.

  The remaining values are either deterministic within both DGPs or deterministically indexed given encouragements and subpopulations.
  Column $D$ reports the observed treatment sequences given the preceeding encouragement vectors and subpopulations.
  Columns $Y(\cdot), Y'(\cdot)$ indicate potential outcomes under two different data Generating Processes where as $Y(D(\zeta))$ indicates the potential outcomes that are the same between DGPs.
  Columns $Y, Y'$ tablulate observed outcomes under two different Data Generating Processes with the presence and absence of an apostrophe indicating the DGP.
  Columns $\tau, \tau'$ analogously tabulate effects for different subpopulations: When-to-Treat for Table \ref{tbl:non_ided_under_monotonicity} and Always-Treat for Table \ref{tbl:non_ided_under_one_sided_noncompliance}.
}
\end{table}

\subsection{Nonidentifiability of When-to-Treat LATE under Sequential Monotonicity}

Without One Sided Noncompliance, individuals who comply under $D(1,0)$ may not comply in both time periods under $D(0,0)$.
This means that the following equality which holds under One Sided Noncompliance does not necessarily hold under Sequential Monotonicity: 
 $\mathbb{E}[Y(1,0) - Y(0,0) \mid D(1,0)-D(0,0)=(1,0)] = \mathbb{E}[Y(1,0) - Y(0,0) \mid D(1,0)=(1,0)]$.
In this subsection, we show that the left hand side of this inequality is not identified if we replace One Sided Noncompliance with Sequential Monotonicity within Theorem \ref{thm:2_identification}.
To be concrete, the following Assumption extends Montonicity in  \cite{imbens1994identification} to the sequential setting.

\begin{assumption}[Sequential Monotonicity]\label{assume:2_sequential_monotonicity}
\begin{align*}
D_1(1) &\geq D_1(0) \\ 
D_2(Z_1, 1) &\geq D_2(Z_1, 0).
\end{align*}
\end{assumption}

Table \ref{tbl:non_ided_under_monotonicity} illustrates two DGPs to show that Sequential Monotonicity (Assumption \ref{assume:2_sequential_monotonicity}) is not sufficient for identification.
Encouragement sequences $z \in \{0,1\}^2$ each have probability .25 so that Consistency, Exclusion, Ignorability, and Overlap (Assumptions \ref{assume:2_consistency}, \ref{assume:2_ignorability}, and \ref{assume:2_overlap}) hold. 
We simplify the example so that there are two subpopulations almost surely by setting $\Pr\{D(0,0)=d\}=.5$ for $d \in \{(0,0), (0,1)\}$ and $\Pr\{D(z)=z\}=1$ for $z \neq (0,0)$. Taken together, each encouragement vector and subpopulation pair has probability 1/8. Note that we can already conclude Sequential Relevance and Sequential Monotonicity are Satisfied. Next, we set potential outcomes $Y(D(\zeta))=-2$ for all $\zeta \in \{(0,0), (0,1), (1,1)\}$ so that we construct the counterexample by setting the potential outcomes for $Y(1,0), Y’(1,0)$ only. To construct the counterexample, we then set $Y(1,0)$ to have opposite sites between the two subpopulation within DGPs and flip signs between DGPs. This results in $\tau=4$ in one DGP and $\tau=0$ in another. Yet, in both cases, the observed distribution of $Z, D, Y$ is the same: $\Pr\{Z=z, D=d\}=1/8$ in both DGPs, $\Pr\{Y=2 \mid Z\neq(0,0), D\neq(0,0)\}=1$, and $\Pr\{Y=y \mid Z=(0,0), D=(0,0)\}=.5$ for $y \in \{-2, 2\}$.

\subsection{Non-Identifiability of Always-Treat LATE under  One Sided Noncompliance}

Table \ref{tbl:non_ided_under_one_sided_noncompliance} illustrates two DGPs to show that One Sided Noncompliance is not generally sufficient to enable identification of $\tau_{11}$.
Again, encouragement sequences $z \in \{0,1\}^2$ each have probability .25.
We set $\Pr\{D(1,1)=d\}=.5$ for $d \in \{(1,0), (1,1)\}$ and $\Pr\{D(z)=z\}=1$ for $z \neq (1,1)$ so that there are only two underlying subpopulations.
We also set $\Pr\{ Y(D(\zeta))=2 \}=1$ for $\zeta \neq (0,0)$ so that we can construct the counterexample by setting the potential outcomes for $Y(0,0), Y'(0,0)$ only.
In particular, within each DGP the two subpopulations have opposite signed $Y(0,0)$ potential outcomes that trade signs between DGPs.
This results in $\tau_{11}=4$ in one DGP and $\tau_{11}=0$ in another.
In both cases, the observed distribution of $Z, D, Y$ is the same and all of the Assumptions of Theorem \ref{thm:2_identification} hold: Exclusion, Consistency, Ignorability, Overlap, Relevance, One Sided Noncompliance hold in Table \ref{tbl:non_ided_under_one_sided_noncompliance} (Assumptions \ref{assume:2_consistency},  \ref{assume:2_ignorability}, \ref{assume:2_overlap}, \ref{assume:2_relevance}, and \ref{assume:2_one_sided}).
Concretely, every row in Table \ref{tbl:non_ided_under_one_sided_noncompliance} sets $\Pr\{Z=z, D=d\}=1/8$ in both DGPs, $\Pr\{Y=2 \mid Z\neq(0,0), D\neq(0,0)\}=1$, and $\Pr\{Y=y \mid Z=(0,0), D=(0,0)\}=.5$ for $y \in \{-2, 2\}$.

\section{Identifying Always-Treat LATE under Further Restrictions}\label{sec:always-treat}

Given the impossibility results in the previous section, we identify a further restriction that is natural and can be justified in certain domains, which enables identification of the Always-Treat LATE. We show that one sufficient condition for the identification of the Always-Treat LATE is that the effect of the first period treatment (in the absence of any subsequent treatments), i.e. $Y(1,0)-Y(0,0)$, is independent of the second period compliance, when a unit is recommended treatment at both periods, conditional on the initial state $S_0$ and conditional on having been recommended the treatment in the first period ($Z=1$) and having complied with that recommendation ($D=1$).

The quantity $Y(1,0)-Y(0,0)$, is many times referred to as the blip effect of the first period treatment \citep{robins2004optimal}. Hence, we assume that conditional on your initial state (e.g. initial characteristics of the unit), if you are a complier and took the treatment in the first period, then the blip effect of the first period treatment is independent whether you are going to comply in the second period, in the event that you are also recommended treatment in the second period. In particular, we will only require mean-independence. For this purpose, we define the following short-hand notation for conditional mean-independence:
\begin{align}
X \mci~& Y \mid Z \Leftrightarrow \E[X Y\mid Z] = \E[X\mid Z] \E[Y\mid Z].
\end{align}
and require that
\begin{align}
        Y(1,0) - Y(0,0) \mci D_2(1, 1) \mid Z_1=1, D_1=1, S_0 
\end{align}
Since $D_2(1,1)$ is a binary random variable, the latter is equivalent to a mean-zero covariance property:

\begin{assumption}[Cross-Period Effect-Compliance Independence]\label{assume:mean_independence}
Assume that:
\begin{align}
\mathrm{Cov}\left(Y(1,0) - Y(0,0), D_2(1, 1) \mid Z_1=1, D_1=1, S_0 \right) = 0
\end{align}
\end{assumption}

For instance, consider the case when, if a unit is encouraged to be treated in the first period and they accept the treatment, then they also opt for the treatment in all subsequent periods, whenever encouraged to do so.
We will call this setting as having \emph{Staggered Compliance}, which differs from Staggered Adoption because treatment can be switched off in the former (due to One Sided Noncompliance), but not in the latter where once treated, the individual remains treated.
Moreover, Staggered Adoption can be cast as a special case of Staggered Compliance, where in the former the instrument in future periods after the treatment has been administered is artificially set to $1$. Under Staggered Compliance, Assumption~\ref{assume:mean_independence} is trivially satisfied, since the event $D_2(1,1)=1$, conditional on $D_1=1, Z_1=1, S_0$ is a deterministic event, and thereby independent of any random quantity. Thus the identification result in this section directly applies to Staggered Compliance settings, i.e., when
\begin{align}
    \Pr(D_2(1,1) = 1\mid D_1=1, Z_1=1, S_0)=1.
\end{align}
Another setting where this independence is also trivially satisfied, is when noncompliance occurs only at the entry of the dynamic treatment regime. Units choose whether or not to enter some adaptive trial and once chosen to enter, they follow the recommendations of the trial. This is depicted in Figure~\ref{fig:dag_2}. We can rethink of this setting as the instrument in the second period being a perfect instrument and hence the treatment always being equal to the treatment (even under arbitrary interventions). Under this re-interpretation, we have again that $D_2(1,1)$ is deterministically equal to $1$, since the second period instrument is a perfect instrument. Thus for the setting in Figure~\ref{fig:dag_2}, we are also in a setting where:
\begin{align}
    D_2(z, d_1, s) = z_2, ~\quad\text{a.s.}\tag{endogenous entry into adaptive trial}
\end{align}
Our Assumption~\ref{assume:mean_independence} generalizes both of these scenarios and allows for other non-deterministic switches in compliance behaviors in the second period. As long as the reasons why a unit flipped its compliance behavior in the second period is not correlated with the blip effect of the first period treatment.

We will also need to invoke an extra benign ignorability condition that is also implied by the structural causal model associated with the causal graph depicted in Figure~\ref{fig:dag_1} under either the independent errors assumption \citep{pearl1995causal}, or the more relaxed FFRCISTG variant \citep{robins1986new} (see Appendix~\ref{app:ignorability}).

\begin{assumption}[Ignorability of Instruments with joint Instrument and Treatment Interventions]\label{assume:3_ignorability}
For every $d\in {\cal D}^2$ and $z\in {\cal Z}^2$:
\begin{align*}
    \{Y(d), D_1(z_1)\} \cindep& Z_1 \mid S_0
\end{align*}
\end{assumption}
\noindent Informally, this means that the instrument is conditionally independent of future potential outcomes and treatments, that would have been generated under instrument and treatment interventions, given the observed history of instruments, treatments, and states.

\begin{theorem}\label{thm:late11_id}
    Under Assumptions~\ref{assume:2_consistency}, 
 \ref{assume:2_ignorability}, \ref{assume:2_overlap}, \ref{assume:2_relevance}, \ref{assume:2_one_sided}, \ref{assume:mean_independence} and \ref{assume:3_ignorability}
 the Always-Treat LATE $\tau_{11}$ and the conditional Always-Treat LATE $\tau_{11}(S_0)$ are identified as follows.
    \begin{align*}
    \tau_{1,1}(S_0) :=~& \frac{\beta(S_0) - \tau_{10}(S_0) (1 - \gamma_{1,1}(S_0))}{\gamma_{1,1}(S_0)}
    \end{align*}
    where:
    \begin{align*}
        \beta(S_0) :=~&  \E[Y(D(1,1))-Y(D(0,0))\mid D_1(1)=1, S_0] \\
        \gamma_{1,1}(S_0) :=~& \Pr\{D_2(1,1)=1\mid D_1(1)=1, S_0\}
    \end{align*}
    In turn, letting $H_1=S_0$ and $H_2=(S, Z_1, D_1)$, these quantities are identified as:
    \begin{align*}
        \beta(S_0) =~& \frac{\Gamma}{\Pr(D_1=1\mid S_0, Z_1=1)}\\
        \Gamma :=~& \E[\E[Y\mid H_2, Z_2=D_1]\mid S_0, Z_1=1] - \E[\E[Y\mid H_2, Z_2=0]\mid S_0, Z_1=0]\\
        \gamma_{1,1}(S_0) =~& \E[\Pr(D_2=1\mid H_2, Z_2=1)\mid S_0, Z_1=1, D_1=1]
    \end{align*}
    while $\tau_{10}(S_0)$ is the heterogeneous when-to-treat LATE identified as in Theorem~\ref{thm:2_identification_hetero}. The unconditional Always-Treat LATE is identified as:
    \begin{align*}
        \tau_{11} = \frac{\E[\tau_{1,1}(S_0) \Pr(D(1,1)=(1,1)\mid S_0)]}{\Pr(D(1,1)=(1,1))}
    \end{align*}
    $\Pr(D(1,1)=(1,1))$ and $\Pr(D(1,1)=(1,1)\mid S_0)$ are identified in Theorem~\ref{thm:2_identification} and Theorem~\ref{thm:2_identification_hetero}, correspondingly. The latter can be simplified as:
    \begin{align*}
        \tau_{11} = \frac{\E[\Gamma - \tau_{10}(S_0)\, \E[\Pr(D=(1,0)\mid H_2, Z_2=1)\mid S_0, Z_1=1]]}{\E[\E[\Pr(D=(1,1)\mid H_2, Z_2=1]\mid S_0, Z_1=1]]}
    \end{align*}
\end{theorem}

For the case of Staggered Compliance, the identification formula for $\tau_{11}$ implied by our proof can also be written in the following form (see Section~\ref{sec:staggered} for the derivation):
\begin{align*}
    \frac{\E[Y(D(1,1))] - \E[Y(D(0,0))] - \E[\E[f_2(H_2, 1)-f_2(H_2,0)\mid Z_1=1, D_1=0, S_0]\Pr(D_1=0\mid Z_1=1, S_0)]}{\E[\Pr(D_1=1\mid Z_1=1, S_0)]} 
\end{align*}
where the quantities $\E[Y(D(1,1))] - \E[Y(D(0,0))]$ are identified as prescribed in Theorem~\ref{thm:2_identification} and $H_2=(S, Z_1, D_1)$ (i.e., the information available prior to the period $2$ encouragement) and $f_2(H_2,Z_2)$ is defined as:
\begin{align*}
    f_2(H_2, Z_2) = \E[Y\mid H_2, Z_2]
\end{align*}
Interpreting the formula for $\tau_{11}$, we see that without the negative correction term, the first term roughly coincides to the quantity:
\begin{align}
    \frac{\E[Y(D(1,1)) - Y(D(0,0))]}{\Pr(D_1(1)=1)}
\end{align}
This would have been the LATE, if there was only one compliance choice in the data and $D(1,1)=(1,1)$ if $D_1(1)=1$ and $D(1,1)=(0,0)$ otherwise. Note that in this case, due to the fact that in the event that $D_1=0$, the future treatments are deterministically $0$ and due to the exclusion restriction, the states and outcomes only depend on the instruments through the treatments, in this scenario, we would have that $f_2(H_2, 1) = f_2(H_2,0)$ a.s., conditional on $D_1=0$. Hence, the correction term would cancel. However, in the data that we hypothesize, units that did not comply in the first period can potentially comply in the second period. We only assume Staggered Compliance. In this case, future instruments can affect future treatments, which subsequently affect states and outcome, making $f_2(H_2, 1) \neq f_2(H_2, 0)$. In this case, the correction term, essentially removes from the vanilla term above, the bias in the vanilla effect calculation induced by these second period compliers. 

Moreover, we can further simplify the identification formula for the case of Staggered Compliance, to a variant that is much more amenable to estimation:
\begin{lemma}[Always-Treat LATE under Staggered Compliance]\label{lem:staggered} Under Assumptions~\ref{assume:2_consistency}, 
 \ref{assume:2_ignorability}, \ref{assume:2_overlap}, \ref{assume:2_relevance}, \ref{assume:2_one_sided},   \ref{assume:3_ignorability} and the Staggered Compliance Assumption, letting $H_1=S_0$ and $H_2=(S, Z_1, D_1)$, the Always-Treat LATE is identified as:
    \begin{align}\label{eqn:staggered-id}
        \tau_{11} :=~& \frac{\E[\E[\E[Y\mid H_2, Z_2=D_1]\mid H_1, Z_1=1] - \E[\E[Y\mid H_2, Z_2=0]\mid H_1, Z_1=0]]}{\E[\Pr(D_1=1\mid Z_1=1, H_1)]} 
    \end{align}
    and the conditional Always-Treat LATE is identified as:
    \begin{align}
        \tau_{1,1}(S_0) :=~& \frac{\E[\E[Y\mid H_2, Z_2=D_1]\mid H_1, Z_1=1] - \E[\E[Y\mid H_2, Z_2=0]\mid H_1, Z_1=0]}{\Pr(D_1=1\mid Z_1=1, H_1)} 
    \end{align}
\end{lemma}
\begin{figure}[H]
\centering
\begin{tikzpicture}[
  -{Latex[length=3mm,width=2mm]}, thick, node distance=.75cm,
  thick,
  every node/.style={scale=.8},
  unobserved/.style={
    circle,
      draw=black!60,
      fill=white!60,
      very thick,
      minimum size=10mm,
      font=\sffamily\large\bfseries
  },
  observed/.style={
    circle,
    draw=black!60,
    fill=gray!25,
    very thick,
    minimum size=10mm,
    font=\sffamily\large\bfseries
  },
  placeholder/.style={
    circle,
    draw=white!60,
    fill=white!25,
    very thick,
    minimum size=10mm,
    font=\sffamily\large\bfseries
  },
  box/.style = {draw,very thick,black,inner sep=10pt,rounded corners=5pt, inner sep=1cm}] 
]

\node[observed]     (Z1)                      {$Z_{1}$};
\node[placeholder]  (P1)   [right=of Z1]      {};
\node[placeholder]  (P2)   [below=of P1]      {};
\node[placeholder]     (Z2)   [right=1.5 of P1]  {};

\node[observed]     (D1)   [below=of Z1]      {$D_1$};
\node[observed]     (D2)   [below=of Z2]      {$D_2$};

\node[observed]     (S1)   [below=of P2]      {$S_1$};

\node[placeholder]  (P0)   [left=of S1]       {};
\node[unobserved]   (C1)   [left=of P0]       {$U$};
\node[observed]     (S0)   [above=of C1]       {$S_0$};
\node[placeholder]  (P4a)  [above=.7 of S0]   {};
\node[placeholder]  (P4b)  [above=1.4 of S0]  {};
\node[placeholder]  (P4c)  [above=2 of S0]    {};

\node[placeholder]  (P3)   [right=of D2]      {};
\node[observed]     (Y)    [below=of P3]      {$Y$};

\node[box,rotate fit=0,fit=(Z1)(P1)(P2)(Z2)(S1)(Y)] (plate1) {};

\path[every node/.style={font=\sffamily\small}]
    (S0) edge[solid]node[] {} (plate1)
    (Z1) edge[solid] node[] {} (D1)
    (Z1) edge[solid] node[] {} (D2)

    (S1)  edge[solid] node[] {} (Y)
    (S1)  edge[solid] node[] {} (D2)

    (D1) edge[solid] node[] {} (D2)
    (D1) edge[solid] node[] {} (S1)
    (D1) edge[solid] node[] {} (Y)
    
    (D2) edge[solid] node[] {} (Y)

;

\path[every path/.style={font=\sffamily\small,draw=orange}]
    (C1) edge[bend left=15]   node  [left] {} (D1)
    (C1) edge[bend right=25]  node  [left] {} (S1)
    (C1) edge[bend left=25]   node  [left] {} (S0)
    (C1) edge[bend right=25]  node  [left] {} (Y)
    
;
\end{tikzpicture}
\caption{Causal graph depicting a scenario with endogenous entry into an adaptive stratified trial. The variable $S_0$ should be thought as having outgoing edges to all the observed variables (gray).}
\label{fig:dag_2}
\end{figure}

\section{Automatic Debiased Estimation and Inference}\label{sec:dml}

In this section, we will develop an automatic debiased machine learning estimate for the When-to-Treat LATEs, for $d=z\in \{(0,1), (1,0)\}$, which were the key identified estimands of Theorem~\ref{thm:2_identification}, as well as the always-treat LATE under Staggered Compliance that we identified in Theorem~\ref{thm:late11_id}. This will lead to a straightforward asymptotically valid confidence interval construction procedure.

To state our main inference results, we first need to define a set of nuisance functions that arise from our identification theorem. For convenience, we will define the information or history sets 
\begin{align*}
    H_1=~& S_0, & H_2=~& (S, D_1, Z_1).
\end{align*}
We first define nested regression functions related to the dynamic effects of the instruments on the outcome:\footnote{Even though $f_2^z=f_2^\0$, we define them separately for notational convenience.}
\begin{equation}\label{eqn:outcome-nested}
\begin{aligned}
    f_2^z(H_2, Z_2) =~& \E[Y\mid H_2, Z_2] \quad\quad &\quad\quad f_2^{\0}(H_2, Z_2) =~& \E[Y\mid H_2, Z_2]\\
    f_1^z(H_1, Z_1) =~& \E[f_2^z(H_2, z_2)\mid H_1, Z_1] \quad\quad&\quad\quad
    f_1^\0(H_1, Z_1) =~& \E[f_2^\0(H_2, 0)\mid H_1, Z_1]
\end{aligned}
\end{equation}
and nested regression functions related to the dynamic effects of the instruments on the treatments:
\begin{equation}\label{eqn:treatment-nested}
\begin{aligned}
    g_2(H_2, Z_2) =~& \Pr(D=d\mid H_2, Z_2) = \E[\I(D=d)\mid H_2, Z_2]\\
    g_1(H_1, Z_1) =~& \E[g_2(H_2, z_2)\mid H_1, Z_1]
\end{aligned}
\end{equation}
These nested regression functions can be estimated via recursive empirical risk minimization and existing results provide mean-squared-error bounds for such recursive procedures as a function of measures of the statistical complexity of the function spaces used and their approximation error (see e.g. Lemma~10 and Theorem~11 of \cite{lewis2020double}). In practice, these functions can be estimated via recursively invoking generic machine learning regression oracles on appropriately defined outcome variables that are recursively constructed based on the model output by the prior regression oracle call.

\subsection{Debiased Estimation of When-to-Treat LATEs}\label{sec:dml-when}

Given these definitions, we can re-state our identification result in Theorem~\ref{thm:2_identification} as, for any $d=z\in \{(0,1), (1,0)\}$:
\begin{align*}
    \tau_d = \frac{\E\left[f_1^z(H_1, z_1) - f_1^\0(H_1, 0)\right]}{\E\left[g_1(H_1, z_1)\right]}
\end{align*}
Equivalently, $\tau_d$ can be viewed as the unique solution, with respect to $\tau_d$ of the moment equation:
\begin{align*}
    \E\left[f_1^z(H_1, 1) - f_1^\0(H_1, 0) - \tau_d \cdot g_1(H_1, z_1)\right] = 0 
\end{align*}

Thus the target parameter of interest $\tau_d$ can be viewed as the solution to a moment equation that is linear in the parameter of interest and which also depends linearly on a set of nuisance functions that are defined as the solution to nested regression problems. This exactly falls into the class of estimands that can be automatically debiased following the results of \cite{chernozhukov2022automatic}. In particular Thoerem~6 of \cite{chernozhukov2022automatic}, shows how to automatically construct a debiased moment that the target parameter also needs to satisfy, but which is also Neyman orthogonal with respect to all the nuisance functions \citep{chernozhukov2018double}. 

To define the automatic debiasing procedure we first need to introduce the Riesz representers of appropriately defined linear functionals. We introduce the function $a_1^z(H_1, Z_1)$ as the Riesz representer of the linear functional $L_1(q)=\E[q(H_1, z_1)]$ and the function $a_2^z(H_2, Z_2)$ as the Riesz representer of the linear functional $L_2(q)=\E[q(H_2, z_2) a_1^{z}(H_1, Z_1)]$. These can be estimated via empirical risk minimization of the loss functions:
\begin{equation}\label{eqn:riesz}
\begin{aligned}
    a_1^z =~& \argmin_{a\in A_1} \E\left[a(H_1, Z_1)^2 - 2a(H_1, z_1)\right]\\
    a_2^z =~& \argmin_{a\in A_2} \E\left[a(H_2, Z_2)^2 - 2 a_1^{z}(H_1, Z_1)\,a(H_2, z_2)\right]
\end{aligned}
\end{equation}
where $A_1, A_2$ are function spaces that can approximate well the true riesz function and can be taken to be sample dependent growing non-linear sieve spaces. \cite{chernozhukov2022automatic} provide mean-squared-error guarantees for this recursive estimation procedure. Alternatively, these Riesz representers can be characterized in closed form as the following propensity ratios:
\begin{align}\label{eqn:riesz-closed}
    a_1^z(H_1, Z_1) =~& \frac{1\{Z_1=z_1\}}{\Pr\{Z_1=z_1\mid H_1\}} &
    a_2^z(H_1, Z_1) =~& \frac{1\{Z_1=z_1\}}{\Pr\{Z_1=z_1\mid H_1\}} \cdot \frac{1\{Z_2=z_2\}}{\Pr\{Z_2=z_2\mid H_2\}} 
\end{align}
and estimated in a plug-in manner, by first estimating the propensities that appear in the denominators, via generic ML classification approaches. Similarly, we can define the Riesz representers $a_1^\0$ and $a_2^\0$ as the functions $a_1^z, a_2^z$ for $z=(0,0)$.

Letting $W=(S, Z, D, Y)$ and $f=\{f_1^z, f_2^z, f_1^\0, f_2^\0\}$ and $g=\{g_1, g_2\}$ and $a=\{a_1, a_2, a_1^\0, a_2^\0\}$, we can now write the debiased moment equation that the target parameter $\tau_d$ must also satisfy, as:
\begin{align}
    \E\left[{\phi}_z(W;f, a) - {\phi}_0(W;f, a) - \tau_d\cdot  {\psi}(W; g, a)\right] = 0
\end{align}
where
\begin{equation}\label{eqn:psi}
\begin{aligned}
    {\phi}_z(W; f,a) =~& {f}_1^z(H_1, z_1) + {a}_1^z(H_1, Z_1) ({f}_2^z(H_2, z_2) - {f}_1^z(H_1, Z_1)) + {a}_2^z(H_2, Z_2)(Y - {f}_2^z(H_2, Z_2))\\
    {\psi}(W; g, a) =~& {g}_1(H_1, z_1) + {a}_1(H_1, Z_1) ({g}_2(H_2, z_2) - {g}_1(H_1, Z_1)) + {a}_2(H_2, z_2)(\I(D=d) - {g}_2(H_2, Z_2))
\end{aligned}
\end{equation}
This moment equation is now Neyman orthogonal with respect to all the nuisance functions $f, g, a$ that it depends on. Thus we can invoke the general debiased machine learning estimation paradigm.

Given any ML estimates $\hat{f}, \hat{g}, \hat{a}$, constructed on a separate sample (or in a cross-fitting manner), we can construct the doubly robust estimate of $\tau_d$ as the solution to the empirical plug-in analogue of the Neyman orthogonal moment equation, with respect to $\hat{\tau}_d$, i.e.:
\begin{align*}
    \E_n\left[{\phi}_z(W;\hat{f}, \hat{a}) - {\phi}_0(W;\hat{f}, \hat{a}) - \hat{\tau}_d\cdot  {\psi}(W; \hat{g}, \hat{a})\right] = 0
\end{align*}
which takes the form:
\begin{align}\label{eqn:estimate}
    \hat{\tau}_d = \frac{\E_n\left[{\phi}_z(W;\hat{f}, \hat{a}) - {\phi}_0(W;\hat{f}, \hat{a})\right]}{\E_n[{\psi}(W; \hat{g}, \hat{a})]}
\end{align}
Based on an analysis identical to the one in Thoerem~9 of \cite{chernozhukov2022automatic}, we can state the following asymptotic normality and confidence interval construction statement. In the statement we will use the norm notation $\|h\|_2 = \sqrt{\E[h(X)^2]}$ for any function $h$ that takes as input a random variable $X$ and the norm notation $\|x\|_p$ as the $p$-th norm of a vector $x$. 
\begin{theorem}[Inference on When-to-Treat LATE]\label{thm:inference}
For any fixed $d=z\in \{(0,1), (1,0)\}$, assume that
\begin{enumerate}
 \item Sequential strong overlap holds, i.e. for $t\in \{1,2\}$: $\Pr\{Z_t=z_t\mid H_t\}, \Pr\{Z_t=0\mid H_t\}\geq \epsilon > 0$ for a constant $\epsilon$.
    \item All random variables and function estimates are uniformly bounded by a constant.\footnote{This assumption can be relaxed to $\|\hat{f}_2(H_2, z_2)\|_p, \|\hat{f}_1(H_1, z_1)\|_p\leq C$, $\|\hat{f}_2(H_2, Z_2)\|_p, \|\hat{f}_1(H_1, Z_1)\|_p \leq C$, and $\|\hat{a}_1(H_1, Z_1)\|_{\infty}, \|\hat{a}_2(H_2, Z_2)\|_{\infty}\leq C$  for some $p>2$ and $|\tau_d|<C$ and $\E[\left(f_2(H_2, z_2) - f_1(H_1, Z_1)\right)^2\mid H_1, Z_1],\E[\left(Y - f_2(H_2, Z_2)\right)^2\mid H_2, Z_2]\leq \bar{\sigma}^2$}
    \item All nuisance estimates are mean-square-consistent, i.e., for $t\in \{1, 2\}$
    \begin{align*}
        \|\hat{f}_t^z(H_t, Z_t) - f_t^z(H_t, Z_t)\|_2, \text{ and } \|\hat{g}_t(H_t, Z_t) - g_t(H_t, Z_t)\|_2 \text{ and } \|\hat{a}_t^z(H_t, Z_t) - a_t^z(H_t, Z_t)\|_2=~& o_p(1)\\
        \|\hat{f}_t^\0(H_t, Z_t) - f_t^\0(H_t, Z_t)\|_2, \text{ and } \|\hat{a}_t^\0(H_t, Z_t) - a_t^\0(H_t, Z_t)\|_2=~& o_p(1)
    \end{align*}
    \item The following product root-mean-squared-error nuisance rates are satisfied: for $t\in \{1, 2\}$
    \begin{align*}
    \sqrt{n}\|\hat{a}_t^z(H_t, Z_t) - a_t^z(H_t, Z_t)\|_2\cdot \max\left\{\|\hat{f}_t(H_t, Z_t) - f_t^z(H_t, Z_t)\|_2, \|\hat{g}_t(H_t, Z_t) - g_t(H_t, Z_t)\|_2\right\} =& o_p(1)\\
    \sqrt{n} \|\hat{a}_1^z(H_1, Z_1) - a_1^z(H_1, Z_1)\|_2\cdot \max\left\{\|\hat{f}_2^z(H_2, Z_2) - f_2^z(H_2, Z_2)\|_2,\|\hat{g}_2(H_2, Z_2) - g_2(H_2, Z_2)\|_2\right\} =& o_p(1)\\
    \sqrt{n}\|\hat{a}_t^\0(H_t, Z_t) - a_t^\0(H_t, Z_t)\|_2\cdot \|\hat{f}_t^\0(H_t, Z_t) - f_t^\0(H_t, Z_t)\|_2 =& o_p(1)\\
    \sqrt{n} \|\hat{a}_1^0(H_1, Z_1) - a_1^\0(H_1, Z_1)\|_2\cdot \|\hat{f}_2^\0(H_2, Z_2) - f_2^\0(H_2, Z_2)\|_2 =& o_p(1)
    \end{align*}
    \item The variance of the moment is non-zero and bounded: 
    $$
    \sigma^2=\E\left[\psi(W; g, a)\right]^{-2} \E\left[\left\{\phi_z(W;f, a) - \phi_0(W;f, a) - \tau_d\cdot   \psi(W; g, a)\right\}^2\right] \in (0, \infty).
    $$
\end{enumerate}
Then the estimate $\hat{\tau}_d$ defined in Equation~\eqref{eqn:estimate} satisfies
$$
\hat{\tau}_d\overset{p}{\rightarrow}\tau_d,\quad \frac{\sqrt{n}}{\sigma}(\hat{\tau}_d-\tau_d)\overset{d}{\rightarrow}\mathcal{N}(0,1),\quad \mathbb{P}(\tau_d\in \hat{\tau}_d \pm 1.96 \hat{\sigma}n^{-1/2})\rightarrow 0.95. 
$$
where:
\begin{align}\label{eqn:approx-variance}
    \hat{\sigma}^2 = \E_n\left[\psi(W; \hat{g}, \hat{a})\right]^{-2} \E_n\left[\left\{\phi_z(W; \hat{f}, \hat{a}) - \phi_0(W; \hat{f}, \hat{a}) - \hat{\tau}_d\cdot   \psi(W; \hat{g}, \hat{a})\right\}^2\right]
\end{align}
Moreover, the estimate is asymptotically linear, with influence function representation:
\begin{align}
    \sqrt{n} (\hat{\tau}_d - \tau_d) = \frac{1}{\sqrt{n}} \sum_{i=1}^n \E[\psi(W;g,a)]^{-1}(\phi_z(W; f, a) - \phi_0(W; f, a) - \tau_d\cdot \psi(W; g,a))
\end{align}
\end{theorem}
\noindent The statement follows along very similar lines as Theorem~9 of \citep{chernozhukov2022automatic}, hence we omit its proof. This theorem allows us to easily construct confidence intervals for the LATE of interest. The asymptotic normality result also holds without sample-splitting, under more assumptions on the machine learning estimators. For instance, based on the results in \citep{belloni2017program,chernozhukov2020adversarial,chen2022debiased}, sample-splitting can be avoided if the function spaces used by the machine learning estimators have small statistical complexity as captured by the notion of the critical radius \citep{chernozhukov2022automatic} or the machine learning estimation algorithms used are leave-one-out stable \citep{chen2022debiased} or if we use particular algorithms such as the Lasso based estimators under sparsity assumptions and with a theoretically-driven regularization penalty \citep{belloni2017program}.

An identical estimator and theorem can also be developed for the quantity $\beta_{1,1}$ identified in Theorem~\ref{thm:2_identification}, by simply defining the function $g_2$ as:
\begin{align}
    g_2(H_2, Z_2) = \Pr(D\neq (0,0)\mid H_2, Z_2) = \E[\I(D\neq (0,0))\mid H_2, Z_2]
\end{align}
The rest of the results of this section follow verbatim. Similarly, under the strong assumptions presented in Equation~\eqref{eqn:strong-ass}, an asymptotically linear estimate $\hat{\tau}_{1,1}$ for the Always-Treat LATE $\tau_{11}$ can be easily constructed. In particular, note that the estimates $\hat{\tau}_{1,0}, \hat{\tau}_{0,1}, \hat{\beta}_{1,1}$ are asymptotically linear. Furthermore, the quantity $\gamma_{d,z}:=\Pr(D(z)=d)$ can also be automatically debiased, for any $d\in {\cal D}^2$ and $z\in {\cal Z}^2$, and an asymptotically linear estimate can be constructed as $\hat{\gamma}_{d,z} = \E_n[\psi(W;\hat{g}, \hat{a})]$, with $\psi$ as defined in Equation~\eqref{eqn:psi} and nuisance functions $g$ as defined in Equation~\eqref{eqn:treatment-nested} and $a$ as defined in Equation~\eqref{eqn:riesz}. Since Equation~\eqref{eqn:easy-always-treat} states that $\tau_{11}$ can be expressed in terms of all these quantities, the plug-in estimate:
\begin{align*}
    \hat{\tau}_{1,1} =~& \frac{\hat{\beta}_{1,1} \sum_{d'\leq z:d\neq (0,0)} \hat{\gamma}_{d',(1,1)}}{\hat{\gamma}_{(1,1),(1,1)}} - \hat{\tau}_{1,0} \frac{\hat{\gamma}_{(1,0), (1,1)}}{\hat{\gamma}_{(1,1),(1,1)}} - \hat{\tau}_{0,1} \frac{\hat{\gamma}_{(0,1), (1,1)}}{\hat{\gamma}_{(1,1),(1,1)}}
\end{align*}
Since all the separate estimates are asymptotically linear, the estimate $\hat{\tau}_{1,1}$ will also be asymptotically linear and its influence function representation (and therefore its variance) can be easily calculated using standard influence function calculus \citep{newey1994large}.

\subsection{Debiased Estimation of Always-Treat LATE under Staggered Compliance}\label{sec:dml-always}

A similar automatic debiased machine learning estimation and inference procedure can be developed for the always-treat LATE under Staggered Compliance. Given the notation introduced in this section, we can re-write the statistical estimand in Equation~\eqref{eqn:staggered-id}, which identifies the when-to-treat LATE as:
\begin{align*}
    \tau_{11} = \frac{\E\left[q(H_1, 1) - f_1^\0(H_1, 0)\right]}{\E\left[p(H_1, 1)\right]} 
\end{align*}
where the function $f_1^z$ as defined in the when-to-treat case and the functions $q,p$ defind as the following nested regression functions\footnote{Note here that $f_2=f_2^\1=f_2^\0$ is the regression function $\E[Y\mid H_2, Z_2]$, not $\E[Y\mid H_2, D_1]$.}
\begin{align}
    q(H_1, Z_1)  :=~& \E\left[f_2(H_2, D_1)\mid H_1, Z_1\right]\\
    p(H_1, Z_1) :=~& \E\left[D_1\mid H_1, Z_1\right]
\end{align}
where we used the short-hand $f_2=f_2^\1=f_2^\0$, since the function $f_2$ is independent of the intervention vector $z$. 
Thus $\tau_{11}$ can be viewed as the solution to a moment equation, which is linear in $\tau_{11}$ and linear in a set of nuisance functions that correspond to nested regression functions:
\begin{align*}
    \E\left[q(H_1, 1) - f_1^\0(H_1, 0) - \tau_{11} \cdot p(H_1, 1)\right] = 0
\end{align*}

The debiased version of this moment is of the form:
\begin{align*}
   \E\left[\rho(W; q, f, a, \gamma) - {\phi}_0(W;f, a)  - \tau_{11}\cdot  {\pi}(W; p, a)\right] = 0
\end{align*}
where $\phi_0$, $f$, $a$, exactly as defined in the when-to-treat LATE case and:
\begin{align*}
    \rho(W; q, f, \gamma) =~& q(H_1, 1) + a_1^\1(H_1, Z_1) (f_2(H_2, D_1) + q(H_1, Z_1)) - \gamma(H_2, Z_2) (Y - f_2(H_2, Z_2)),\\
    \pi(W;p,a) =~& p(H_1, 1) + a_1^\1(H_1, Z_1)\,(D_1 - p(H_1, Z_1)) 
\end{align*}
where $\alpha_1^\1$ is the Riesz representer of the functional $L(q) = \E[q(H_1, 1)]$ (as defined in the when-to-treat case in Equation~\eqref{eqn:riesz-closed} and Equation~\eqref{eqn:riesz}) and $\gamma$ is the Riesz representer of the functional \begin{align*}
L(f)=\E\left[\alpha_1^\1(H_1, Z_1) f(H_2, D_1)\right], \quad\quad  f: (H_2, Z_2) \to R.
\end{align*}
This Riesz representer $\gamma$ can be estimated in an automatic manner, as the solution to the risk minimization problem:
\begin{align}
    \gamma = \argmin_{\gamma \in \Gamma} \E\left[\gamma(H_2, Z_2)^2 - 2 a_1^\1(H_1, Z_1) \gamma(H_2, D_1)\right]
\end{align}
or alternatively, in a plug-in manner using the closed form characterization:
\begin{align*}
    \gamma(H_2, Z_2) =~& \frac{1\{Z_1=1\}}{\Pr(Z_1=1\mid H_1)} \frac{1\{Z_2=D_1\}}{\Pr(Z_2=D_1\mid H_1)}
\end{align*}

With the debiased and Neyman orthogonal moment at hand, we can construct an asymptotically normal estimator, by estimating the nuisance functions on a separate part using generic machine learning procedures (or in a cross-fitting manner) and then solving the empirical plug-in moment equation:
\begin{align*}
    \E_n\left[\rho(W; \hat{q}, \hat{f}, \hat{a}, \hat{\gamma}) - {\phi}_0(W;\hat{f}, \hat{a})  - \hat{\tau}_{1,1}\cdot  {\pi}(W; \hat{p}, \hat{a})\right] = 0
\end{align*}
or in closed form:
\begin{align}
    \hat{\tau}_{1,1} = \frac{\E_n\left[\rho(W; \hat{q}, \hat{f}, \hat{a}, \hat{\gamma}) - {\phi}_0(W;\hat{f}, \hat{a})\right]}{\E_n\left[{\pi}(W; \hat{p}, \hat{a})\right]}
\end{align}
Under assumptions analogous to Theorem~\ref{thm:inference}, one can show that this estimate satisfies the same asymptotic normality, asymptotic linearity and confidence interval construction statements. For example, a 95\% confidence interval can be constructed as:
\begin{align*}
\tau_{11} \in~& [\hat{\tau}_{1,1} \pm 1.96 \hat{\sigma} n^{-1/2}], &
    \hat{\sigma}^2 =~& \E_n\left[\pi(W; \hat{p}, \hat{a})\right]^{-2} \E_n\left[\left\{\rho(W; \hat{q}, \hat{f}, \hat{a}, \hat{\gamma}) - \phi_\0(W; \hat{f}, \hat{a}) - \hat{\tau}_{1,1}\cdot   \pi(W; \hat{p}, \hat{a})\right\}^2\right]
\end{align*}

\section{Dynamic LATEs with Multiple Periods}

In the many period setting, we observe sequences of random variables $\{S_{t-1}, Z_t, D_t\}_{t=1}^T$ and a final outcome $Y$ for each unit. We write the realized states as $s=s_{0:T-1} \in \mathcal{S}^T$, encouragements as $z=z_{1:T}, z'=z_{1:T}' \in \mathcal{Z}^T$, treatments as $d=d_{1:T} \in \mathcal{D}^T$, and long term outcomes as $y \in \mathcal{Y}$. We will denote the set of all observed random states, encouragements and treatments as $S=S_{0:T-1}$, $Z=Z_{1:T}$ and $D=D_{1:T}$.
States $S = S_{0:T-1}$ are sets of observed random variables that may contain short term outcomes, time varying covariates, covariates prior to the experiment. We will also use the notational convention 
\begin{align*}
d_{0}, z_{0}, s_{-1}, D_{0}, Z_{0}, S_{-1}=\emptyset, \text{ and } S_{T}\equiv Y,
\end{align*}
i.e. the outcome can be thought as the state at the end of the last period. We assume that the observed variables adhere to a nonparametric structural causal model that satisfies the exclusion restrictions in Assumption \ref{assume:scm}, which is the natural extension of the causal graph presented in Figure~\ref{fig:dag_1} to many periods.

\begin{assumption}[Counterfactuals, Consistency and Exclusion Restrictions]\label{assume:scm_many} We assume the existence of a set of primitive counterfactual random processes that satisfy the exclusion restrictions encoded in the graph of Figure~\ref{fig:dag_1}, i.e.,
\begin{align}\label{eqn:counterfactuals_many}
\left\{S_0(.), \{Z_t(z_{<t}, d_{<t}, s_{<t}), D_t(z_{\leq t}, d_{<t}, s_{<t}), S_{t}(d_{\leq t}, s_{<t})\}_{t=1}^T\mid s\in {\cal S}^T, z\in {\cal Z}^T, d\in {\cal D}^T\right\}.
\end{align}
The observed variables $\{Z_t, D_t, S_{t}\}_{t=1}^T$ are generated by first drawing the counterfactual random processes from some fixed distribution, independently across units, and then recursively evaluating them, i.e., $S_0 := S_(.)$ and for all $t\in \{1,\ldots, T\}$
\begin{align*}
Z_{t} :=~& Z_{t}(Z_{<t}, D_{<t}, S_{<t})\\
D_{t} :=~& D_{t}(Z_{\leq t}, D_{<t}, S_{<t})\\
S_{t} :=~& S_{t}(D_{\leq t}, S_{<t})
\end{align*}
\end{assumption}
\noindent
That is, for each time period $t > 0$, an encouragement $Z_t$ is allowed to be given based on prior encouragements, treatments, and states; treatment $D_t$ is allowed to be a function of any and all prior observed random variables; and states may only be functions of prior treatments and states. Since the counterfactual random variables are identically and independently distributed across units, we implicitly assume SUTVA.

Definition~\ref{defn:int-cnt} of intervention counterfactuals also applies directly to the many period setting. Moreover, these counterfactuals satisfy the consistency property in Assumption~\ref{assume:2_consistency}.
Analogous to Definition~\ref{defn:central-cnt}, we will also use the shorthand notation $Y(d)$ and $D(z)$, for $d\in {\cal D}^T$, $z\in {\cal Z}^T$, for the following intervention counterfactuals:
\begin{align}
Y(d)\equiv~& Y(D\to d) & 
D(z)\equiv~& D(Z\to z).
\end{align}

We now define the set of minimal properties that we require for our counterfactual processes, as defined in Assumption~\ref{assume:scm_many}, for our identification argument. Our identification result will be applicable to any structural causal model that implies these properties. For instance, these properties are satisfied for many period analogue of the structural causal model associated with the graph in Figure~\ref{fig:dag_1} under the FFRCISTG model. 

To define our assumptions, it will be helpful to think of the process as a set of $T$ \emph{episodes}. At each episode $t$, first the instrument $Z_t$ gets determined based on the history prior to the episode, then the treatment $D_t$ is determined based on the instrument $Z_t$ and the history prior to the episode and then the state $S_t$ is determined based on the treatment $D_t$ and the history prior to the episode (excluding prior encouragements). For this reason it is helpful to introduce shorthand notation for the random variable $H_t$ that corresponds to the history of observations prior to epsidode $t$ and a realization $h_t$ of that history, as:
\begin{align}
    H_t :=~& \{Z_{<t}, D_{<t}, S_{<t}\}, & h_t :=~& \{z_{<t}, d_{<t}, s_{<t}\}
\end{align}
Then note that we can define the counterfactual processes and the recursive evaluation as:
\begin{align}
    \left\{S_0(.), \{Z_{t}(h_{t}), D_{t}(z_t, h_{t}), S_{t}(d_{\leq t}, s_{< t})\}_{t=1}^T\mid s\in {\cal S}^T, z\in {\cal Z}^T, d\in {\cal D}^T\right\},
\end{align}
with the convention that $S_{T}(d_T, h_T)\equiv Y(d_T, h_T)$ and for $t\in \{1,\ldots, T\}$:
\begin{align}
    Z_t :=~& Z_t(H_t), & D_t :=~& D_t(Z_t, H_t) & S_t :=~& S_t(D_{\leq t}, S_{<t})
\end{align}
with the convention that $S_T\equiv Y$.

Given these definitions, we can now state the assumptions that enable our main identification results. These assumptions are direct analogues of the two-period assumptions, therefore we omit their interpretation and refer the reader to the corresponding interpretation of each assumption that we presented in the two period setting.

\begin{assumption}[Sequential Ignorability of Instruments]\label{assume:2_ignorability_many}
For every $z\in {\cal Z}^T$ and every $t \in \{1,\ldots, T\}$:
\begin{align*}
    \{Y(D_{< t}, D_{\geq t}(Z_{<t}, z_{\geq t})), D_{\geq t}(Z_{<t}, z_{\geq t})\} \cindep& Z_t \mid H_t
\end{align*}
\end{assumption}

\begin{assumption}[Sequential Overlap and Relevance]\label{assume:2_overlap_many}
For all $z\in {\cal Z}^T$ and for all $t\in \{1,\ldots, T\}$:
\begin{align*}
\Pr\{Z_t=z_t \mid H_t\} >~& 0
\end{align*}
and if $z_t=1$ then:
\begin{align*}
\Pr\{D_t(z_{\leq t}) = 1 \mid H_t\} >~& 0
\end{align*}
\end{assumption}

\begin{assumption}[One Sided Noncompliance]\label{assume:2_one_sided_many}
$\Pr\{D_t(z_{\leq t}) \leq z_t\}=1$ for all $z\in {\cal Z}^T$.
\end{assumption}

Our target estimands will be Dynamic Local Average Treatment Effects, which can be analogously defined for many periods:
\begin{definition}[Dynamic Local Average Treatment Effect (LATE)]\label{def:dyn_late_many} For any $z\in {\cal Z}^T$ and $d\in {\cal D}^T$ we define the dynamic LATE as:
\begin{align}
    \theta(z,d) = \E[Y(d) - Y(0,0)\mid D(z)=d]
\end{align}
i.e., the average treatment effect of treatment intervention vector $d$, among the subset of the population who chooses treatment vector $d$, when recommended vector $z$. When $z=d$, we will use the short-hand notation 
\begin{align}
    \tau_d = \theta(d,d)
\end{align}
\end{definition}

\subsection{Identification and Estimation of When-to-Treat LATEs}\label{sec:dml_many}

Given these assumptions, we can now present our main identification result for When-to-Treat LATEs.
\begin{theorem}[Identification of Dynamic When-to-Treat LATEs]\label{thm:2_identification_many}
Assume
 \ref{assume:scm_many}, 
 \ref{assume:2_ignorability_many}, \ref{assume:2_overlap_many}, \ref{assume:2_one_sided_many}.
If $z=d\in\{(\0_{<t},1,\0_{>t}): t\in \{1, \ldots, T\}\}$, then: 
\begin{align*}
\tau_d = \theta(z,d) := \E[Y(d) - Y(\0) &\mid D(z)=d]
 =\frac{
    \mathbb{E}[Y(D(z)) ]
    - \mathbb{E}[Y(D(\0)) ]
  }{\Pr\{ D(z)=d \}
  }.
\end{align*}
Moreover, for all $z \in \{0,1\}^T$, we can similarly write: 
\begin{align*}
\beta_z := \mathbb{E}[Y(D(z)) - Y(\0) &\mid D(z) \neq \0]
 =\frac{
    \mathbb{E}[ Y(D(z)) ]
    - \mathbb{E}[ Y(D(\0)) ]
  }{1 - \Pr\{ D(z)=\0 \}
  }.
\end{align*}
Moreover, for any $d\in {\cal D}^T$ and $z\in {\cal Z}^T$ the counterfactual average $\E[Y(D(z))]$ is identified recursively as: 
\begin{align*}
\E[Y(D(z))]
=~&\E[f_1^{z}(H_1, z_1)]\\
\forall t\in \{1,\ldots, T-1\}: f_t^z(H_t, Z_t) :=~& \E[f_{t+1}^z(H_{t+1}, z_{t+1})\mid H_t, Z_t]\\
f_{T}^z(H_T, Z_T) :=~& \E[Y\mid H_T, Z_T]
\end{align*}
and the counterfactual probability $\Pr(D(z)=d)$ is identified recursively as:
\begin{align*}
\Pr(D(z)=d)
=~& \E[g_1^z(H_1, z_1)]\\
\forall t\in \{1,\ldots, T-1\}: g_t^z(H_t, Z_t) :=~& \E[g_{t+1}^z(H_{t+1}, z_{t+1})\mid H_t, Z_t]\\
g_{T}^z(H_T, Z_T) :=~& \Pr(D=d\mid H_T, Z_T)
\end{align*}
\end{theorem}
An analogue of Theorem~\ref{thm:2_identification_hetero} for the identification of heterogeneous dynamic When-to-Treat LATEs can also be proven in the multiple period setting, but we omit it for conciseness. In this case the heterogeneous dynamic When-to-Treat LATE would be identified as:
\begin{align}
    \tau_d(S_0) = \frac{\E[Y(D(z))\mid S_0] - \E[Y(D(\0))\mid S_0]}{\Pr\{D(z)=d\mid S_0\}}
\end{align}
with $\E[Y(D(z))\mid S_0] = f_1^z(S_0, z_1)$ and $\Pr(D(z)=d\mid S_0)=g_1^z(S_0, z_1)$ for any $z\in {\cal Z}^T$ and $d\in {\cal D}^T$.

\paragraph{Automatically Debiased Estimation and Inference}

Analogous to the two period setting, the When-to-Treat LATE $\tau_d$, for some given $z=d\in\{(\0_{<t},1,\0_{>t}): t\in \{1, \ldots, T\}\}$, identified in Theorem~\ref{thm:2_identification_many} can be estimated in an automatically debiased manner, enabling confidence interval construction even when generic machine learning estimators are used to estimate the various nuisance functions that are involved. 

To define the automatic debiasing procedure we can introduce recursively defined Riesz representers. In particular, the function $a_1^z(H_1, Z_1)$ is the Riesz representer of the linear functional $L_1(q)=\E[q(H_1, z_1)]$ and for $t\in \{2, \ldots, T\}$, the function $a_t^z(H_t, Z_t)$ is defined recursively as the Riesz representer of the linear functional $L_t(q)=\E[q(H_t, z_t) a_{t-1}^z(H_{t-1}, Z_{t-1})]$. Using the convenction that $a_0^z(H_0, Z_0)=1$, these can be estimated via recursive empirical risk minimization of the loss functions:
\begin{equation}\label{eqn:riesz_many}
\begin{aligned}
    \forall t\in \{1,\ldots, T\}: a_t^z =~& \argmin_{a\in A_t} \E[a(H_t, Z_t)^2 - 2 a_{t-1}^z(H_{t-1}, Z_{t-1})\,a(H_t, z_t)]
\end{aligned}
\end{equation}
\cite{chernozhukov2022automatic} provide mean-squared-error guarantees for this recursive estimation procedure. Alternatively, these Riesz representers can be characterized in closed form as the following recursively defined propensity ratios: 
\begin{align}
    \forall t\in \{1,\ldots, T\}: a_t^z(H_t, Z_t) =~& a_{t-1}^z(H_{t-1}, Z_{t-1}) \cdot \frac{1\{Z_t=z_t\}}{\Pr\{Z_t=z_t\mid H_t\}} 
\end{align}
and estimated in a plug-in manner, by first estimating the propensities that appear in the denominators, via generic ML classification approaches.

Letting $W=(S, Z, D, Y)$ and $f=\{f_t^z, f_t^\0\}_{t=1}^T, f_2\}$ and $g=\{g_t\}_{t=1}^T$ and $a=\{a_t^z, a_t^\0\}_{t=1}^T$, we can define a debiased moment equation that the target parameter $\tau_d$ must satisfy, as:
\begin{align}
    \E[{\phi}_z(W; f, a) - {\phi}_\0(W; f, a) - \tau_d\cdot  {\psi}(W; g, a)] = 0
\end{align}
where
\begin{equation*}
\begin{aligned}
    {\phi}_z(W; f,a) =~& {f}_1^z(H_1, z_1) + \sum_{t=1}^{T-1} a_t^z(H_t, Z_t) ({f}_{t+1}^z(H_{t+1}, z_t) - {f}_t^z(H_t, Z_t)) + a_T^z(H_T, Z_T) (Y - {f}_t^z(H_t, Z_t))\\
    {\psi}(W; g, a) =~& {g}_1(H_1, z_1) + \sum_{t=1}^{T-1} a_t^z(H_t, Z_t) (g_{t+1}(H_{t+1}, z_{t+1}) - {g}_t(H_t, Z_t)) + {a}_T^z(H_T, z_T)(\I(D=d) - {g}_T(H_T, Z_T))
\end{aligned}
\end{equation*}
Following identical arguments as in Theorem~6 of \cite{chernozhukov2022automatic}, this moment equation is Neyman orthogonal with respect to all the nuisance functions $f, g, a$ that it depends on. Thus we can invoke the general debiased machine learning estimation paradigm.

Given any ML estimates $\hat{f}, \hat{g}, \hat{a}$, constructed on a separate sample (or in a cross-fitting manner), we can construct a debiased estimate of $\tau_d$ as the solution to the empirical plug-in analogue of the Neyman orthogonal moment equation, with respect to $\hat{\tau}_d$, i.e.:
\begin{align*}
    \E_n\left[{\phi}_z(W;\hat{f}, \hat{a}) - {\phi}_\0(W;\hat{f}, \hat{a}) - \hat{\tau}_d\cdot  {\psi}(W; \hat{g}, \hat{a})\right] = 0
\end{align*}
which takes the form:
\begin{align}\label{eqn:estimate_many}
    \hat{\tau}_d = \frac{\E_n[{\phi}_z(W;\hat{f}, \hat{a}) - {\phi}_\0(W;\hat{f}, \hat{a})]}{\E_n[{\psi}(W; \hat{g}, \hat{a})]}
\end{align}
Under assumptions analogous to Theorem~\ref{thm:inference}, one can show that this estimate satisfies the same asymptotic normality, asymptotic linearity and confidence interval construction statements. For instance, a 95\% confidence interval can be constructed as:
\begin{align*}
\tau_d \in~& [\hat{\tau}_d \pm 1.96 \hat{\sigma} n^{-1/2}], &
    \hat{\sigma}^2 =~& \E_n\left[\psi(W; \hat{g}, \hat{a})\right]^{-2} \E_n\left[\left\{\phi_z(W; \hat{f}, \hat{a}) - \phi_\0(W; \hat{f}, \hat{a}) - \hat{\tau}_d\cdot   \psi(W; \hat{g}, \hat{a})\right\}^2\right]
\end{align*}
For computational improvements, we remark that the functions $f_t^z$, only depend on the $z$ values after period $t$, i.e. $z_{>t}$, and the functions $a_t^z$ only depend on the $z$ values prior to, and including, period $t$, i.e. $z_{\leq t}$. Hence, these models can be shared if $z_{\leq t}=\0_{\leq t}$ or $z_{>t}=\0_{>t}$.

\subsection{Always-Treat LATE under Staggered Compliance}\label{sec:always-many}

We finally provide an extension of the two period identification result of the Always-Treat LATE under further restrictions, to many periods. We only provide such an extension under the stronger Staggered Compliance setting, 
\begin{assumption}[Staggered Compliance]\label{assume:staggered} If a unit is encouraged to be treated and complies with the recommendation in the first period, then they deterministically choose treatment in all subsequent periods if they are encouraged to do so, i.e.,
\begin{align}
    \Pr(D_{>1}(\1)=\1_{>1}\mid D_1=1, Z_1=1, S_0) = 1, ~~~ \text{ a.s. }
\end{align}
\end{assumption}
\noindent Note that this assumption states that if you were encouraged to take the treatment in the first period and you chose to do so, then you will take the treatment in every subsequent period, if you were also encouraged to do so. The setting of Staggered Adoption can be thought of as a special case of the above assumption, where the instruments after a treatment compliance event are thought to be set to $1$, in which case whenever a unit starts the treatment, they remain on the treatment in the future.

We also utilitize the following sequential ignorability property, which is also implied by the structural causal model that corresponds to the many period generalization of Figure~\ref{fig:dag_1}, under the FFRCISTG model.

\begin{assumption}[Sequential Ignorability of Instruments with Joint Instrument and Treatment Interventions]\label{assume:3_ignorability_many}
For every $z\in {\cal Z}^T$ and $d\in {\cal D}^T$
\begin{align*}
    \{Y(d), D_{1}(z_1)\} \cindep& Z_1 \mid H_1
\end{align*}
\end{assumption}

\begin{theorem}[Always-Treat LATE under Staggered Compliance]\label{thm:late11_id_many}
    Under Assumptions~\ref{assume:scm_many}, 
 \ref{assume:2_ignorability_many}, \ref{assume:2_overlap_many}, \ref{assume:2_one_sided_many}, \ref{assume:staggered} and \ref{assume:3_ignorability_many}, the Always-Treat LATE and the Always-Treat conditional LATE defined below
    \begin{align}
        \tau_{\1} :=~& \E[Y(\1) - Y(\0)\mid D(\1)=\1]\\
        \tau_{\1}(S_0):=~& \E[Y(\1) - Y(\0)\mid D(\1)=\1, S_0]
    \end{align}
    are identified as:
    \begin{align*}
    \tau_\1 = \frac{\E\left[f_1^{\1}(S_0,1) - f_1^{\0}(S_0, 0) - \E\left[f_2^{\1}(H_2, 1) - f_2^{\0}(H_2,0)\mid Z_1=1, D_1=0, S_0\right]\Pr(D_1=0\mid Z_1=1, S_0)\right]}{\E\left[\Pr(D_1=1\mid Z_1=1, S_0)\right]}\\
        \tau_\1(S_0) = \frac{f_1^{\1}(S_0,1) - f_1^{\0}(S_0, 0) - \E\left[f_2^{\1}(H_2, 1) - f_2^{\0}(H_2,0)\mid Z_1=1, D_1=0, S_0\right]\Pr(D_1=0\mid Z_1=1, S_0)}{\Pr(D_1=1\mid Z_1=1, S_0)}
    \end{align*}
with the functions $f_t^{z}$ for $z\in \{\0,\1\}$ as recursively defined as in Theorem~\ref{thm:2_identification_many}.
\end{theorem}

Interpreting the formula for $\tau_1$, we see that without the negative correction term, the first term roughly coincides to the quantity:
\begin{align}
    \frac{\E\left[Y(D(\1)) - Y(D(\0))\right]}{\Pr(D_1(1)=1)}
\end{align}
This would have been the LATE, if there was only one compliance choice in the data and $D(\1)=\1$ if $D_1(1)=1$ and $D(\1)=\0$ otherwise. Note that in this case, due to the fact in the event that $D_1=0$, the future treatments are deterministically $0$ and due to the exclusion restriction, the states and outcomes only depend on the instruments through the treatments, in this scenario, we would have that $f_2^{\1}(H_2,1) = f_2^{\0}(H_2,1)$ a.s., conditional on $D_1=0$. Hence, the correction term would cancel. However, in the data that we hypothesize, units that did not comply in the first period can potentially comply in subsequent periods. We only assume Staggered Compliance. In this case, future instruments can affect future treatments, which subsequently affect states and outcome, making $f_2^{\1}(H_2, 1) \neq f_2^{\0}(H_2, 0)$. In this case, the correction term, essentially removes from the vanilla term above, the bias in the vanilla effect calculation induced by these future compliers. 

Analogous to Lemma~\ref{lem:staggered}, the identification formula in Theorem~\ref{thm:late11_id_many} for $\tau_\1$, can be further simplified as:
\begin{lemma}\label{lem:staggered-many}
Under the same assumptions as in Theorem~\ref{thm:late11_id_many}, the always-treat LATE can be identified as prescribed by the following lemma.
    \begin{align}
        \tau_\1 = \frac{\E\left[\E\left[f_2^{D_1\cdot \1}(H_2, D_1)\mid H_1, Z_1=1\right] - f_1^{\0}(H_1, 0)\right]}{\E[\Pr(D_1=1\mid Z_1=1, H_1)]}
    \end{align}
\end{lemma}

\begin{remark}[Interpretation as Local Effect of Adaptive Encouragement Policy]
The numerator in the formula presented in Lemma~\ref{lem:staggered-many} has the following interpretation. Suppose that we want to evaluate the following counterfactual dynamic encouragement policy: in the first period encourage treatment (i.e. $Z_1=1$) and in the subsequent periods encourage treatment (i.e., $Z_t=1$) if the unit took the treatment in the first period, otherwise encourage no treatment (i.e., $Z_t=0$). Let $\{\pi_t\}_{t=1}^T$, denote this adaptive policy, where $\pi_t: {\cal H}_t \to {\cal Z}$, maps the information prior to episode $t$, to an encouragement $Z_t$. For convenience of notation, let $Y(D(\pi))$ denote the counterfactual outcome if we intervene and administer encouragements based on this policy. Then one can show that the first term in the numerator is exactly the average value of this counterfactual encouragement policy, i.e.:
\begin{align*}
    \E[Y(\pi)] = \E\left[\E\left[f_2^{D_1\cdot \1}(H_2, D_1)\mid H_1, Z_1=1\right]\right]
\end{align*}
Thus our identified formula can also be written as:
\begin{align*}
    \tau_\1 = \frac{\E[Y(D(\pi)) - Y(D(\0))]}{\Pr(D_1(1)=1)}
\end{align*}
Note that under this counterfactual policy, if a unit complies in the first period, then the policy encourages treatments in all subsequent periods, and since we have Staggered Compliance, the units comply with this recommendation deterministically and take the treatment. Moreover, if a unit does not comply with the recommendation, all future encouragements are zero and by  One Sided Noncompliance, all future treatments are zero. Thus, we see that under this encouragement policy, the unit takes treatments in all subsequent periods if and only if they are a complier in the first period. Hence, under this counterfactual encouragement policy, we enforced the type of artificial behavior we were looking for, i.e., that a person takes treatment in all periods, if and only if, they take treatment in the first period. Using this property, we can also deduce that $Y(D(\pi)) - Y(D(\0))=0$, when $D_1(1)=0$. Thus we can further re-interpret the ratio as the LATE of this adaptive encouragement policy:
\begin{align}
\tau_\1 = \frac{\E[(Y(D(\pi)) - Y(D(\0)))\, \I(D_1(1)=1)]}{\Pr(D_1(1)=1)} = \E[Y(D(\pi)) - Y(D(\0))\mid D_1(1)=1]
\end{align}
Under Staggered Compliance and  One Sided Noncompliance, the always treat LATE can be re-interpreted as the local effect of an adaptive encouragement policy that encourages treatment in the first period and switches off encouraging treatment in subsequent periods, if someone does not comply in the first period, and switches on encouraging treatment if someone complies.
\end{remark}

\paragraph{Automatically Debiased Estimation and Inference} 
Lemma~\ref{lem:staggered-many} provides an identification formula that is amenable to estimation and inference via automatic debiased machine learning. If we define the nested regression functions:
\begin{align*}
q(H_1, Z_1) =~& \E\left[f_2^{D_1\cdot \1}(H_2, D_1)\mid H_1, Z_1\right]\\
p(H_1, Z_1) =~& \E[D_1\mid H_1, Z_1]
\end{align*}
then we can re-write the statistical estimand for $\tau_\1$ as:
\begin{align}
    \tau_\1 = \frac{\E\left[q(H_1, 1) - f_1^\0(H_1, 0)\right]}{\E[p(H_1, 1)]}
\end{align}
This quantity is again the solution to a moment equation that depends on nuisance functions that correspond to nested regression problems 
\begin{align}
    \E\left[q(H_1, 1) - f_1^\0(H_1, 0) - \tau_\1\, \pi(H_1, 1)\right] = 0
\end{align}
and the same paradigm as in Section~\ref{sec:dml_many} can be followed to construct an automatically debiased estimate. For instance, in this case the sequence of Riesz representers that would correpond to debiasing the term $\E[q(H_1, 1)]$ would be of the form:
\begin{align}
    \gamma_1(H_1, Z_1) =~& \frac{1\{Z_1=1\}}{\Pr(Z_1=1\mid H_1)}, &
    \gamma_t(H_t, Z_t) =~& \gamma_{t-1}(H_{t-1}, Z_{t-1}) \cdot \frac{1\{Z_t=D_1\}}{\Pr(Z_t=D_1\mid H_t)}
\end{align}
while the Riesz representers used to debias the terms $\E[f_1^\0(H_1,0)]$ and $\E[p(H_1,1)]$ are the same as in the when-to-treat case.

The debiased version of this moment is of the form:
\begin{align*}
   \E[\rho(W; q, f, a, \gamma) - {\phi}_0(W;f, a)  - \tau_{\1}\cdot  {\pi}(W; p, a)] = 0
\end{align*}
where $\phi_0$, $f$, $a$, exactly as defined in the when-to-treat LATE case and:
\begin{align*}
    \rho(W; q, f, \gamma) =& q(H_1, 1) + \gamma_1(H_1, Z_1)\, \left(f_2^{D_1\cdot \1}(H_2, D_1) - q(H_1, Z_1)\right)\\
    ~&~ +\sum_{t=2}^{T-1} \gamma_t(H_t, Z_t) \left({f}_{t+1}^{D_1\cdot \1}(H_{t+1}, D_1) - {f}_t^{D_1\cdot \1}(H_t, Z_t)\right) + \gamma_T(H_T, Z_T) \left(Y - {f}_t^{D_1\cdot \1}(H_t, Z_t)\right)\\
    \pi(W;p,a) =& p(H_1, 1) + a_1^\1(H_1, Z_1)\,\left(D_1 - p(H_1, Z_1)\right) 
\end{align*}

With the debiased and Neyman orthogonal moment at hand, we can construct an asymptotically normal estimator, by estimating the nuisance functions on a separate part using generic machine learning procedures (or in a cross-fitting manner) and then solving the empirical plug-in moment equation:
\begin{align*}
    \E_n\left[\rho(W; \hat{q}, \hat{f}, \hat{a}, \hat{\gamma}) - {\phi}_0(W;\hat{f}, \hat{a})  - \hat{\tau}_{\1}\cdot  {\pi}(W; \hat{p}, \hat{a})\right] = 0
\end{align*}
or in closed form:
\begin{align}
    \hat{\tau}_{\1} = \frac{\E_n\left[\rho(W; \hat{q}, \hat{f}, \hat{a}, \hat{\gamma}) - {\phi}_0(W;\hat{f}, \hat{a})\right]}{\E_n\left[{\pi}(W; \hat{p}, \hat{a})\right]}
\end{align}
Under assumptions analogous to Theorem~\ref{thm:inference}, one can show that this estimate satisfies the same asymptotic normality, asymptotic linearity and confidence interval construction statements. For example, a 95\% confidence interval can be constructed as:
\begin{align*}
\tau_{\1} \in~& \left[\hat{\tau}_{\1} \pm 1.96 \hat{\sigma} n^{-1/2}\right], &
    \hat{\sigma}^2 =~& \E_n\left[\pi(W; \hat{p}, \hat{a})\right]^{-2} \E_n\left[\left\{\rho(W; \hat{q}, \hat{f}, \hat{a}, \hat{\gamma}) - \phi_\0(W; \hat{f}, \hat{a}) - \hat{\tau}_{\1}\cdot   \pi(W; \hat{p}, \hat{a})\right\}^2\right]
\end{align*}

\section{Experimental Evaluation}

\subsection{When-to-Treat LATEs}

We examine the performance of our debiased ML inference procedure for When-to-Treat LATEs using a simple markovian logistic-linear structural equation model:\footnote{Replication code can be found in the following link: \href{https://drive.google.com/file/d/1Zlvo1Pa_PNCg-5fy9tyyXhxhbmx2JZ-R/view?usp=sharing}{Colab Notebook}}
\begin{align*}
    U \sim~& \text{clip}(N(0, 1), -2, 2)\\
    S_0 \sim~& N(0, I_p)\\
\forall t \in \{1,\ldots, T\}: 
    Z_{t} \sim~& \text{Binomial}\left(1, \text{Logistic}(S_{t-1}[0])\right)\\
\forall t \in \{1,\ldots, T\}: 
    D_t \sim~& Z_t \cdot \text{Binomial}\left(1, \text{Logistic}(2Z_t - 1 + U)\right)\\
\forall t \in \{1,\ldots, T\}: 
    S_t \sim~& .5 S_{t-1} + D_{t-1} + U + N(0, I_p)\\
    Y \sim~& D_T + S_{T-1} + U + N(0, 1)
\end{align*}

We implemented the debiased estimation and confidence interval construction procedure described in Section~\ref{sec:dml} and Section~\ref{sec:dml_many}. We used the plug-in approach for estimating the Riesz representers, as opposed to the automatic approach. We used Lasso as a regression oracle for all nuisance functions that correspond to regression problems and $\ell_2$-penalized Logistic regression for all the nuisance functions that correspond to classification problems and clipped the propensity produced by the models to lie in $[0.01, 1]$ (whenever these propensities appear in denominators) to avoid instabilities due to extreme predicted propensities. We estimated the nuisance models in a cross-fitting manner, with 5-fold cross-validation. The regularization weight for each method was chosen via nested cross-validation.

In Figure~\ref{fig:experiments} we report the root-mean-squared error and the bias of our point estimate for the When-to-Treat LATE, as well as the coverage of the 95\% confidence interval, over 100 repetitions of the synthetic experiment. We find that our estimation procedure recovers accurate LATE estimates for $n=5000$ samples and always has approximately nominal coverage, even in high-dimensional state settings with $p=500$-dimensional state variables, under sparsity. For $n=2000$ samples, we find that the procedure recovers accurate LATE estimates and has approximately correct coverage up to $p=100$-dimensional state variables but suffers substantial error and lower-than-nominal coverage for $p=500$-dimensional state variables for the $(1,0)$ LATE. In Figure~\ref{fig:experiments-3}, we find that the good performance of our method for up to $p=100$-dimensional state variables extends also to the case of three period when-to-treat LATEs, with approximately nominal coverage and bias of lower order than RMSE. We also find that using non-linear gradient boosted forests with early stopping for all classification models, substantially improves performance in terms of RMSE (at the cost of a small increase in bias), and maintains approximately nominal coverage, in all cases when the logistic model achieved nominal coverage (see Figure~\ref{fig:experiments-nonlin}).

\begin{figure}[H]
\centering
\begin{tabular}{lc|rrr|rrr|rrr}
\toprule
     &        & \multicolumn{3}{r}{$p=$10} & \multicolumn{3}{r}{$p=$100} & \multicolumn{3}{r}{$p=$500} \\
     &    LATE    &  rmse &  bias & coverage &  rmse &  bias & coverage &  rmse &  bias & coverage \\
\midrule
\multirow[t]{2}{*}{$n=2000$} & (1, 0) & 0.205 & 0.040 & 0.940 & 0.254 & 0.043 & 0.910 & 0.406 & 0.038 & 0.800 \\
 & (0, 1) & 0.136 & 0.029 & 0.960 & 0.136 & 0.031 & 0.970 & 0.182 & 0.042 & 0.960 \\
\cline{1-11}
\multirow[t]{2}{*}{$n=5000$} & (1, 0) & 0.161 & 0.034 & 0.870 & 0.159 & 0.054 & 0.940 & 0.166 & 0.025 & 0.940 \\
 & (0, 1) & 0.079 & 0.007 & 0.960 & 0.093 & 0.011 & 0.950 & 0.098 & 0.034 & 0.910 \\
\cline{1-11}
\bottomrule
\end{tabular}
\caption{RMSE, Bias and Coverage on Synthetic Data with two periods.}\label{fig:experiments}
\end{figure}

\begin{figure}[H]
\centering
\begin{tabular}{lc|rrr|rrr}
\toprule
     &           & \multicolumn{3}{r}{$p=$10} & \multicolumn{3}{r}{$p=$100} \\
     &           &  rmse &  bias & coverage &  rmse &  bias & coverage \\
\midrule
\multirow{3}{*}{$n=2000$} & (1, 0, 0) & 0.563 & 0.050 & 0.930 & 0.668 & 0.089 & 0.930 \\
 & (0, 1, 0) & 0.377 & 0.031 & 0.930 & 0.521 & 0.115 & 0.960 \\
 & (0, 0, 1) & 0.347 & 0.039 & 0.910 & 0.468 & 0.050 & 0.930 \\
\cline{1-8}
\multirow[t]{3}{*}{$n=5000$} & (1, 0, 0) & 0.269 & 0.029 & 0.920 & 0.342 & 0.073 & 0.940 \\
 & (0, 1, 0) & 0.183 & 0.009 & 0.920 & 0.233 & 0.010 & 0.970 \\
 & (0, 0, 1) & 0.159 & 0.001 & 0.930 & 0.209 & 0.017 & 0.930 \\
\cline{1-8}
\bottomrule
\end{tabular}
\caption{RMSE, Bias and Coverage on Synthetic Data with three periods.}\label{fig:experiments-3}
\end{figure}

\begin{figure}[H]
\centering
\begin{tabular}{lc|rrr|rrr|rrr}
\toprule
     &        & \multicolumn{3}{r}{$p=$10} & \multicolumn{3}{r}{$p=$100} & \multicolumn{3}{r}{$p=$500} \\
     &        &  rmse &  bias & coverage &  rmse &  bias & coverage & rmse &  bias & coverage\\
\midrule
\multirow{2}{*}{$n=2000$} & (1, 0) & 0.206 & 0.025 & 0.930 & 0.234 & 0.091 & 0.890 & 0.296 & 0.182 & 0.790 \\
 & (0, 1) & 0.126 & 0.030 & 0.950 & 0.112 & 0.040 & 0.940 & 0.122 & 0.045 & 0.910 \\
\cline{1-11}
\multirow[t]{2}{*}{$n=5000$} & (1, 0) & 0.150 & 0.026 & 0.900 & 0.144 & 0.069 & 0.920 & 0.146 & 0.076 & 0.890 \\
 & (0, 1) & 0.076 & 0.006 & 0.960 & 0.080 & 0.022 & 0.940 & 0.088 & 0.026 & 0.890 \\
\cline{1-11}
\bottomrule
\end{tabular}
\caption{RMSE, Bias and Coverage on Synthetic Data with two periods and gradient boosted forests for all classification models.}\label{fig:experiments-nonlin}
\end{figure}

\subsection{Always-Treat LATEs under Staggered Compliance}

We examine the performance of our debiased ML inference procedure for Always-Treat LATEs using a simple markovian logistic-linear structural equation model that satisfies the staggered compliance assumption:\footnote{Replication code can be found in the following link: \href{https://colab.research.google.com/drive/1LBuJPAIYytMU1gt4aLbNt5GTBMdlDilm}{Alwas-Treat Colab Notebook}}
\begin{align*}
    U \sim~& \text{clip}(N(0, 1), -2, 2)\\
    S_0 \sim~& N(0, I_p)\\
\forall t \in \{1,\ldots, T\}: 
    Z_{t} \sim~& \text{Binomial}\left(1, \text{Logistic}(S_{t-1}[0])\right)\\
    D_1 \sim~& Z_1 \cdot \text{Binomial}\left(1, \text{Logistic}(2Z_1 - 1 + U)\right)\\
\forall t \in \{2,\ldots, T\}: 
    D_t \sim~& \I(D_1=1)\, Z_t + \I(D_1=0)\, Z_t \cdot \text{Binomial}\left(1, \text{Logistic}(2Z_t - 1 + U)\right)\\
\forall t \in \{1,\ldots, T\}: 
    S_t \sim~& .5 S_{t-1} + D_{t-1} + U + N(0, I_p)\\
    Y \sim~& D_T + S_{T-1} + U + N(0, 1)
\end{align*}

We implemented the debiased estimation and confidence interval construction procedure described in Section~\ref{sec:dml-always} and Section~\ref{sec:always-many}. We used the plug-in approach for estimating the Riesz representers, as opposed to the automatic approach. We used Lasso as a regression oracle for all nuisance functions that correspond to regression problems and $\ell_2$-penalized Logistic regression for all the nuisance functions that correspond to classification problems and clipped the propensity produced by the models to lie in $[0.01, 1]$ (whenever these propensities appear in denominators) to avoid instabilities due to extreme predicted propensities. We estimated the nuisance models in a cross-fitting manner, with 5-fold cross-validation. The regularization weight for each method was chosen via nested cross-validation.

In Figure~\ref{fig:experiments} we report the root-mean-squared error and the bias of our point estimate for the When-to-Treat LATE, as well as the coverage of the 95\% confidence interval, over 100 repetitions of the synthetic experiment. We find that our estimation procedure recovers accurate LATE estimates for $n=5000$ samples and always has approximately nominal coverage, even in high-dimensional state settings with $p=100$-dimensional state variables, under sparsity. For $n=2000$ samples, we find that the procedure recovers accurate LATE estimates and has approximately correct coverage up to $p=100$-dimensional state variables but suffers substantial bias when $p=100$. In Figure~\ref{fig:experiments-3}, we find similar performance for the case of three period always-treat LATEs, with approximately nominal coverage and bias substantially lower than RMSE. We also find that using non-linear gradient boosted forests with early stopping for all classification models, substantially improves performance in terms of RMSE (at the cost of a small increase in bias), and maintains approximately nominal coverage (see Figure~\ref{fig:experiments-always-nonlin}).

\begin{figure}[H]
\centering
\begin{tabular}{lc|rrr|rrr}
\toprule
 &  & \multicolumn{3}{r}{$p=$10} & \multicolumn{3}{r}{$p=$100} \\
 & LATE & rmse & bias & coverage & rmse & bias & coverage \\
\midrule
$n=$2000 & (1, 1) & 0.184 & 0.035 & 0.960 & 0.224 & 0.015 & 0.970  \\
\cline{1-8}
$n=$5000 & (1, 1) & 0.138 & 0.022 & 0.930 & 0.147 & 0.023 & 0.930  \\
\cline{1-8}
\bottomrule
\end{tabular}

\caption{RMSE, Bias and Coverage of Always-Treat LATE estimate on Synthetic Data with two periods and staggered compliance.}\label{fig:experiments-always}
\end{figure}

\begin{figure}[H]
\centering
\begin{tabular}{lc|rrr|rrr}
\toprule
 &  & \multicolumn{3}{r}{$p=$10} & \multicolumn{3}{r}{$p=$100} \\
 & LATE & rmse & bias & coverage & rmse & bias & coverage \\
\midrule
$n=$2000 & (1, 1, 1) & 0.523 & 0.069 & 0.930 & 1.357 & 0.274 & 0.970 \\
\cline{1-8}
$n=$5000 & (1, 1, 1) & 0.223 & 0.038 & 0.910 & 0.277 & 0.023 & 0.900 \\
\cline{1-8}
\bottomrule
\end{tabular}
\caption{RMSE, Bias and Coverage of Always-Treat LATE estimate on Synthetic Data with three periods and staggered compliance.}\label{fig:experiments-always-3}
\end{figure}

\begin{figure}[H]
\centering
\begin{tabular}{lc|rrr|rrr}
\toprule
 &  & \multicolumn{3}{r}{10} & \multicolumn{3}{r}{100} \\
 & LATE & rmse & bias & coverage & rmse & bias & coverage \\
\midrule
2000 & (1, 1) & 0.176 & 0.018 & 0.950 & 0.191 & 0.080 & 0.920 \\
\cline{1-8}
5000 & (1, 1) & 0.132 & 0.010 & 0.930 & 0.132 & 0.046 & 0.910 \\
\cline{1-8}
\bottomrule
\end{tabular}
\caption{RMSE, Bias and Coverage on Synthetic Data with two periods and gradient boosted forests for all classification models.}\label{fig:experiments-always-nonlin}
\end{figure}

\bibliographystyle{plainnat}
\bibliography{references}

\newpage

\appendix 

\section{Examples of Intervention Counterfactuals}\label{app:examples}

For clarity of exposition, we explicitly present the recursive generative process of some intervention counterfactuals that will be used throughout our analysis. These counterfactuals are generated by first drawing the primitive counterfactual random processes from Assumption~\ref{assume:2_consistency} from their fixed distribution and then recursively generating the observed counterfactual variables based on the recursive equations below.

The intervention counterfactuals $V(Z\to z)$, which we short-hand in this context as $V(z)$, for the intervention $Z\to z$, are recursively defined as:
\begin{align*}
S_0(z) :=~& S_0(.)   \\
Z_1(z) :=~& Z_{1}(S_0(z))           \\
D_1(z) :=~& D_{1}(z_1, S_0(z))   \\
S_1(z) :=~& S_{1}(D_1(z), S_0(z))   \\
Z_2(z) :=~& Z_{2}(z_1, D_1(z), S(z))   \\
D_2(z) :=~& D_{2}(z, D_1(z), S(z))  \\
Y(z) :=~& Y(D(z), S(z))
\end{align*}
These counterfactual variables are also depicted visually in Figure~\ref{fig:swig_1}.

The intervention counterfactuals $V(D\to d)$, which we short-hand in this context as $V(d)$, for the intervention $D\to d$, are recursively defined as:
\begin{align*}
S_0(d) :=~& S_0(.)   \\
Z_1(d) :=~& Z_{1}(S_0(d))           \\
D_1(d) :=~& D_{1}(Z_1(d), S_0(d))   \\
S_1(d) :=~& S_{1}(d_1, S_0(d))   \\
Z_2(d) :=~& Z_{2}(Z_1(d), d_1, S(d))   \\
D_2(d) :=~& D_{2}(Z(d), d_1, S(d))  \\
Y(d) :=~& Y(d, S(d))
\end{align*}

The intervention counterfactuals $V((Z,D)\to (z,d))$, which we short-hand in this context as $V(z,d)$, for the intervention $(Z,D)\to (z,d)$, are recursively defined as:
\begin{align*}
S_0(z,d) :=~& S_0(.)   \\
Z_1(z,d) :=~& Z_{1}(S_0(z,d))           \\
D_1(z,d) :=~& D_{1}(z_1, S_0(z,d))   \\
S_1(z,d) :=~& S_{1}(d_1, S_0(z,d))   \\
Z_2(z,d) :=~& Z_{2}(z_1, d_1, S(z,d))   \\
D_2(z,d) :=~& D_{2}(z, d_1, S(z,d))  \\
Y(z,d) :=~& Y(d, S(z,d))
\end{align*}
These counterfactual varialbes are also depicted visually in Figure~\ref{fig:swig_3}.

Note that $S_0(z)=S_0(d)=S_0(z,d)=S_0$, $Z_1(z)=Z_1(d)=Z_1(z,d)=Z_1$, $D_1(z)=D_1(z_1)$, $D_1(d)=D_1$, and most crucially $S_1(z)=S_1(D_1(z))$, $S_1(z,d)=S_1(d_1)$ and $Y(z)=Y(D(z))$.

\section{Verifying Sequential Ignorability Properties}\label{app:ignorability}

In this section, we verify the sequential ignorability properties presented in Assumption~\ref{assume:2_ignorability} and Assumption~\ref{assume:3_ignorability}, for the strucutral causal model depicted in Figure~\ref{fig:dag_1}, under the FFRCISTG assumption on the primitive counterfactual processes. We will use the machinery of single-world intervention graphs developed in \citep{richardson2013single} and simply verify that the $d$-separation (directed separation) graphical criterion holds on appropriately constructed intervention graphs that depict both original intervening variables and counterfactual downstream variables in the same Bayesian network. Each of the graphs can be thought as corresponding to a hypothetical real world intervention.

For any fixed $z\in {\cal Z}^2$, the sequential ignorability property:
\begin{align}
    \{Y(D(z)), D(z)\} \ci Z_1 \mid S_0
\end{align}
can be verified by visually verifying that the $d$-separation criterion holds in the single-world intervention graph depicted in Figure~\ref{fig:swig_1}. For simplicity in the graph we use $Y(z)$ as short hand for $Y(D(z))$. We see that the node $S_0$, $d$-separates the set of nodes $\{Y(z), D_1(z), D_2(z)\}$ from the node $Z_1$, implying the joint conditional independent statement.

\begin{figure}[H]
\centering
\begin{tikzpicture}[
  -{Latex[length=3mm,width=2mm]}, thick, node distance=.75cm,
  every node/.style={scale=.8},
  unobserved/.style={
    circle,
    font=\small,
    very thick,
      draw=black!60,
      fill=white!60,
      minimum size=12mm,
  },
  intervene/.style={
    circle, 
    font=\small,
    very thick,
    dashed,
      draw=black!60,
      fill=white!60,
      minimum size=12mm,
  },
  observed/.style={
    circle,
    font=\small,
    draw=black!60,
    fill=gray!25,
    very thick,
    minimum size=12mm,
  },
  placeholder/.style={
    circle,
    draw=white!60,
    fill=white!25,
    very thick,
    minimum size=12mm
  }
]

\node[intervene]     (Z1)            {$Z_{1}$};
\node[intervene]     (z1)   [right=.2 of Z1]                  {$z_{1}$};
\node[placeholder]  (P1)  [right=of z1]      {};
\node[placeholder]  (P2)  [below=of P1]      {};
\node[intervene]     (Z2)  [right=of P1]      {$Z_2(z_1)$};
\node[intervene]     (z2)  [right=.2 of Z2]      {$z_2$};

\node[observed]     (D1)  [below=of z1]      {$D_1(z_1)$};
\node[observed]     (D2)  [below=of z2]      {$D_2(z)$};

\node[observed]     (S1)   [below=of P2]     {$S_1(z_1)$};

\node[unobserved]   (C1)  [left=of S1]       {$U$};
\node[observed]     (S0)  [left=1.5 of C1]       {$S_0$};

\node[placeholder]  (P3)  [right=of D2]      {};

\node[observed]     (Y)   [below=of P3]      {$Y(z)$};

\path[every node/.style={font=\sffamily\small}]
    (z1) edge[solid] node[] {} (D1)
    (z1) edge[solid] node[] {} (Z2)
    (z1) edge[solid] node[] {} (D2)

    (z2) edge[solid] node[] {} (D2)
    
    (S0)  edge[bend left] node[] {} (Z1)
    (S0)  edge[bend left=70, looseness=1.5] node[] {} (Z2)
    (S1)  edge[solid] node[] {} (Y)
    (S1)  edge[solid] node[] {} (Z2)
    (S1)  edge[solid] node[] {} (D2)

    (D1) edge[solid] node[] {} (Z2)
    (D1) edge[solid] node[] {} (D2)
    (D1) edge[solid] node[] {} (S1)
    (D1) edge[solid] node[] {} (Y)
    
    (D2) edge[solid] node[] {} (Y)
;

\path[every path/.style={font=\sffamily\small,draw=orange}]
    (C1) edge[bend left]     node  [left] {} (D1)
    (C1) edge[bend right]    node  [left] {} (S1)
    (C1) edge[bend left=20]  node  [left] {} (D2)
    (C1) edge[bend left]    node  [left] {} (S0)
    (C1) edge[bend right]    node  [left] {} (Y)
    
;
\end{tikzpicture}
\caption{
  Single-World Intervention Graph for the intervention on all instruments $Z\to z$. Note that due to the exclusion restrictions implied from the graph, we have $Y(z) = Y(D(z))$. The node $S0$ also contains edges to all the gray nodes, which are omitted for simplicity.
}\label{fig:swig_1}
\end{figure}

For any fixed $z\in {\cal Z}^2$, the sequential ignorability property:
\begin{align}
    Y(D_1, D_2(Z_1, z_2)), D_2(Z_1, z_2)\} \ci Z_2 \mid S_0, S_1, Z_1, D_1
\end{align}
can be verified by visually verifying that the $d$-separation criterion holds in the single-world intervention graph depicted in Figure~\ref{fig:swig_2}. For simplicity, in the graph we use $Y(z_2)$ ad short hand for $Y(D_1, D_2(Z_1, z_2))$ and $D_2(z_2)$ as short hand for $D_2(Z_1, z_2)$. We see that the set of nodes $\{S_0, S_1, Z_1, D_1\}$, $d$-separates the set of nodes $\{Y(z_2), D_2(z_2)\}$ from the node $Z_2$, implying the joint conditional independent statement.

\begin{figure}[H]
\centering
\begin{tikzpicture}[
  -{Latex[length=3mm,width=2mm]}, thick, node distance=.75cm,
  every node/.style={scale=.8},
  unobserved/.style={
    circle,
    font=\small,
    very thick,
      draw=black!60,
      fill=white!60,
      minimum size=12mm,
  },
  intervene/.style={
    circle, 
    font=\small,
    very thick,
    dashed,
      draw=black!60,
      fill=white!60,
      minimum size=12mm,
  },
  observed/.style={
    circle,
    font=\small,
    draw=black!60,
    fill=gray!25,
    very thick,
    minimum size=12mm,
  },
  placeholder/.style={
    circle,
    draw=white!60,
    fill=white!25,
    very thick,
    minimum size=12mm
  }
]

\node[observed]     (Z1)            {$Z_{1}$};
\node[placeholder]  (P1)  [right=of Z1]      {};
\node[placeholder]  (P2)  [below=of P1]      {};
\node[intervene]     (Z2)  [right=of P1]      {$Z_2$};
\node[intervene]     (z2)  [right=.2 of Z2]      {$z_2$};

\node[observed]     (D1)  [below=of Z1]      {$D_1$};
\node[observed]     (D2)  [below=of z2]      {$D_2(z_2)$};

\node[observed]     (S1)   [below=of P2]     {$S_1$};

\node[unobserved]   (C1)  [left=of S1]       {$U$};
\node[observed]     (S0)  [left=1.5 of C1]       {$S_0$};

\node[placeholder]  (P3)  [right=of D2]      {};

\node[observed]     (Y)   [below=of P3]      {$Y(z_2)$};

\path[every node/.style={font=\sffamily\small}]
    (Z1) edge[solid] node[] {} (D1)
    (Z1) edge[solid] node[] {} (Z2)
    (Z1) edge[solid] node[] {} (D2)

    (z2) edge[solid] node[] {} (D2)

    (S0)  edge[bend left=70, looseness=1.5] node[] {} (Z2)
    (S1)  edge[solid] node[] {} (Y)
    (S1)  edge[solid] node[] {} (Z2)
    (S1)  edge[solid] node[] {} (D2)

    (D1) edge[solid] node[] {} (Z2)
    (D1) edge[solid] node[] {} (D2)
    (D1) edge[solid] node[] {} (S1)
    (D1) edge[solid] node[] {} (Y)
    
    (D2) edge[solid] node[] {} (Y)
;

\path[every path/.style={font=\sffamily\small,draw=orange}]
    (C1) edge[bend left]     node  [left] {} (D1)
    (C1) edge[bend right]    node  [left] {} (S1)
    (C1) edge[bend left=20]  node  [left] {} (D2)
    (C1) edge[bend left]    node  [left] {} (S0)
    (C1) edge[bend right]    node  [left] {} (Y)
    
;
\end{tikzpicture}
\caption{
  Single-World Intervention Graph for the intervention on the second period instrument $Z_2\to z_2$. Note that due to the exclusion restrictions implied from the graph, we use the short-hand notations on the nodes: $Y(z_2) \leftrightarrow Y(D_1, D_2(Z_1, z_2))$ and $D_2(z_2) \leftrightarrow D_2(Z_1, z_2)$. The node $S0$ also contains edges to all the gray nodes, which are omitted for simplicity.
}\label{fig:swig_2}
\end{figure}

For any fixed $d\in {\cal D}^2$ and $z\in {\cal Z}^2$, the sequential ignorability property:
\begin{align}
    \{Y(d), D_1(z_1)\} \ci Z_1 \mid S_0
\end{align}
can be verified by visually verifying that the $d$-separation criterion holds in the single-world intervention graph depicted in Figure~\ref{fig:swig_3}. We see that the node $S_0$, $d$-separates the set of nodes $\{Y(d), D_1(z_1)\}$ from the node $Z_1$, implying the joint conditional independent statement.

\begin{figure}[H]
\centering
\begin{tikzpicture}[
  -{Latex[length=3mm,width=2mm]}, thick, node distance=.75cm,
  every node/.style={scale=.8},
  unobserved/.style={
    circle,
    font=\small,
    very thick,
      draw=black!60,
      fill=white!60,
      minimum size=12mm,
  },
  intervene/.style={
    circle, 
    font=\small,
    very thick,
    dashed,
      draw=black!60,
      fill=white!60,
      minimum size=12mm,
  },
  observed/.style={
    circle,
    font=\small,
    draw=black!60,
    fill=gray!25,
    very thick,
    minimum size=12mm,
  },
  placeholder/.style={
    circle,
    draw=white!60,
    fill=white!25,
    very thick,
    minimum size=12mm
  }
]

\node[intervene]     (Z1)            {$Z_{1}$};
\node[intervene]     (z1)   [right=.2 of Z1]                  {$z_{1}$};
\node[placeholder]  (P1)  [right=of z1]      {};
\node[placeholder]  (P2)  [below=of P1]      {};
\node[intervene]     (Z2)  [right=of P1]      {$Z_2(z_1, d_1)$};
\node[intervene]     (z2)  [right=.2 of Z2]      {$z_2$};

\node[intervene]     (D1)  [below=of z1]      {$D_1(z_1)$};
\node[intervene]     (d1)  [right=.2 of D1]      {$d_1$};
\node[intervene]     (D2)  [below=of z2]      {$D_2(z_1, d_1)$};
\node[intervene]     (d2)  [right=.2 of D2]      {$d_2$};

\node[observed]     (S1)   [below=of P2]     {$S_1(d_1)$};

\node[unobserved]   (C1)  [left=of S1]       {$U$};
\node[observed]     (S0)  [left=1.5 of C1]       {$S_0$};

\node[placeholder]  (P3)  [right=of d2]      {};

\node[observed]     (Y)   [below=of P3]      {$Y(d)$};

\path[every node/.style={font=\sffamily\small}]
    (z1) edge[solid] node[] {} (D1)
    (z1) edge[solid] node[] {} (Z2)
    (z1) edge[solid] node[] {} (D2)

    (z2) edge[solid] node[] {} (D2)
    
    (S0)  edge[bend left] node[] {} (Z1)
    (S0)  edge[bend left=70, looseness=1.5] node[] {} (Z2)
    (S0)  edge[bend left] node[] {} (D1)
    (S0)  edge[bend left=52] node[] {} (D2)
    (S1)  edge[solid] node[] {} (Y)
    (S1)  edge[solid] node[] {} (Z2)
    (S1)  edge[solid] node[] {} (D2)

    (d1) edge[solid] node[] {} (Z2)
    (d1) edge[solid] node[] {} (D2)
    (d1) edge[solid] node[] {} (S1)
    (d1) edge[solid] node[] {} (Y)
    
    (d2) edge[solid] node[] {} (Y)
;

\path[every path/.style={font=\sffamily\small,draw=orange}]
    (C1) edge[bend left]     node  [left] {} (D1)
    (C1) edge[bend right]    node  [left] {} (S1)
    (C1) edge[bend left=20]  node  [left] {} (D2)
    (C1) edge[bend left]    node  [left] {} (S0)
    (C1) edge[bend right]    node  [left] {} (Y)
    
;
\end{tikzpicture}
\caption{
  Single-World Intervention Graph for the intervention on all instruments $Z\to z$ and all treatments $D\to d$. Note that due to the exclusion restrictions implied from the graph, we have $Y(d, z) = Y(d)$. The node $S0$ also contains edges to all the gray nodes, which are omitted for simplicity.
}\label{fig:swig_3}
\end{figure}

The sequential exogeneity property:
\begin{align}
    \{Y(D_1, d_2), D_2(Z_1, z_2)\} \cindep& Z_2 \mid S_0, S_1, Z_1, D_1
\end{align}
can be verified by visually verifying that the $d$-separation criterion holds in the single-world intervention graph depicted in Figure~\ref{fig:swig_4}. For simplicity, in the graph we use $Y(d_2)$ ad short hand for $Y(D_1, d_2)$ and $D_2(z_2)$ as short hand for $D_2(Z_1, z_2)$. We see that the set of nodes $\{S_0, S_1, Z_1, D_1\}$, $d$-separates the set of nodes $\{Y(d_2), D_2(z_2)\}$ from the node $Z_2$, implying the joint conditional independent statement.

\begin{figure}[H]
\centering
\begin{tikzpicture}[
  -{Latex[length=3mm,width=2mm]}, thick, node distance=.75cm,
  every node/.style={scale=.8},
  unobserved/.style={
    circle,
    font=\small,
    very thick,
      draw=black!60,
      fill=white!60,
      minimum size=12mm,
  },
  intervene/.style={
    circle, 
    font=\small,
    very thick,
    dashed,
      draw=black!60,
      fill=white!60,
      minimum size=12mm,
  },
  observed/.style={
    circle,
    font=\small,
    draw=black!60,
    fill=gray!25,
    very thick,
    minimum size=12mm,
  },
  placeholder/.style={
    circle,
    draw=white!60,
    fill=white!25,
    very thick,
    minimum size=12mm
  }
]

\node[observed]     (Z1)            {$Z_{1}$};
\node[placeholder]  (P1)  [right=of Z1]      {};
\node[placeholder]  (P2)  [below=of P1]      {};
\node[intervene]     (Z2)  [right=of P1]      {$Z_2$};
\node[intervene]     (z2)  [right=.2 of Z2]      {$z_2$};

\node[observed]     (D1)  [below=of Z1]      {$D_1$};
\node[intervene]     (D2)  [below=of z2]      {$D_2(z_2)$};
\node[intervene]     (d2)  [right=.2 of D2]      {$d_2$};

\node[observed]     (S1)   [below=of P2]     {$S_1$};

\node[unobserved]   (C1)  [left=of S1]       {$U$};
\node[observed]     (S0)  [left=1.5 of C1]       {$S_0$};

\node[placeholder]  (P3)  [right=of d2]      {};

\node[observed]     (Y)   [below=of P3]      {$Y(d_2)$};

\path[every node/.style={font=\sffamily\small}]
    (Z1) edge[solid] node[] {} (D1)
    (Z1) edge[solid] node[] {} (Z2)
    (Z1) edge[solid] node[] {} (D2)

    (z2) edge[solid] node[] {} (D2)

    (S0)  edge[bend left=70, looseness=1.5] node[] {} (Z2)
    (S0)  edge[bend left=52] node[] {} (D2)
    (S1)  edge[solid] node[] {} (Y)
    (S1)  edge[solid] node[] {} (Z2)
    (S1)  edge[solid] node[] {} (D2)

    (D1) edge[solid] node[] {} (Z2)
    (D1) edge[solid] node[] {} (D2)
    (D1) edge[solid] node[] {} (S1)
    (D1) edge[solid] node[] {} (Y)
    
    (d2) edge[solid] node[] {} (Y)
;

\path[every path/.style={font=\sffamily\small,draw=orange}]
    (C1) edge[bend left]     node  [left] {} (D1)
    (C1) edge[bend right]    node  [left] {} (S1)
    (C1) edge[bend left=20]  node  [left] {} (D2)
    (C1) edge[bend left]    node  [left] {} (S0)
    (C1) edge[bend right]    node  [left] {} (Y)
    
;
\end{tikzpicture}
\caption{
  Single-World Intervention Graph for the intervention on the second period instrument $Z_2\to z_2$ and treatment $D_2\to d_2$. We use the short-hand notation: $Y(d_2) \leftrightarrow Y(D_1, d_2)$ and $D_2(z_2) \leftrightarrow D_2(Z_1, z_2)$. The node $S0$ also contains edges to all the gray nodes, which are omitted for simplicity.
}\label{fig:swig_4}
\end{figure}

\section{Proofs}

\subsection{Proof of \Cref{thm:2_identification}}
\begin{proof}
First, we show that identification of $\beta_z$ for every $z\in {\cal Z}^2$ such that $z\neq(0,0)$ implies identification of $\tau_d$ for $d\in \{(0,1), (1,0)\}$ as a special case.
Note that for every $d=z\in \{(0,1), (1,0)\}$, we have that the event $D(z)\neq (0,0)$ is equivalent to the event $D(z)=d$.
This is because under One Sided Noncompliance, $D(z)=d$ and $D(z)=(0,0)$ are complements whenever $d=z\in \{(0,1), (1,0)\}$ since the unit can be treated only when encouraged.
Hence, for $d=z\in \{(0,1), (1,0)\}$ we almost surely have that 
\begin{align*}
    \E[Y(d)-Y(0,0)\mid D(z)=d]&\stackrel{(a)}{=} \E[Y(D(z)) - Y(0,0)\mid D(z)=d] \stackrel{(b)}{=} \E[Y(D(z)) - Y(0,0)\mid D(z)\neq (0,0)]\\
    \Pr(D(z)=d) &\stackrel{(c)}{=} \Pr(D(z)\neq (0,0)) = 1 - \Pr(D(z)=(0,0)), 
\end{align*}
\noindent 
where
 equality $(a)$ uses the conditioning event and 
 equalities $(b),(c)$ use that One Sided Noncompliance implies $D(z)\in \{(0,0), z\}$ with probability one for $d=z\in\{(0,1), (1,0)\}$.

Based on these two properties, it suffices to prove the second identification result of the Theorem that
\begin{align*}
    \beta_z \triangleq \E[Y(D(z)) - Y(0,0) &\mid D(z) \neq (0,0)]
 =\frac{
    \mathbb{E}[ Y(D(z)) ]
    - \mathbb{E}[ Y(D(0,0)) ]
  }{1 - \Pr\{ D(z)=(0,0) \}
  }
\end{align*}
for any $z\in {\cal Z}^2$. Then the first identification result
\begin{align*}
    \tau_d \triangleq \E[Y(d) - Y(0,0) &\mid D(z)=d]
 =\frac{
    \mathbb{E}[Y(D(z)) ]
    - \mathbb{E}[Y(D(0,0)) ]
  }{\Pr\{ D(z)=d \}
  }
\end{align*}
for $d\in \{(0,1), (1,0)\}$, follows as a special case.

Second, by Bayes rule:
\begin{align*}
\E[Y(D(z)) - Y(0,0) \mid D(z) \neq (0,0)]
=~&\frac{
    \mathbb{E}\left[
      \left( Y(D(z)) - Y(0,0) \right) 
      \cdot \mathbb{I} \{ D(z)\neq (0,0) \}
    \right]
  }{ 1 - \Pr\{ D(z)=(0,0) \} }.
\end{align*}
When $D(z)=(0,0)$, we have $Y(D(z)) - Y(0,0)=0$, which almost surely implies
\begin{align*}
    \left( Y(D(z)) - Y(0,0) \right) 
      \cdot \mathbb{I} \{ D(z)\neq (0,0) \} 
      =~& Y(D(z)) - Y(0,0).
\end{align*}
Combining the last two two equalities, we derive
\begin{align*}
 \E[Y(D(z)) - Y(0,0) \mid D(z) \neq (0,0)]=~& \frac{
    \mathbb{E}\left[ Y(D(z)) - Y(0,0) \right]
  }{ 1 - \Pr\{ D(z)=(0,0) \} }.
\end{align*}
Moreover, under one sided noncompliance we have
\begin{align*}
    D(0,0) = (0,0), \quad \text{a.s.}
\end{align*}
Thus, we conclude
\begin{align*}
 \E[Y(D(z)) - Y(0,0) \mid D(z) \neq (0,0)]=~& \frac{
    \mathbb{E}\left[ Y(D(z)) - Y(D(0,0)) \right]
  }{ 1 - \Pr\{ D(z)=(0,0) \} }=\frac{
    \mathbb{E}\left[ Y(D(z))\right] - \E\left[Y(D(0,0)) \right]
  }{ 1 - \Pr\{ D(z)=(0,0) \} }.
\end{align*}
The second part of the Theorem follows by applying Lemma~\ref{lem:2_dte} to each of the counterfactual means in the fractions.
\end{proof}

\begin{lemma}[Expected Outcomes and Treatments Under Encouragements]\label{lem:2_dte}
Assume \ref{assume:2_consistency},  \ref{assume:2_ignorability} and \ref{assume:2_overlap}.
Then, 
\begin{align*}
\E[Y(D(z))\mid S_0]
=~& \E[\E[Y \mid D_1, Z=z, S] \mid Z_1=z_1, S_0]\\
\Pr(D(z)=d\mid S_0)
=~& \E[\Pr(D=d\mid S, D_1, Z=z)\mid Z_1=z_1, S_0].
\end{align*}
and the quantities $\E[Y(D(z))]$ and $\Pr(D(z)=d)$ can be identified by applying the tower rule and taking an expectation of the above two equations over $S_0$.
\end{lemma}
\begin{proof}
First note that due to the exclusion restrictions imposed by Assumption~\ref{assume:2_consistency}, we have that $Y(D(z))=Y(Z\to z)$. In particular, this follows since states and outcomes only depend on prior treatments and states and not directly on encouragements. More formally, note that $Y(d)=Y(d, S(d))$, which implies that $Y(D(z))=Y(D(z), S(D(z)))$. Furthermore, by Definition~\ref{defn:int-cnt} of intervention counterfactuals (see also Appendix~\ref{app:examples}), $Y(Z\to z) = Y(D(z), S(z))$ and $S(z)=(S_0, S_1(D_1(z_1), S_0))=S(D(z))$. Thus, $Y(Z\to z)=Y(D(z), S(D(z))=Y(D(z))$. Hence, $Y(D(z))$ corresponds to the mean counterfactual outcome $Y$ under interventions on the instruments $(Z_1, Z_2)$. Since the instruments satisfy the sequential conditional ignorability Assumption~\ref{assume:2_ignorability} and the sequential overlap Asssumption~\ref{assume:2_overlap}, the following identification result follows from the generalized g-formula (see e.g. \cite[Chapter 21]{hernancausal} or \cite[Theorem 1]{chernozhukov2022automatic}), which we also derive in Lemma~\ref{lem:g-formula} for completeness:
\begin{align*}
    \mathbb{E}[Y(D(z))\mid S_0] = \mathbb{E}[ 
    \mathbb{E}[Y \mid Z_1, S, D_1, Z_2=z_2] \mid Z_1=z_1, S_0] = 
    \mathbb{E}[\mathbb{E}[Y \mid S, D_1, Z=z] \mid Z_1=z_1, S_0]
  ]
\end{align*}
Similarly, define the variable $W=\I(D=d)$. Note that $\Pr(D(z)=d\mid S_0) = \E[W(Z\to z)\mid S_0]$. Thus, $\Pr(D(z)=d\mid S_0)$ corresponds to the mean counterfactual outcome of the variable $W$, under interventions on the instruments $(Z_1, Z_2)$. Thus we can apply again the g-formula (see Lemma~\ref{lem:g-formula}) to derive:
\begin{align*}
    \Pr(D(z)=d\mid S_0) =~& \E[\E[W\mid S, D_1, Z=z]\mid Z_1=z_1, S_0]\\
    =~& \E[\E[\I(D=d)\mid S, D_1, Z=z]\mid Z_1=z_1, S_0]\\
    =~& \E[\Pr(D=d\mid S, D_1, Z=z)\mid Z_1=z_1, S_0]
\end{align*}

\end{proof}

\begin{lemma}[g-formula for Instrument Interventions]\label{lem:g-formula} For any variable $V$, let $V(z)$ be shorthand notation for $V(Z\to z)$. If the following sequential ignorability holds for $V$:
\begin{align*}
    V(z) \cindep~& Z_1 \mid S_0 & 
    V(Z_1, z_2) \cindep~& Z_2 \mid S, D_1, Z_1
\end{align*}
and the sequential overlap Assumption~\ref{assume:2_overlap} holds, then:
\begin{align}
    \E[V(z)] =~& \E[\E[\E[V\mid S, D_1, Z=z]\mid Z_1=z_1, S_0]]\\
    \E[V(z)\mid S_0] =~& \E[\E[V\mid S, D_1, Z=z]\mid Z_1=z_1, S_0]
\end{align}
\end{lemma}
\begin{proof}
The conditional part of the lemma follows by the following sequence of equalities.
    \begin{align*}
        \E[V(z)\mid S_0]
        =~& \E[V(z)\mid S_0, Z_1=z_1] \tag{Ignorability}\\
        =~& \E[V(Z_1, z_2)\mid S_0, Z_1=z_1] \tag{Conditioning Event} \\
        =~& \E[\E[V(Z_1, z_2)\mid S, D_1, Z_1]\mid S_0, Z_1=z_1] \tag{Tower Rule}\\
        =~& \E[\E[V(Z_1, z_2)\mid S, D_1, Z_1, Z_2=z_2]\mid S_0, Z_1=z_1] \tag{Ignorability}\\
        =~& \E[\E[V(Z_1, Z_2)\mid S, D_1, Z_1, Z_2=z_2]\mid S_0, Z_1=z_1] \tag{Conditioning Event}\\
        =~& \E[\E[V\mid S, D_1, Z_1, Z_2=z_2]\mid S_0, Z_1=z_1] \tag{Consistency}\\
        =~& \E[\E[V\mid S, D_1, Z=z]\mid S_0, Z_1=z_1]
    \end{align*}
    The unconditional expectation follows by taking the expectation over $S_0$: $\E[V(z)] = \E[\E[V(z)\mid S_0]]$.
\end{proof}

\subsection{Proof of \Cref{thm:2_identification_hetero}}

Following identical reasoning as in the proof of Theorem~\ref{thm:2_identification}, it suffices to only prove the theorem for the quantity $\beta_z(S_0)$. For the latter, by a standard application of the Bayes rule:
\begin{align*}
\E[Y(D(z)) - Y(0,0) \mid D(z) \neq (0,0), S_0]
=~&\frac{
    \mathbb{E}\left[
      \left( Y(D(z)) - Y(0,0) \right) 
      \cdot \mathbb{I} \{ D(z)\neq (0,0) \}\mid S_0
    \right]
  }{ 1 - \Pr\{ D(z)=(0,0) \mid S_0\} }
\end{align*}
However, when $D(z)=(0,0)$, we have $Y(D(z)) - Y(0,0)=0$. 
The latter implies that almost surely:
\begin{align*}
    \left( Y(D(z)) - Y(0,0) \right) 
      \cdot \mathbb{I} \{ D(z)\neq (0,0) \} 
      =~& Y(D(z)) - Y(0,0) 
\end{align*}
Thus we derive that:
\begin{align*}
 \E[Y(D(z)) - Y(0,0) \mid D(z) \neq (0,0), S_0]=~& \frac{
    \mathbb{E}\left[ Y(D(z)) - Y(0,0) \mid S_0\right]
  }{ 1 - \Pr\{ D(z)=(0,0) \mid S_0\} }
\end{align*}
Moreover, under one sided noncompliance we have:
\begin{align*}
    D(0,0) = (0,0), \quad \text{a.s.}
\end{align*}
Thus we can conclude:
\begin{align*}
 \E[Y(D(z)) - Y(0,0) \mid D(z) \neq (0,0), S_0]=~& \frac{
    \mathbb{E}\left[ Y(D(z))\mid S_0\right] - \E\left[Y(D(0,0)) \mid S_0\right]
  }{ 1 - \Pr\{ D(z)=(0,0) \mid S_0\} }
\end{align*}
as desired for the first part of the theorem. The theorem then follows by applying Lemma~\ref{lem:2_dte}.

\subsection{Proof of \Cref{prop:mixture}}
\begin{proof}
    \begin{align*}
        \E[Y(D(z)) - Y(0,0) &\mid D(z) \neq (0,0)] \\
        &\stackrel{(a)}{=} \sum_{d \preceq z:~d\neq (0,0)} \E[Y(D(z)) - Y(0,0) \mid D(z)=d] \Pr(D(z)=d\mid D(z)\neq (0,0))\\
        &\stackrel{(b)}{=} \sum_{d \preceq z:~d\neq (0,0)} \E[Y(d) - Y(0,0) \mid D(z)=d] \Pr(D(z)=d\mid D(z)\neq (0,0))\\
        &\stackrel{(c)}{=} \sum_{d \preceq z:~d\neq (0,0)} \E[Y(d) - Y(0,0) \mid D(z)=d] \frac{\Pr(D(z)=d)}{\Pr(D(z)\neq (0,0))}\\
        &\stackrel{(d)}{=} \sum_{d \preceq z:~d\neq (0,0)} \theta(z,d) w(z,d), 
    \end{align*}
where
 equality $(a)$ applies Law of total probability; 
 equality $(b)$ uses the conditioning event; 
 equality $(c)$ uses Bayes rule and the fact that the event $D(z)\neq(0,0)$ contains the event $D(z)=d$; and
 equality $(d)$ applies the definitions of $w$ and $\theta$ in the Proposition and Estimand \ref{def:dyn_late}, respectively.
\end{proof}

\subsection{Proof of \Cref{thm:late11_id}}

We first prove identification of the conditional LATE and in the subsequent section we invoke this result to identify the un-conditional LATE.

\subsubsection{Identification of Conditional Always-Treat LATE}
\begin{proof}
Define $\Delta = Y(D(1,1)) - Y(0,0)$.
\begin{align}
\tau_{1,1}(S_0)
&\triangleq \E[ Y(1,1) - Y(0,0) \mid D(1,1)=(1,1), S_0 ]        \nonumber \\
&\stackrel{(b)}{=} \E[ Y(D(1,1)) - Y(0,0) \mid D(1, 1)=(1, 1), S_0 ]  \nonumber \\
&\stackrel{(c)}{=} \E[ \Delta \mid D_1(1)=1, D_2(1,1)=1, S_0 ]  \nonumber \\
&\stackrel{(d)}{=} \frac{
    \E[ \Delta \mid D_1(1)=1, S_0 ]
    - \E[\Delta \mid D_1(1)=1, D_2(1,1)=0, S_0]
       \Pr\{ D_2(1,1)=0 \mid D_1(1)=1, S_0 \}
  }{ \Pr\{ D_2(1,1)=1 \mid D_1(1)=1, S_0 \} } \label{eqn:initial-formula}
\end{align}
where
  $(b)$ uses the conditioning event $D(1,1)=(1,1)$; 
  $(c)$ plugs in $\Delta$; and 
  $(d)$ uses Law of Total Probability.

If $\Pr\{D_2(1,1)=0 \mid D_1(1)=1, S_0\}=0$ as in Staggered Compliance settings, then $\E[\Delta \mid D_1(1)=1, D_2(1,1)=0, S_0]$ is not well defined since the conditioning event has zero probability, but the product with $\Pr\{D_2(1,1)=0 \mid D_1(1)=1, S_0\}$ will equal zero and the denominator will equal one, so the ratio will simplify to 
$\tau_{1,1}(S_0)=\E[ \Delta \mid D_1(1)=1, S_0 ]$.
However, if $\Pr\{D_2(1,1)=0 \mid D_1(1)=1, S_0\}\neq0$, we need to identify both of the conditional expectations and one of the conditional probabilities since $\Pr\{ D_2(1,1)=1 \mid D_1(1)=1, S_0 \} = 1- \Pr\{ D_2(1,1)=0 \mid D_1(1)=1, S_0 \}$.
We first rewrite the definitions in the Theorem as three terms: 
\begin{align*}
\beta(S_0) :=\E[ \Delta \mid D_1(1)=1, S_0 ]
&=\E[ Y(D(1,1)) - Y(0,0) \mid D_1(1)=1, S_0 ]                           \\
&=\E[ Y(D(1,1)) \mid D_1(1)=1, S_0 ] - \E[ Y(0,0) \mid D_1(1)=1, S_0 ]  \\
\gamma_{1,1}(S_0) = \Pr\{ D_2(1,1)=1 \mid D_1(1)=1, S_0 \}
&=\E[ D_2(1,1) \mid D_1(1)=1, S_0 ].
\end{align*}
Then, we identify two of the terms by defining $\E[V(1,1) \mid D_1(1)=1, S_0]$ for $V(1,1) \in \{Y(D(1,1)), D_2(1,1)\}$: 
\begin{align*}
\E[ V(1,1) \mid D_1(1)=1, S_0 ]
&\stackrel{(a)}{=}
  \E[ V(1,1) \mid D_1(1)=1, Z_1=1, S_0 ]  \\
&\stackrel{(b)}{=}
  \E[ V(1,1) \mid D_1=1, Z_1=1, S_0 ]     \\
&\stackrel{(c)}{=}
  \E[ \E[ V(1,1) \mid Z_1=1, D_1=1, S_0, S_1 ] \mid Z_1=1, D_1=1, S_0  ] \\
&=
  \E[ \E[ V(Z_1,1) \mid Z_1=1, D_1=1, S_0, S_1 ] \mid Z_1=1, D_1=1, S_0  ] \\
&\stackrel{(d)}{=}
  \E[ \E[ V(Z_1,1) \mid Z=(1,1), D_1=1, S ] \mid Z_1=1, D_1=1, S_0  ]      \\
&\stackrel{(e)}{=}
  \E[ \E[ V \mid Z=(1,1), D_1=1, S] \mid Z_1=1, D_1=1, S_0  ], \\
&=
  \E[ \E[ V \mid S, D_1, Z_1, Z_2=1] \mid Z_1=1, D_1=1, S_0  ], 
\end{align*}
\noindent where
 $(a)$ uses ignorability (i.e. $Y(D(z)), D(z) \cindep Z_1 \mid S_0$, which together with Lemma~\ref{lem:alternative-ignorability} also implies $Y(D(z)) \cindep Z_1 \mid D_1(z), S_0$ and $D_2(z) \cindep Z_1 \mid D_1(z), S_0$);  
 $(b)$ uses Consistency; 
 $(c)$ uses Tower Rule; 
 $(d)$ uses Ignorability; and 
 $(e)$ use Consistency.
The other conditional probability is equal to one minus this quantity.
Next, we identify the third term with
\begin{align*}
\E[ Y(0,0) \mid D_1(1)=1, S_0 ]
&=\frac{
    \E[ Y(0,0) \mid S_0 ]
    - \E[ Y(0,0) \mid D_1(1)=0, S_0 ] \times \Pr\{ D_1(1)=0 \mid S_0 \}
  }{ \Pr\{ D_1(1)=1 \mid S_0 \} }, 
\end{align*}
where
\begin{align*}
\E[ Y(0,0) \mid S_0 ]
=\E[ Y(D(0,0)) \mid S_0 ]                      
&\stackrel{(a)}{=}\E[ \E[ Y \mid Z=(0,0), D=(0,0), S ] \mid Z_1=0, D_1=0, S_0 ]\\
&=\E[ \E[ Y \mid S, D_1, Z_1, Z_2=0] \mid Z_1=0, S_0],\\
\E[ Y(0,0) \mid D_1(1)=0, S_0 ]
&\stackrel{(b)}{=}\E[ \E[ Y \mid Z=(1,0), D=(0,0), S ] \mid Z_1=1, D_1=0, S_0 ]          \\
&=\E[ \E[ Y \mid S, D_1, Z_1, Z_2=0] \mid Z_1=1, D_1=0, S_0 ]          \\
\Pr\{ D_1(1)=1 \mid S_0 \}
&=\Pr\{ D_1(1)=1 \mid Z_1=1, S_0 \}  
\stackrel{(c)}{=}\Pr\{ D_1=1 \mid Z_1=1, S_0 \}     \\
\Pr\{ D_1(1)=0 \mid S_0 \}
&=1 - \Pr\{ D_1(1)=1 \mid S_0 \}, 
\end{align*}
which 
 use the G-formula and One Sided Noncompliance to obtain $(a)$ and
 Consistency to obtain $(c)$.
For $(b)$, we use the new ignorability Assumption~\ref{assume:3_ignorability}, (i.e., that the following holds $\{Y(d), D_1(z)\} \cindep Z_1 \mid S_0$, which together with Lemma~\ref{lem:alternative-ignorability} implies $Y(d) \cindep Z_1 \mid D_1(z), S_0$) and the original ignorability Assumption~\ref{assume:2_ignorability} (i.e. that $Y(D_1, D_2(Z_1, z_2)) \ci Z_2 \mid S, D_1, Z_1$) as follows:
\begin{align*}
\E[ Y(0,0) \mid S_0, D_1(1)=0 ]
&=
  \E[ Y(0,0) \mid S_0, D_1(1)=0, Z_1=1 ]  \\
&=
  \E[ Y(0,0) \mid S_0, D_1=0, Z_1=1 ]     \\
&=
  \E[ Y(D_1,0) \mid S_0, D_1=0, Z_1=1 ]   \\
&=
  \E[ \E[ Y(D_1,0) \mid S, D_1, Z_1] \mid D_1=0, S_0, Z_1=1 ]   \\
&=
  \E[ \E[ Y(D_1, D_2(Z_1, 0)) \mid S, D_1, Z_1 ] \mid S_0, D_1=0, Z_1=1 ]        \tag{One-Sided} \\
&=
  \E[ \E[ Y(D_1, D_2(Z_1, 0)) \mid S, D_1, Z_1, Z_2=0] \mid S_0, D_1=0, Z_1=1 ]       \\
&=
  \E[ \E[ Y(D_1, D_2(Z_1, Z_2)) \mid S, D_1, Z_1, Z_2=0] \mid S_0, D_1=0, Z_1=1 ]       \\
&=
  \E[ \E[ Y(D_1, D_2) \mid S, D_1, Z_1, Z_2=0] \mid S_0, D_1=0, Z_1=1 ]       \\
&=
  \E[ \E[ Y \mid S, D_1, Z_1, Z_2=0] \mid S_0, D_1=0, Z_1=1 ]
\end{align*}

Thus, letting $H_1=S_0$ and $H_2=(S, D_1, Z_1)$ and $f_2(H_2, Z_2)=\E[Y\mid H_2, Z_2]$, we have derived that: 
\begin{align*}
    \beta(S_0) =~& \E[f_2(H_2, 1) \mid H_1, Z_1=1, D_1=1]\\
    ~&~~ - \frac{\E[f_2(H_2, 0) \mid H_1, Z_1=0] - \E[f_2(H_2, 0)\mid H_1, Z_1=1, D_1=0]\Pr(D_1=0\mid Z_1=1, S_0)}{\Pr\{D_1=1\mid H_1, Z_1=1\}}\\
    \gamma(S_0) =~& \E[\Pr(D_2=1 \mid H_2, Z_2=1)\mid H_1, Z_1=1, D_1=1]
\end{align*}
The simplified form given in the Theorem for $\beta(S_0)$ is provided in Lemma~\ref{lem:simplification} below.

\noindent
  Finally, to identify the second conditional expectation, we invoke the Conditional Independence in Means restriction (Assumption \ref{assume:mean_independence}): 
\begin{align*}
\E[ \Delta \mid D_1(1)=1, D_2(1,1)=0, S_0 ]           
&\stackrel{(a)}{=}\E[ Y(D(1,1)) - Y(0,0) \mid D_1(1)=1, D_2(1,1)=0, S_0 ]  \\
&\stackrel{(b)}{=}\E[ Y(1,0) - Y(0,0) \mid D_1(1)=1, D_2(1,1)=0, S_0 ]  \\
&\stackrel{(c)}{=}\E[ Y(1,0) - Y(0,0) \mid D_1(1)=1, S_0 ]\\
&\stackrel{(d)}{=}\E[ Y(1,0) - Y(0,0) \mid D_1(1)=1, D_2(1,0)=0, S_0 ] \\
&\stackrel{(e)}{=} \tau_{10}(S_0),
\end{align*}
where
 $(a)$ uses the definition of $\Delta$; 
 $(b)$ uses the conditioning event; 
 $(c)$ uses Assumption~\ref{assume:mean_independence}; 
 $(d)$ uses One Sided Noncompliance; and
 $(e)$ uses the identification result in Theorem \ref{thm:2_identification}.

We have thus completed the identification of all the quantities that are involved on the right hand side of Equation~\eqref{eqn:initial-formula}.

\end{proof}

\begin{lemma}\label{lem:simplification}
Letting $H_1=S_0$ and $H_2=(S, D_1, Z_1)$ and $f_2(H_2, Z_2)=\E[Y\mid H_2, Z_2]$, the quantity $\beta(S_0)$ can be written in the following simplified forms:
\begin{align*}
\beta(S_0) =~& \frac{\E[f_2(H_2, 1)\mid Z_1=1, S_0]-\E[f_2(H_2, 0)\mid Z_1=0, S_0]}{\Pr(D_1=1\mid Z_1=1, S_0)} \\
~& - \frac{\E[f_2(H_2, 1)-f_2(H_2,0)\mid Z_1=1, D_1=0, S_0]\Pr(D_1=0\mid Z_1=1, S_0)}{\Pr(D_1=1\mid Z_1=1, S_0)}
\end{align*}
and
\begin{align*}
    \beta(S_0) =~& \frac{\E[f_2(H_2, D_1)\mid H_1, Z_1=1] - \E[f_2(H_2, 0)\mid H_1, Z_1=0]}{\Pr(D_1=1\mid Z_1=1, H_1)} 
\end{align*}
\end{lemma}
\begin{proof}
The proof of Theorem~\ref{thm:late11_id} implies that $\beta(S_0)$ is identified as
\begin{multline*}
    \E[f_2(H_2, 1)\mid Z_1=1, D_1=1, S_0] \\
    - \frac{\E[f_2(H_2, 0) \mid Z_1=0, S_0] - \E[f_2(H_2, 0)\mid Z_1=1, D_1=0, S_0]\Pr(D_1=0\mid Z_1=1, S_0)}{\Pr(D_1=1\mid Z_1=1, S_0)}.
\end{multline*}
By the law of total probability, the first term equals
\begin{align*}
    \frac{\E[f_2(H_2, 1)\mid Z_1=1, S_0] - \E[f_2(H_2, 1)\mid Z_1=1, D_1=0, S_0]\Pr(D_1=0\mid Z_1=1, S_0)}{\Pr(D_1=1\mid Z_1=1, S_0)}.
\end{align*}
Thus, 
\begin{align*}
\beta(S_0) =~& \frac{\E[f_2(H_2, 1)\mid Z_1=1, S_0]-\E[f_2(H_2, 0)\mid Z_1=0, S_0]}{\Pr(D_1=1\mid Z_1=1, S_0)} \\
~& - \frac{\E[f_2(H_2, 1)-f_2(H_2,0)\mid Z_1=1, D_1=0, S_0]\Pr(D_1=0\mid Z_1=1, S_0)}{\Pr(D_1=1\mid Z_1=1, S_0)}.
\end{align*}
Note that the numerator of the second term can be re-written as
\begin{align*}
     (II) :=~& \E[f_2(H_2, 1)-f_2(H_2,0)\mid Z_1=1, D_1=0, H_1]\Pr(D_1=0\mid Z_1=1, H_1)]\\
     =~& \E[(f_2(H_2, 1)-f_2(H_2,0))(1-D_1)\mid Z_1=1, H_1].
\end{align*}
Merging this with the term,
\begin{align*}
    (I) := \E[f_2(H_2, 1)\mid Z_1=1, H_1], 
\end{align*}
we can simplify: 
\begin{align*}
    (I) + (II) =~& \E[f_2(H_2, 1) - (f_2(H_2, 1)-f_2(H_2,0))(1-D_1)\mid Z_1=1, H_1]\\
    =`& \E[D_1\, f_2(H_2, 1) + f_2(H_2, 0) (1-D_1)\mid Z_1=1, H_1]\\
    =~&\E[f_2(H_2, D_1)\mid Z_1=1, H_1].
\end{align*}
Thus,
\begin{align*}
    \beta(S_0) = \frac{\E[f_2(H_2, D_1)\mid H_1, Z_1=1] - \E[f_2(H_2, 0)\mid H_1, Z_1=0]}{\Pr(D_1=1\mid Z_1=1, H_1)}.
\end{align*}
\end{proof}

\begin{lemma}[From Joint Independence to Conditional Independence]\label{lem:alternative-ignorability}
    Suppose that four variables $L,M,O,N$ satisfy the conditional independence property:
    \begin{align}\label{eqn:ci-prop}
        \{L, M\} \cindep N \mid O
    \end{align}
    then for any $M$ such that $\Pr\{M \mid O\}>0$, we also have
    \begin{align}
        L \cindep N \mid M, O.
    \end{align}
\end{lemma}
\begin{proof}
By the conditional independence property in Equation~\ref{eqn:ci-prop}, we have that the probability law of $L,M,N,O$ can be factorized as 
\begin{align*}
    \Pr(L,M,N \mid O)=~&
  \Pr\{L,M \mid O\}\times \Pr\{N\mid O\}.
\end{align*}
From this we can conclude
\begin{align*}
\Pr\{ M \mid O \} \times \Pr\{ L,N  \mid M,O \}
=~& \Pr\{M\mid O\}\times \Pr\{L\mid O, M\}\times \Pr\{N\mid O\}.
\end{align*}
Since $\Pr\{M\mid O\} > 0$, this implies
\begin{align*}
\Pr\{ L,N  \mid M,O \}
=~& \Pr\{L\mid O, M\}\times \Pr\{N\mid O\}.
\end{align*}
The assumption in Equation~\ref{eqn:ci-prop} implies that $M\cindep N \mid O$, which implies $\Pr\{N\mid O\} = \Pr\{N\mid O, M\}$. Thus, 
\begin{align*}
\Pr\{ L,N  \mid M,O \}
=~& \Pr\{L\mid O, M\}\times \Pr\{N\mid O, M\}.
\end{align*}
The latter statement shows that the conditional probability law of $L, N$, conditional on $M,O$, factorizes and therefore $L\ci N \mid M, O$.
\end{proof}

\subsubsection{Identification of Unconditional Always-Treat LATE}

We will reduce the identification to the conditional LATE case. Note that we can write:
\begin{align}
    \tau_{11} =~& \E\left[Y(1,1) - Y(0,0) \mid D(1,1)=(1,1)\right]\\
    =~& \E\left[\E\left[Y(1,1) - Y(0,0)\mid S_0, D(1,1)=(1,1)\right] \mid D(1,1)=(1,1)\right]\\
    =~& \E\left[\tau_{1,1}(S_0) \mid D(1,1)=(1,1)\right]\\
    =~& \frac{\E\left[\tau_{1,1}(S_0)\, \I(D(1,1)=(1,1))\right]}{\Pr(D(1,1)=(1,1))}\\
    =~& \frac{\E\left[\tau_{1,1}(S_0)\, \Pr(D(1,1)=(1,1)\mid S_0)\right]}{\Pr(D(1,1)=(1,1))}
\end{align}
The quantity $\tau_{1,1}(S_0)$ was identified in the prior section. The quantity $\Pr(D(1,1)=(1,1))$ was identified in Theorem~\ref{thm:2_identification}. It only remains to identify $\Pr(D(1,1)=(1,1)\mid S_0)$, which can be identified by the following g-formula (applying Lemma~\ref{lem:g-formula} to the random variable $V=\I(D(1,1)=(1,1))$):
\begin{align*}
    \Pr(D(1,1)=(1,1)\mid S_0)
    =~& \E[\E[\I(D=(1,1)\mid S, D_1, Z_1, Z_2=1] \mid S_0, Z_1=1]\\
    =~& \E[\Pr(D=(1,1)\mid S, D_1, Z_1, Z_2=1) \mid S_0, Z_1=1]
\end{align*}
This completes the identification of the quantity $\tau_{11}$.

We further simplify the identification formula as follows. The quantity inside the expectation in the numerator of $\tau_{11}$ takes the form
\begin{align*}
\tau_{1,1}(S_0) \Pr(D(1,1)=(1,1)\mid S_0) = \frac{\beta(S_0) - \tau_{10}(S_0) (1 - \gamma_{1,1}(S_0))}{\gamma_{1,1}(S_0)}\Pr(D(1,1)=(1,1)\mid S_0)
\end{align*}
Moreover, $\beta(S_0)$ has $\Pr(D_1(1)\mid S_0)$ in the denominator. Thus we can write the first term as:
\begin{align*}
    (I) :=~& \frac{\beta(S_0)}{\gamma_{1,1}(S_0)}\Pr(D(1,1)=(1,1)\mid S_0) \\
    =~& \frac{\Gamma}{\Pr(D_1(1)=1\mid S_0) \Pr(D_2(1,1)=1\mid D_1(1)=1, S_0)} \Pr(D(1,1)=(1,1)\mid S_0)\\
    =~& \Gamma
\end{align*}
with
\begin{align*}
    \Gamma = \E[\E[Y\mid H_2, Z_2=D_1]\mid S_0, Z_1=1] - \E[\E[Y\mid H_2, Z_2=0]\mid S_0, Z_1=0]
\end{align*}
The second negative term, then also takes the form:
\begin{align*}
    (II) :=~& \tau_{10}(S_0) \frac{1 - \gamma_{1,1}(S_0)}{\Pr(D_2(1,1)=1\mid D_1(1)=1,S_0)} \Pr(D(1,1)=(1,1)\mid S_0)\\
    =~&
    \tau_{10}(S_0) (1 - \gamma_{1,1}(S_0))\Pr(D_1(1)=1\mid S_0)\\
    =~&\tau_{10}(S_0) (\Pr(D_1(1)=1\mid S_0) - \Pr(D(1,1)=(1,1)\mid S_0))
\end{align*}
Thus overall we get that the numerator in the $\tau_{11}$ can be written as:
\begin{align*}
    \E[\Gamma + \tau_{10}(S_0) (\Pr(D_1(1)=1\mid S_0) - \Pr(D(1,1)=(1,1)\mid S_0)] 
\end{align*}
Note that:
\begin{align*}
    \Pr(D(1,1)=(1,1)\mid S_0) =~& \E[\Pr(D=(1,1)\mid H_2, Z_2=1)\mid S_0, Z_1=1] \\
    =~& \E[D_1 \Pr(D_2=1\mid H_2, Z_2=1)\mid S_0, Z_1=1]\\
    \Pr(D_1(1)=1\mid S_0) =~& \E[D_1 \mid S_0, Z_1=1]
\end{align*}
We can thus write:
\begin{align*}
    \Pr(D_1(1)=1\mid S_0) - \Pr(D(1,1)=(1,1)\mid S_0) =~& \E[D_1 (1 - \Pr(D_2=1\mid H_2,Z_2=1))\mid S_0, Z_1=1]\\
    =~& \E[D_1 \Pr(D_2=0\mid H_2, Z_2=1)\mid S_0, Z_1=1]\\
    =~& \E[\Pr(D=(1, 0)\mid H_2, Z_2=1)\mid S_0, Z_1=1]
\end{align*}

\subsubsection{Proof of Lemma~\ref{lem:staggered}: Staggered Compliance Special Case}\label{sec:staggered}

Note that under Staggered Compliance,  Theorem~\ref{thm:late11_id} implies that:
\begin{align}\label{eqn:delta}
    \tau_{1,1}(S_0) = \beta(S_0) = \E[\Delta \mid D_1(1)=1, S_0]
\end{align}
where $\beta(S_0)$ is identified by the simplified formulas provided in Lemma~\ref{lem:simplification}.

Moreover, note that under Staggered Compliance $\Pr(D(1,1)=(1,1)\mid S_0)=\Pr(D_1(1)=1\mid S_0)=\Pr(D_1=1\mid Z_1=1, S_0)$. Thus the unconditional Always-Treat LATE, also simplifies to:
\begin{align*}
    \tau_{11} =~& \frac{\E[\E[f_2(H_2, 1)\mid Z_1=1, S_0]-\E[f_2(H_2, 0)\mid Z_1=0, S_0]]}{\E[\Pr(D_1=1\mid Z_1=1, S_0)]} \\
~& - \frac{\E[\E[f_2(H_2, 1)-f_2(H_2,0)\mid Z_1=1, D_1=0, S_0]\Pr(D_1=0\mid Z_1=1, S_0)]}{\E[\Pr(D_1=1\mid Z_1=1, S_0)]}
\end{align*}
and
\begin{align*}
    \tau_{11} = \frac{\E[\E[f_2(H_2, D_1)\mid H_1, Z_1=1] - \E[f_2(H_2, 0)\mid H_1, Z_1=0]]}{\E[\Pr(D_1=1\mid Z_1=1, H_1)]} 
\end{align*}

\subsection{Proof of \Cref{thm:2_identification_many}}
\begin{proof} Note that under one sided noncompliance, for any $z\in \{(\0_{<t}, 1, \0_{>t}): t\in \{1, \ldots, T\}\}$, the complement of $D(z)=d$ is $D(z)=\0$, since the unit can choose the treatment only when it is recommended. Thus for any $d=z\in \{(\0_{<t}, 1, \0_{>t}): t\in \{1, \ldots, T\}\}$, we have that the event $D(z)\neq \0$ is equivalent to $D(z)=d$. Hence, for $d=z\in \{(\0_{<t}, 1, \0_{>t}): t\in \{1, \ldots, T\}\}$ we have that:
\begin{align*}
    \E[Y(d)-Y(\0)\mid D(z)=d]=~& \E[Y(D(z)) - Y(\0)\mid D(z)=d] = \E[Y(D(z)) - Y(\0)\mid D(z)\neq \0]\\
    \Pr(D(z)=d) =~& \Pr(D(z)\neq \0) = 1 - \Pr(D(z)=\0) 
\end{align*}
Based on these two properties, it suffices to prove the second identification statement, i.e., 
\begin{align*}
    \E[Y(D(z)) - Y(\0) &\mid D(z) \neq \0]
 =\frac{
    \mathbb{E}[ Y(D(z)) ]
    - \mathbb{E}[ Y(D(\0)) ]
  }{1 - \Pr\{ D(z)=\0 \}
  }.
\end{align*}
and the first identification statement, i.e., that
\begin{align*}
    \tau_d := \E[Y(d) - Y(\0) &\mid D(z)=d]
 =\frac{
    \mathbb{E}[Y(D(z)) ]
    - \mathbb{E}[Y(D(\0)) ]
  }{\Pr\{ D(z)=d \}
  }.
\end{align*}
follows as a special case.

By a standard application of the Bayes rule:
\begin{align*}
\E[Y(D(z)) - Y(\0) \mid D(z) \neq \0]
=~&\frac{
    \mathbb{E}\left[
      \left( Y(D(z)) - Y(\0) \right) 
      \cdot \mathbb{I} \{ D(z)\neq \0 \}
    \right]
  }{ 1 - \Pr\{ D(z)=\0 \} }
\end{align*}
However, when $D(z)=\0$, we have $Y(D(z)) - Y(\0)=0$. 
The latter implies that almost surely:
\begin{align*}
    \left( Y(D(z)) - Y(\0) \right) 
      \cdot \mathbb{I} \{ D(z)\neq \0 \} 
      =~& Y(D(z)) - Y(\0) 
\end{align*}
Thus we derive that:
\begin{align*}
 \E[Y(D(z)) - Y(\0) \mid D(z) \neq \0]=~& \frac{
    \mathbb{E}\left[ Y(D(z)) - Y(\0) \right]
  }{ 1 - \Pr\{ D(z)=\0 \} }
\end{align*}
Moreover, under one sided noncompliance we have:
\begin{align*}
    D(\0) = \0, \quad \text{a.s.}
\end{align*}
Thus we can conclude:
\begin{align*}
 \E[Y(D(z)) - Y(\0) \mid D(z) \neq \0]=~& \frac{
    \mathbb{E}\left[ Y(D(z)) - Y(D(\0)) \right]
  }{ 1 - \Pr\{ D(z)=\0 \} }=\frac{
    \mathbb{E}\left[ Y(D(z))\right] - \E\left[Y(D(\0)) \right]
  }{ 1 - \Pr\{ D(z)=\0 \} }
\end{align*}
as desired for the first part of the theorem.
The second part of the Theorem follows by the following Lemma~\ref{lem:2_dte_many} to each of the terms in the fractions.
\end{proof}

\begin{lemma}[Expected Outcomes and Treatments Under Encouragement Interventions]\label{lem:2_dte_many}
Assume \ref{assume:scm_many}, \ref{assume:2_ignorability_many} and \ref{assume:2_overlap_many}.
for any $d\in {\cal D}^T$ and $z\in {\cal Z}^T$ the counterfactual average $\E[Y(D(z))]$ is identified recursively as: 
\begin{align*}
\E[Y(D(z))]
=~&\E[f_1(H_1, z_1)]\\
\forall t\in \{1,\ldots, T-1\}: f_t(H_t, Z_t) :=~& \E[f_{t+1}(H_{t+1}, z_{t+1})\mid H_t, Z_t]\\
f_{T}(H_T, Z_T) :=~& \E[Y\mid H_T, Z_T]
\end{align*}
and the counterfactual probability $\Pr(D(z)=d)$ is identified recursively as:
\begin{align*}
\Pr(D(z)=d)
=~& \E[g_1(H_1, z_1)]\\
\forall t\in \{1,\ldots, T-1\}: g_t(H_t, Z_t) :=~& \E[g_{t+1}(H_{t+1}, z_{t+1})\mid H_t, Z_t]\\
g_{T}(H_T, Z_T) :=~& \Pr(D=d\mid H_T, Z_T)\\
\end{align*}
\end{lemma}
\begin{proof}
First note that due to the exclusion restriction imposed by Assumption~\ref{assume:scm_many}, we have that $Y(D(z))=Y(Z\to z)$. Hence, $Y(D(z))$ corresponds to the mean counterfactual outcome $Y$ under interventions on the instruments $Z_{1:T}$. Since the instruments satisfy the sequential conditional ignorability Assumption~\ref{assume:2_ignorability_many} and the sequential overlap Asssumption~\ref{assume:2_overlap_many}, the following identification result follows from the generalized g-formula (see e.g. \cite[Chapter 21]{hernancausal} or \cite[Theorem 1]{chernozhukov2022automatic}), which we also derive in Lemma~\ref{lem:g-formula_many} for completeness:
\begin{align*}
\E[Y(D(z))]
=~&\E[f_1^z(H_1, z_1)]\\
\forall t\in \{1,\ldots, T-1\}: f_t^z(H_t, Z_t) :=~& \E[f_{t+1}^z(H_{t+1}, z_{t+1})\mid H_t, Z_t]\\
f_{T}^z(H_T, Z_T) :=~& \E[Y\mid H_T, Z_T]
\end{align*}
Similarly, define the variable $W=\I(D=d)$. Note that $\Pr(D(z)=d) = \E[W(Z\to z)]$. Thus, $\Pr(D(z)=d)$ corresponds to the mean counterfactual outcome of the variable $W$, under interventions on the instruments $Z_{1:T}$, i.e., $\E[W(z)]$. Thus we can apply again the g-formula (see Lemma~\ref{lem:g-formula_many}) to derive:
\begin{align*}
\Pr(D(z)=d)
=\E[W(z)]=~& \E[g_1^z(H_1, z_1)]\\
\forall t\in \{1,\ldots, T-1\}: g_t^z(H_t, Z_t) :=~& \E[g_{t+1}^z(H_{t+1}, z_{t+1})\mid H_t, Z_t]\\
g_{T}^z(H_T, Z_T) :=~& \E[\I(D=d)\mid H_T, Z_T] = \Pr(D=d\mid H_T, Z_T)\\
\end{align*}

\end{proof}

\begin{lemma}[g-formula for Instrument Interventions]\label{lem:g-formula_many} For any variable $V$, let $V(z)$ be shorthand notation for $V(Z\to z)$. If the following sequential ignorability holds for $V$:
\begin{align*}
    V(Z_{<t}, z_{\geq t}) \cindep~& Z_t \mid H_t
\end{align*}
and the sequential overlap Assumption~\ref{assume:2_overlap_many} holds, then:
\begin{align*}
\E[V(z) \mid S_0] = \E[V(z) \mid H_1]
=~& f_1(H_1, z_1)\\
\forall t\in \{1,\ldots, T-1\}: f_t(H_t, Z_t) :=~& \E[f_{t+1}(H_{t+1}, z_{t+1})\mid H_t, Z_t]\\
f_{T}(H_T, Z_T) :=~& \E[V\mid H_T, Z_T]
\end{align*}
Moreover:
\begin{align*}
\E[V(z)] = \E[\E[V(z)\mid S_0]] = \E[f_1(H_1, z_1)]
\end{align*}
and for any $t\in \{1, \ldots, T\}$:
\begin{align}\label{eqn:middle-step-g-formula}
    \E[V(Z_{\leq t}, z_{> t})\mid H_t, Z_t] = f_t(H_t, Z_t)
\end{align}
\end{lemma}
\begin{proof}
We prove the conditional on $S_0$ statement. The unconditional statement then follows via a simple application of the tower rule of expectations. We also note that the history sets are a nested sequence, in that $H_t=\{Z_{<t}, D_{<t}, S_{<t}\}$ contains all the information contained in $H_{t-1}\cup Z_{t-1}$. Then we note that:
    \begin{align*}
        \E[V(z)\mid S_0] =~& \E[V(z)\mid H_1]\\
        =~& \E[V(z)\mid H_1, Z_1=z_1] \tag{ignorability}\\
        =~& \E[V(Z_1, z_{> 1})\mid H_1, Z_1=z_1]
    \end{align*}
We then show by induction that for all $t\in \{1, \ldots, T\}$:
\begin{align}
    \E[V(Z_{\leq t}, z_{> t})\mid H_t, Z_t] = \E[f_{t+1}(H_{t+1}, z_{t+1})\mid H_t, Z_t] =: f_t(H_t, Z_t)
\end{align}
with the convention that $f_{T+1}(H_{T+1}, z_{T+1})\equiv V$ and $f_t$ is recursively defined in the statement of the Lemma for any $t\leq T$.
The lemma then follows by noting that given this property, we can re-write:
\begin{align*}
    \E[V(z)\mid S_0] = f_1(H_1, z_1)
\end{align*}
For the inductive proof, we first prove this fact for $t=T$, which follows simply by consistency:
\begin{align*}
\E[V(Z_{\leq T})\mid H_T, Z_T] = \E[V\mid H_T, Z_T]
\end{align*}
Assuming this holds for all $t'>t$, we prove this property for $t$:
\begin{align*}
    \E[V(Z_{\leq t}, z_{>t})\mid H_{t}, Z_{t}] =~& \E[\E[V(Z_{\leq t}, z_{>t})\mid H_{t+1}]\mid H_{t}, Z_{t}] \tag{tower rule}\\
    =~& \E[\E[V(Z_{\leq t}, z_{>t})\mid H_{t+1}, Z_{t+1}=z_{t+1}]\mid H_{t}, Z_{t}] \tag{ignorability}\\
    =~& \E[\E[V(Z_{\leq t+1}, z_{>t+1})\mid H_{t+1}, Z_{t+1}=z_{t+1}]\mid H_{t}, Z_{t}]\\
    =~& \E[f_{t+1}(H_{T+1}, z_{t+1}) \mid H_{t}, Z_{t}] \tag{induction hypothesis}\\
    =~& f_{t}(H_t, Z_t)
\end{align*}
which proves the induction step and completes the proof of the lemma.
\end{proof}

\subsection{Proof of \Cref{thm:late11_id_many}}

We first prove identification of the conditional LATE and in the subsequent section we invoke this result to identify the un-conditional LATE.

\subsubsection{Identification of Conditional Always-Treat LATE}
\begin{proof}
Define $\Delta = Y(D(\1)) - Y(\0)$.
\begin{align}
\tau_{\1}(S_0)
&\triangleq \E[ Y(\1) - Y(\0) \mid D(\1)=\1, S_0 ]        \nonumber \\
&= \E[ Y(D(\1)) - Y(\0) \mid D(\1)=\1, S_0 ]  \nonumber \\
&= \E[\Delta \mid D(\1)=\1, S_0 ]  \nonumber \\
&= \E[ \Delta \mid D_1=1, D_{>1}(\1)=\1_{>1}, S_0 ] \nonumber\\
&= \frac{
    \E[ \Delta \mid D_1(1)=1, S_0 ]
    - \E[\Delta \mid D_1(1)=1, D_{>1}(\1)\neq \1_{>1}, S_0]
       \Pr\{ D_{>1}(\1)\neq \1_{>1} \mid D_1(1)=1, S_0 \}
  }{ \Pr\{ D_{>1}(\1)= \1_{>1}  \mid D_1(1)=1, S_0 \} } \nonumber
\end{align}
Note that sequential ignorability implies that, which implies, by invoking Lemma~\ref{lem:alternative-ignorability}, that $D_{>1}(\1) \ci Z_1 \mid D_1(1), S_0$. Thus we can invoke the staggered compliance Assumption~\ref{assume:staggered} to derive:
\begin{align}
    \Pr\{ D_{>1}(\1)\neq \1_{>1} \mid D_1(1)=1, S_0 \} =~& \Pr\{ D_{>1}(\1)\neq \1_{>1} \mid D_1(1)=1, Z_1=1, S_0 \} \nonumber\\
    =~& \Pr\{ D_{>1}(\1)\neq \1_{>1} \mid D_1=1, Z_1=1, S_0 \} \nonumber\\
    =~& 0 \tag{Staggered Compliance}
\end{align}
Hence, the second term in the numerator vanishes and the denominator is equal to one, leading to:
\begin{align}
\tau_{\1}(S_0) &= \E[ \Delta \mid  D_1(1)=1, S_0] 
\end{align}
We then split the final quantity as:
\begin{align*}
\E[ \Delta \mid D_1(1)=1, S_0 ]
&=\E[ Y(D(\1)) - Y(\0) \mid D_1(1)=1, S_0 ]                           \\
&=\E[ Y(D(\1)) \mid D_1(1)=1, S_0 ] - \E[ Y(\0) \mid D_1(1)=1, S_0 ]  
\end{align*}
we first apply the following variant of the g-formula to $\E[Y(D(\1)) \mid D_1(1)=1, S_0]$. In the equations below we use sequential ignorability, consistency, the tower rule of expectations and the fact that $H_t \supseteq H_{t-1} \cup \{D_{t-1}, Z_{t-1}\}$. Moreover, we use Lemma~\ref{lem:alternative-ignorability} to state the alternative conditional ignorability statements: $Y(D(z)) \cindep Z_1 \mid D_1(z), S_0$. With these properties we can write:
\begin{align*}
\E[ Y(D(\1)) \mid D_1(1)=1, S_0 ]
&= 
\E[ Y(D(\1)) \mid D_1(1)=1, H_1 ] \tag{$H_1=S_0$}\\
&=
  \E[ Y(D(\1)) \mid D_1(1)=1, Z_1=1, H_1 ]  \tag{$Y(D(\1))\cindep Z_1\mid D_1(1), S_0$}\\
&=
  \E[ Y(D(\1)) \mid D_1=1, Z_1=1, H_1 ]     \\
&=
  \E[ Y(D(Z_1, \1_{>1})) \mid D_1=1, Z_1=1, H_1 ]     \\
&=
  \E[ \E[ Y(D(Z_1, \1_{>1})) \mid H_2 ] \mid Z_1=1, D_1=1, H_1  ] \tag{$H_2 \supseteq H_1 \cup \{D_1, Z_1\}$}\\
&=
  \E[ \E[ Y(D(Z_1, \1_{>1})) \mid H_2, Z_2=1 ] \mid Z_1=1, D_1=1, H_1]     \tag{ignorability}\\
&=
  \E[ \E[ Y(D(Z_{\leq 2}, \1_{>2})) \mid H_2, Z_2=1 ] \mid Z_1=1, D_1=1, H_1] 
\end{align*}
Applying Lemma~\ref{lem:g-formula_many} and invoking Equation~\eqref{eqn:middle-step-g-formula} from Lemma~\ref{lem:g-formula_many}, we now have that:
\begin{align*}
    \E[ Y(D(Z_{\leq 2}, \1_{>2})) \mid H_2, Z_2=1 ] = f_2^{\1}(H_2, 1)
\end{align*}
where $f_t^{\1}$ as recursively defined in Theorem~\ref{thm:2_identification_many}. Thus we can recursively identify the above quantity as:
\begin{equation}
\begin{aligned}
\E[ Y(D(\1)) \mid D_1(1)=1, S_0 ] =~& \E[ f_2^{\1}(H_2, 1)\mid Z_1=1, D_1=1, H_1]
\end{aligned}
\end{equation}
Next, we argue identification of the quantity $\E[Y(\0)\mid D_1(1)=1, S_0]$.
\begin{align*}
\E[ Y(\0) \mid D_1(1)=1, S_0 ]
&=\frac{
    \E[ Y(\0) \mid S_0 ]
    - \E[ Y(\0) \mid D_1(1)=0, S_0 ] \times \Pr\{ D_1(1)=0 \mid S_0 \}
  }{ \Pr\{ D_1(1)=1 \mid S_0 \} }, 
\end{align*}
First we note that the probabilities can be easily identified by conditional ignorability:
\begin{align}\label{eqn:d1-1}
\Pr\{ D_1(1)=1 \mid S_0 \} &= \Pr\{ D_1(1)=1 \mid H_1 \}
=\Pr\{ D_1(1)=1 \mid Z_1=1, H_1 \}  
=\Pr\{ D_1=1 \mid Z_1=1, H_1 \}     \\
\Pr\{ D_1(1)=0 \mid S_0 \}
&=1 - \Pr\{ D_1(1)=1 \mid S_0 \}, 
\end{align}

Next we delve into the two conditional counterfactual outcome expectations. 
By one sided noncompliance, we have $\E[Y(\0)\mid S_0] = \E[Y(D(\0))\mid S_0]$. The latter falls into the quantities that can be recursively identified by Lemma~\ref{lem:g-formula_many}, as:
\begin{equation}
\begin{aligned}
    \E[Y(\0)\mid S_0] =~& \E[Y(D(\0))\mid H_1] = f_1^{\0}(H_1, 0)
\end{aligned}
\end{equation}
with $f_t^{\0}(H_t, Z_t)$ as defined in Theorem~\ref{thm:2_identification_many}.
The quantity $\E[Y(\0)\mid D_1(1)=0, S_0]$ can be identified by another variant of the g-formula, as follows. We invoke the ignorability Assumption~\ref{assume:3_ignorability_many}, (i.e., that the following holds $\{Y(d), D_1(z_1)\} \cindep Z_1 \mid H_1$, which together with Lemma~\ref{lem:alternative-ignorability} implies $Y(d) \cindep Z_1 \mid D_1(z_1), H_1$). We can then write:
\begin{align*}
\E[ Y(\0) \mid S_0, D_1(1)=0 ] &= \E[ Y(\0) \mid H_1, D_1(1)=0 ]\\
&=
  \E[ Y(\0) \mid H_1, D_1(1)=0, Z_1=1 ]  \tag{$Y(d) \cindep Z_1 \mid D_1(z_1), H_1$}\\
&=
  \E[ Y(\0) \mid H_1, D_1=0, Z_1=1 ]     \\
&=
  \E[ Y(D_1, \0_{>1}) \mid S_0, D_1=0, Z_1=1 ]    \\
&=
  \E[ \E[ Y(D_1, \0_{>1}) \mid H_2] \mid H_1, D_1=0, Z_1=1 ]         \tag{$H_2\supseteq H_1\cup \{D_1, Z_1\}$}\\
&=
  \E[ \E[ Y(D_1, D_{>1}(\0_{>1})) \mid H_2] \mid H_1, D_1=0, Z_1=1 ]        \tag{one-sided}\\
&=
  \E[ \E[ Y(D_1, D_{>1}(\0_{>1})) \mid H_2, Z_2=0] \mid H_1, D_1=0, Z_1=1 ]        \tag{ignorability}\\
&=
  \E[ \E[ Y(D(Z_{\leq 2}, \0_{>2})) \mid H_2, Z_2=0] \mid H_1, D_1=0, Z_1=1 ] 
\end{align*}
Invoking Lemma~\ref{lem:g-formula_many} and, in particular, Equation~\eqref{eqn:middle-step-g-formula} from Lemma~\ref{lem:g-formula_many}, we observe that:
\begin{align*}
    \E[ Y(D(Z_{\leq 2}, \0_{>2})) \mid H_2, Z_2=0] = f_2^{\0}(H_2, 0)
\end{align*}
Thus we can identify the quantity $\E[Y(\0)\mid S_0, D_1(1)=0]$, recursively as:
\begin{align}
    \E[Y(\0)\mid S_0, D_1(1)=0] = \E[Y(\0)\mid H_1, D_1(1)=0] = \E[f_2^{\0}(H_2, 0)\mid H_1, D_1=0, Z_1=1]
\end{align}

Overall, we have derived that:
\begin{align*}
    \tau_{\1}(S_0) =~& \E[ \Delta \mid D_1(1)=1, S_0 ]\\
    =~& \E[f_2^{\1}(H_2, 1)\mid Z_1=1, D_1=1, H_1] \\
    &~ - \frac{f_1^{\0}(H_1, 0) - \E[f_2^{\0}(H_2, 0)\mid H_1, D_1=0, Z_1=1]\Pr(D_1=0\mid Z_1=1, H_1)}{\Pr(D_1=1\mid Z_1=1, H_1)}
\end{align*}
Note also that by the law of total probability the first term on the RHS can be written as:
\begin{align}
\frac{\E[f_2^{\1}(H_2, 1)\mid Z_1=1, H_1] - \E[f_2^{\1}(H_2, 1)\mid Z_1=1, D_1=0, H_1] \Pr(D_1=0\mid Z_1=1, H_1)}{\Pr(D_1=1\mid Z_1=1, H_1)}
\end{align}
Moreover, note that by definition $\E[f_2^{\1}(H_2, 1)\mid Z_1=1, H_1]=f_1^{\1}(H_1, 1)$.
Thus we can simplify $\tau_\1(S_0)$ as:
\begin{align*}
    \tau_{\1}(S_0)
    =~& \frac{f_1^{\1}(H_1,1) - f_1^{\0}(H_1, 0) - \E[f_2^{\1}(H_2, 1) - f_2^{\0}(H_2,0)\mid Z_1=1, D_1=0, H_1]\Pr(D_1=0\mid Z_1=1, H_1)}{\Pr(D_1=1\mid Z_1=1, H_1)}
\end{align*}
which concludes the proof of the identification formula given in the Theorem.

\subsubsection{Identification of Unconditional Always-Treat LATE}

We will reduce the identification to the conditional LATE case. Note that we can write:
\begin{align*}
    \tau_{\1} =~& \E\left[Y(\1) - Y(\0) \mid D(\1)=\1\right]\\
    =~& \E\left[\E\left[Y(\1) - Y(\0)\mid S_0, D(\1)=(\1)\right] \mid D(\1)=\1\right]\\
    =~& \E\left[\tau_{\1}(S_0) \mid D(\1)=\1\right]\\
    =~& \frac{\E\left[\tau_{\1}(S_0)\, \I(D(\1)=\1)\right]}{\Pr(D(\1)=(\1))}\\
    =~& \frac{\E\left[\tau_{\1}(S_0)\, \Pr(D(\1)=\1 \mid S_0)\right]}{\Pr(D(\1)=\1)}
\end{align*}
The quantity $\tau_{\1}(S_0)$ was identified in the prior section. Note that under Staggered Compliance, we have that $\Pr(D(\1)=\1)=\Pr(D_1(1)=1)$ and $\Pr(D(\1)=\1\mid S_0)=\Pr(D_1(1)=1\mid S_0)$. Thus we derive:
\begin{align}
    \tau_{\1} =~& \frac{\E\left[\tau_{\1}(S_0)\, \Pr(D_1(1)=1 \mid S_0)\right]}{\Pr(D_1(1)=1)}
\end{align}
The quantity $\Pr(D_1(1)=1\mid S_0)$, was identified as $\Pr(D=1\mid Z_1=1, S_0)$ in Equation~\eqref{eqn:d1-1}. Hence, by tower rule, we also have that $\Pr(D_1(1)=1)=\E[\Pr(D_1(1)=1\mid S_0)] = \E[\Pr(D=1\mid Z_1=1, S_0)]$.
We further note that the formula we derived for $\tau_\1(S_0)$ is a fraction with $\Pr(D_1=1\mid Z_1=1, S_0)$ in the denominator. Thus we can simplify the formula as:
\begin{align*}
    \tau_\1 = \frac{\E\left[f_1^{\1}(S_0,1) - f_1^{\0}(S_0, 0) - \E[f_2^{\1}(H_2, 1) - f_2^{\0}(H_2,0)\mid Z_1=1, D_1=0, S_0]\Pr(D_1=0\mid Z_1=1, S_0)\right]}{\E[\Pr(D_1=1\mid Z_1=1, S_0)]}
\end{align*}

\subsubsection{Proof of Lemma~\ref{lem:staggered-many}: Simplified Identification Formula}

The product term in the formula for $\tau_\1$, can also be re-written as:
\begin{align*}
     \E[(f_2^\1(H_2, 1)-f_2^\0(H_2,0))(1-D_1)\mid Z_1=1, H_1]
\end{align*}
Observing that:
\begin{align*}
    f_1^\1(H_1, 1) = \E[f_2^\1(H_2, 1)\mid Z_1=1, H_1]
\end{align*}
we can merge the two terms in the numerator as:
\begin{align*}
    \E\left[f_2^\1(H_2, 1)\, D_1 + f_2^\0(H_2, 0)\,(1-D_1)\mid Z_1=1, H_1\right] =  \E\left[f_2^{D_1\cdot \1}(H_2, D_1)\mid Z_1=1, H_1\right]
\end{align*}
Thus we can simplify the identification formula as
\begin{align*}
    \tau_\1 = \frac{\E\left[\E\left[f_2^{D_1\cdot \1}(H_2, D_1)\mid H_1, Z_1=1\right] - f_1^{\0}(H_1, 0)\right]}{\E[\Pr(D_1=1\mid Z_1=1, H_1)]}
\end{align*}
\end{proof}
\end{document}